\documentclass[sn-mathphys,Numbered]{sn-jnl}

\usepackage{graphicx}%
\usepackage{multirow}%
\usepackage{amsmath,amssymb,amsfonts}%
\usepackage{amsthm}%
\usepackage{mathrsfs}%
\usepackage[title]{appendix}%
\usepackage{xcolor}%
\usepackage{textcomp}%
\usepackage{manyfoot}%
\usepackage{booktabs}%
\usepackage{algorithm}%
\usepackage{algorithmicx}%
\usepackage{algpseudocode}%
\usepackage{listings}%

\usepackage{subcaption}
\usepackage{multicol}
\usepackage{multirow}
\usepackage{makecell}
\usepackage{enumerate}
\usepackage{float}
\usepackage{pifont}
\newcommand{\xmark}{\ding{55}}%
\usepackage{booktabs}\let\cline\cmidrule
\usepackage{hyperref}
\usepackage{ circuitikz }
\usetikzlibrary{decorations.pathreplacing}
\definecolor{myorange}{RGB}{225, 111, 35}
\definecolor{myyellow}{RGB}{255, 178, 46}
\definecolor{mygreen}{RGB}{51, 117, 15}
\definecolor{myblue}{RGB}{53, 104, 222}
\usepackage{adjustbox}

\theoremstyle{thmstyleone}%
\newtheorem{theorem}{Theorem}
\newtheorem{proposition}{Proposition}
\newtheorem*{proposition*}{Proposition *}
\newtheorem{lemma}[]{Lemma}
\newtheorem{corollary}[]{Corollary}

\theoremstyle{thmstyletwo}%
\newtheorem{example}{Example}%
\newtheorem{remark}{Remark}%

\theoremstyle{thmstylethree}%
\newtheorem{definition}{Definition}%
\newtheorem*{definition*}{Definition 1*}%

\begin{document}

\title[Article Title]{On the Graph Theory of Majority Illusions: Theoretical Results and Computational Experiments}


\author*[1]{\fnm{Maaike} \sur{Venema-Los}}\email{m.d.los@rug.nl}

\author[1]{\fnm{Zo\'{e}} \sur{Christoff}}\email{z.l.christoff@rug.nl}

\author[1,2]{\fnm{Davide} \sur{Grossi}}\email{d.grossi@rug.nl}




\affil[1]{\orgname{University of Groningen}, \orgaddress{\country{The Netherlands}}}
\affil[2]{\orgname{University of Amsterdam}, \orgaddress{\country{The Netherlands}}}


\abstract{The popularity of an opinion in one's direct circles is not necessarily a good indicator of its popularity in one's entire community. 
Network structures make local information about global properties of the group potentially
inaccurate,
and the way a social network is wired 
constrains 
what kind of information distortion can actually occur.
In this paper, we discuss which classes of networks allow for a large enough proportion of the population to get a wrong enough impression about the overall distribution of opinions. 
We start by focusing on the `majority illusion', the case where one sees a majority opinion in one's direct circles that differs from the global majority. We show that 
no network structure can guarantee that most agents see the correct majority.  
We then  perform computational experiments to study the likelihood of majority illusions in different classes of networks. Finally, we generalize to other types of illusions.}

\keywords{majority illusion, social networks, graph colorings.}



\maketitle

\section{Introduction}\label{sec1}

When forming opinions or making decisions, people often use as information 
the choices of others
in their circles. For example, if many people around you favor the same brand, vote for the same political party, or agree on the same opinion, you are more likely to buy, vote, or think the same \cite{lazarsfeld1954,McPherson2001, Lerman2016, Schelling1973, Granovetter1978, Salganik2006, Centola2010, Young2011, Centola2015}.
Therefore, the \emph{proportion} of people in one's circles
sharing an opinion plays a substantial role: the more people around you holding one opinion, the more likely it is that you will adopt that opinion, a mechanism  captured by well-known network diffusion models \cite{Degroot1974, Harary1959, Friedkin_Johnsen_2011} inspired by epidemics  models \cite{Bettencourt2006, Allen1994}. 
Beyond direct social influence, the popularity of an opinion in one's circles might also be taken as the main indicator of its global popularity: if you witness many people with the same opinion around you, you might get the impression that that opinion is also popular in the entire population.

Sometimes, however, the local information one gets from their neighborhood gives the wrong impression of the global state of the network. 
Consider, for example, the network in Figure \ref{fig:maj-ill-example}, where the nodes represent agents and red and blue colors represent the agents' binary 
opinions on a given issue.
\begin{figure}[t]
    \centering
    \includegraphics[width=0.3\textwidth]{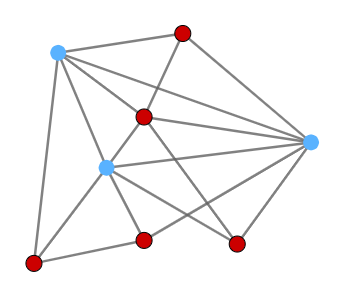}
    \caption{
    Example of a network where all red nodes are under majority illusion: they have more blue neighbors than red ones, while red is the actual global `winner'. 
    }
    \label{fig:maj-ill-example}
\end{figure}
While there are more \emph{red} nodes in the entire network, all red nodes have a majority of \emph{blue} nodes among their direct neighbors. In other words, the neighborhood of the red nodes gives them the wrong impression of what is the most popular opinion. Such a distortion 
is called \emph{majority illusion} \cite{Lerman2016}.
A node is under majority illusion whenever it locally sees a  dominant opinion that differs from the global dominant opinion.

Majority illusions can have important consequences on opinion dynamics. 
One example is in political decision making:
people might change their votes depending on their impression on who is the winning party. On the one hand, people are often more inclined to vote for a party they expect to win, which is known as the \emph{polling bandwagon effect} \cite{SchmittBeck2015Bandwagon}. This is typically explained either by a direct preference to belong to the winning team or by the social inference from a party being popular to there being good reasons for its popularity. 
On the other hand, in some cases, people are more inclined to vote for a party they expect \emph{not} to win, because then their vote is more likely to make a difference, which is known as the \emph{underdog effect} \cite{Auletta2023}.
In both cases, a party can gain votes simply by giving people the \emph{impression} that they are winning/losing. 
Therefore, in principle, one could influence opinions and votes by changing the network structure to tweak the distribution of opinions that agents see locally. 
Such manipulations,  referred to as  `information gerrymandering', can, according to \cite{Bergstrom2019, stewart2019-long}, lead to undemocratic decisions.

The classic \emph{friendship paradox} \cite{Feld1991} is a related example of a situation  
where the neighborhood gives a wrong impression of the global state of the population: 
agents in a network are likely to get the impression that their popularity is lower than average because their friends have more friends than themselves.
A related case study is that of the illusion of popularity of heavy alcohol consumption on university campuses. While this illusion is often explained by pluralistic ignorance \cite{Baer1991, Prentice1993}, it is also explained by differences in students' popularity in \cite{Lerman2016}: If heavy drinkers have more friends, this may cause many students to overestimate their overall friends' alcohol consumption. Indeed, they may observe that most of their friends drink a lot even though in reality the global proportion of heavy drinkers may be small. This, in turn, can make norms of excessive alcohol use that are actually unpopular last for generations of students. 

\medskip

In summary, in political decision making as well as in everyday life, distortions occurring between the distribution of opinions in one's direct circles and their global distribution may have important societal impact.
However, such 
distortions cannot always occur. Whether they are possible or not, and under which conditions, is the question we focus on in this paper. 
We will address this question by solely looking at the structural properties of networks and study
which kind of illusions are possible in which 
classes of graphs.

\paragraph{
Related work}
The concept of majority illusion was first studied in \cite{Lerman2016}, where the authors used computational simulations to find 
how many nodes are typically under majority illusion
in scale-free, Erd\H{o}s-R\'enyi networks \cite{ErdosRenyi}, and in several real-world networks, and show that, unsurprisingly,
nodes in
networks in which high-degree nodes tend to connect to low-degree nodes are most susceptible to this illusion. 
Majority illusions were studied next in \cite{stewart2019-long}, where a voter game is modelled with two competing parties to show that majority illusions can be used for 
\emph{information gerrymandering}: modifying the network structure for the purpose of gaining votes in an election, reminiscent of district gerrymandering \cite{mw:gerrymandering, Schuck1987}. 
The authors predict this by a mathematical 
simulation of the voter game, and confirm their results with a social network experiment with human participants. They find that information gerrymandering can even take place when all agents have the same 
number of connections.
A more recent article \cite{Grandi2022} 
studies the computational complexity of majority illusions. A $q$-majority illusion is defined there as at least a $q$ fraction of agents being under majority illusion, and it is shown that the problem of deciding whether 
a given network can be colored in such a way that it 
is in 
a $q$-majority illusion is NP-complete for $q>\frac{1}{2}$. Whether it also holds for $q\leq \frac{1}{2}$ was left as an open question. Since majority illusions may have detrimental effects and are generally regarded as undesirable, the authors of \cite{Grandi2022} also study how to remedy an illusion from a network by adding or deleting edges. The problem of identifying whether it is possible to change the network in such a way that the number of agents with a $q$-majority illusion is below a given bound is also shown to be NP-complete for $q>\frac{1}{2}$.

\paragraph{
Contribution and outline}
The existing literature on majority illusions does not tell us much about how far social networks can distort local information. On the one hand, in  \cite{Lerman2016}, a relation was found between the network structure and the \emph{probability} of majority illusions. On the other hand,  \cite{Grandi2022} shows that to decide whether, for a given network, it is possible for a certain proportion of people to be under majority illusion is a hard computational problem. The question that stems from both these lines of work is then whether there exist salient properties of the network structure that determine whether or not it is \emph{possible} for a 
certain proportion of nodes to be under majority illusion. In other words, what graphs make such illusions possible or impossible? 
We study this question in the {\em first part}
of the paper.
We consider 
different classes of graphs, for different reasons. 
For each of the classes, we 
investigate
whether it is possible for a \emph{majority of nodes} to be under \emph{majority illusion}.
The reason to focus specifically on when  
a \emph{majority} 
of nodes are under majority illusion is twofold. First, if most people are under illusion, then collectively the group is mostly wrong in the sense that any reasonable way to aggregate their guesses would result in an incorrect collective guess. Second, 
in models where people influence each other by majority dynamics (a basic model of opinion diffusion), if $\frac{1}{2}$ of the nodes is under majority illusion, the group's majority opinion will be inverted. For these reasons, it is a natural focus (see also  \cite{Grandi2022}).
 In Section \ref{sec:arbitrary-networks}, we prove that a weak version of the majority illusion can occur on \emph{all} network structures. 
 In Section \ref{sec:specific-networks}, we provide some stronger results for some specific classes of networks
 .
Given the complexity of the question, we first explore some classes of graphs that are simple enough to be analytically tractable, even though they are not likely to occur in real social networks.
 We find that on 2-colorable graphs with an odd number of nodes or with a node of which all neighbors have degree greater than 2, it is possible that a (weak) majority of the nodes is under majority illusion (Propositions \ref{prop:2-col-maj-maj} and \ref{prop:2col-strict}).
Next, we consider regular graphs, for which the existing literature \cite{Centola2022} suggests 
 that illusions are less likely to occur, which makes sense intuitively because agents have less centralized information. 
 However, we find that stricter majority illusions 
 can still occur on some but not all regular graphs.  
  On the one hand, we find that with certain given relative size and degree it is possible to construct a regular network in which a majority of the nodes is under majority illusion (Theorem \ref{thm:regular-maj-maj}). On the other hand, stricter illusions cannot occur when the degree of the nodes is 2 (a collection of cycles) (Proposition \ref{prop:k=2}), or $n-1$ (a complete graph) (Proposition \ref{prop:complete-weak-maj-maj}). Our main analytical results are summarized in Table \ref{tab:results} at the end of the paper.
In Section \ref{sec:simulations}, we use computational simulations to study the likelihood and strength of majority illusions
on graphs that have properties that are characteristic of real-world social networks
.
The structure of the first part of the paper is illustrated in Figure \ref{fig:overview_tree}.

\medskip

Although majority illusions are the 
only form 
discussed in the literature, other forms of illusion are worth studying. After all, the majority opinion around us is not the only thing that can influence us. We could say that we are influenced by a dominant opinion in our surroundings, but 
`dominant' can be defined in different ways%
, of which `majority' is only one. Therefore,
in the {\em second part} of this article (Section~\ref{sec:generalizations}) we generalize the notion of majority illusion to cover illusions induced by different ways of aggregating neighbors' opinions, namely by quota-rules and the plurality rule.
We establish several results on which networks allow for which kinds of illusions.

\medskip

The paper is an improved and extended version of \cite{Venema-Los2023}, which appeared in the proceedings of EUMAS'23.

\begin{figure}[t]
\centering
\resizebox{1\textwidth}{!}{%
\begin{circuitikz}
\tikzstyle{every node}=[font=\LARGE]

\draw [ color={myblue} , fill={myblue}] (19.5,13.5) rectangle node {\Large \textcolor{white}{all grahps}} (22.75,11.75);
\draw [ color={myblue} , fill={white}] (20,13.5) ellipse (1.5cm and 0.5cm)  node {\large{Section \ref{sec:arbitrary-networks}}};
\draw [ color={myblue} , fill={white}] (23.5,11.25) rectangle node {\large \textcolor{black}{Theorem \ref{thm:maj-weak-maj}}} (20.5,12.25);

\draw [short] (21,11.25) -- (21,10.5);
\draw [short] (7,7.75) -- (7,6);
\draw [short] (17.25,7.75) -- (17.25,6);
\draw [short] (26.25,7.75) -- (26.25,6);

\draw [ color={myorange} , fill={myorange}] (5.75,9.5) rectangle node {\Large \textcolor{white}{limit cases}} (9,7.75);
\draw [ color={myorange}, fill = {white} ] (6,9.5) ellipse (1.5cm and 0.5cm)  node {\large{Section \ref{sec:specific-networks}}};

\draw [ color={myorange} , fill={myorange}] (3,5.5) rectangle  node {\Large \textcolor{white}{odd degree}}(6.25,3.75);
\draw [ color={myorange} , fill={white}] (7,3.25) rectangle node {\large \textcolor{black}{Proposition \ref{thm:degree-odd}}} (4.5,4.25);

\draw [ color={myorange} , fill={myorange}] (7.25,5.5) rectangle  node {\Large \textcolor{white}{2-colorable}}(10.5,3.75);
\draw [ color={myorange} , fill={white}] (11,3.25) rectangle  node {\large \textcolor{black}{Lemma \ref{lem:2-col}}}(8.5,4.25);

\draw [short] (7.75,10.5) -- (33,10.5);
\draw [short] (7.75,10.5) -- (7.75,9.5);
\draw [short] (5,5.5) -- (5,6);
\draw [short] (5,6) -- (9,6);
\draw [short] (9,6) -- (9,5.5);

\draw [ color={myyellow} , fill={myyellow}] (15.75,9.5) rectangle  node {\Large \textcolor{white}{regular}}(19,7.75);
\draw [ color={myyellow} , fill={white}] (16.25,9.5) ellipse (1.5cm and 0.5cm) node {\large{Section \ref{sec:specific-networks}}};
\draw [ color={myyellow} , fill={myyellow}] (13.25,5.5) rectangle  node {\Large \textcolor{white}{2-regular}}(16.5,3.75);
\draw [ color={myyellow} , fill={white}] (17.25,3.25) rectangle node {\large \textcolor{black}{Proposition \ref{prop:k=2}}} (14.75,4.25);
\draw [ color={myyellow} , fill={myyellow}] (17.5,5.5) rectangle  node {\Large \textcolor{white}{complete}} (20.75,3.75);
\draw [ color={myyellow} , fill={white}] (21.25,3.25) rectangle  node {\large \textcolor{black}{Prop. \ref{prop:complete-weak-maj-maj} and \ref{prop:complete-maj-weak-maj}}}(18.75,4.25);

\draw [short] (15.25,5.5) -- (15.25,6);
\draw [short] (15.25,6) -- (19.25,6);
\draw [short] (19.25,6) -- (19.25,5.5);

\draw [short] (18,10.5) -- (18,9.5);
\draw [ color={mygreen} , fill={mygreen}] (24.75,9.5) rectangle node {\Large \textcolor{white}{random}}(28,7.75);
\draw [ color={mygreen} , fill={white}] (25.25,9.5) ellipse (1.5cm and 0.5cm)node {\large{Section \ref{sec:simulations}}};
\draw [short] (27,10.5) -- (27,9.5);
\draw [ color={mygreen} , fill={mygreen}] (22.5,5.5) rectangle  node {\Large \textcolor{white}{Erd\H{o}s-R\'{e}nyi}}(25.75,3.75);
\draw [ color={mygreen} , fill={mygreen}] (26.75,5.5) rectangle  node {\Large \textcolor{white}{Holme-Kim}}(30,3.75);

\draw [short] (24.5,5.5) -- (24.5,6);
\draw [short] (24.5,6) -- (28.5,6);
\draw [short] (28.5,6) -- (28.5,5.5);

\draw [ color={mygreen} , fill={mygreen}] (30.75,9.5) rectangle  node {\Large \textcolor{white}{real world}} (34,7.75);
\draw [ color={mygreen} , fill={white}] (31.25,9.5) ellipse  (1.5cm and 0.5cm) node {\large{Section \ref{sec:simulations}}};
\draw [short] (33,10.5) -- (33,9.5);
\draw [ color={mygreen} , fill={mygreen}] (31.75,5.5) rectangle  node {\Large \textcolor{white}{Facebook}}(35,3.75);
\draw [short] (33,7.75) -- (33,5.5);
\draw [decorate,decoration={brace,amplitude=10,mirror,raise=4ex}]
  (3,3) -- (21,3) node[midway,yshift=-5em]{Theoretical analysis};
\draw [decorate,decoration={brace,amplitude=10,mirror,raise=4ex}]
  (22.5,3) -- (35,3) node[midway,yshift=-5em]{Experimental analysis};
\end{circuitikz}
}%
\caption{Overview of the first part of the paper and its main results.}
\label{fig:overview_tree}
\end{figure}

\section{Preliminaries} \label{sec:prelim}

 A social network (a simple graph) $G= \langle V,E\rangle$  consists of a finite set $V$ of
 agents (nodes/vertices), and a set $E$ of undirected and irreflexive
 edges between agents. 
We do not assume that the graph is connected. 
We call `neighbors' two agents connected by an edge. 
We write $N_i$ for the set of neighbors of $i$ and $d_i$ for its degree $|N_i|$.
We consider a single issue and assume that each agent holds an 
opinion on this issue. 
Since 
opinion distributions can be seen as graph 
colorings, we borrow the terminology of vertex colorings, and use the terms `color' and `opinion' interchangeably. 
We write $\mathcal{C}$ for the set of possible colors, $c_i\in \mathcal{C}$ to refer to agent $i$'s color and $c$ for the coloring of the graph ($c: V\to \mathcal{C}$), that is, the distribution of opinions. A colored graph is a triple $C=\langle V,E,c\rangle$. Thoughout the paper, the term `colored graph' refers to such colored graphs with only two colors ($\mathcal{C}=\{red, blue\}$), except 
for 
Section \ref{sec:generalizations}, where we will  generalize to more than two colors (opinions) being represented. 
Note that, since most of our results in Sections \ref{sec:arbitrary-networks} and \ref{sec:specific-networks} concern the \emph{existence} of a specific type of 2-coloring of the graph, they do immediately hold for $k$-colorings for any $k\geq2$.

In such opinion networks (or colored graphs), 
three types of information can be distinguished. First, every node has an individual opinion. Second, every node has a majority opinion in its neighborhood, the \emph{local majority opinion}. And third, there is the \emph{global} majority opinion, the majority opinion in the entire network. 
Any two of these three types of opinions can be in agreement or not. 
We systematize and illustrate all possible relations between the above types of opinions 
in Table \ref{tab:options}. 
Different fields have been studying disagreement between the different types of information mentioned above.
On the one hand, in social network science and social choice theory,
an agent is under \emph{majority illusion} when the majority of its neighbors disagrees with the global majority \cite{Grandi2022,Lerman2016}. 
In contrast, graph theory has concerned itself with the disagreement between a node's color and the color of its neighbors: a proper coloring requires that no two adjacent nodes are of the same color, that is, that everybody disagrees with all of their neighbors.
A generalisation of that concept is that of majority coloring \cite{Anholcer2020,Bosek2019,Kreutzer2017}, where no agent is of the same color as most of its neighbors.
We call the local disagreement faced by an agent in a majority coloring 
\textit{majority opposition}.
In such a situation, one might get the impression that they belong to a global minority. 
For instance, in Figure \ref{fig:maj-ill-example} 
all nodes are under majority opposition so 
all nodes might have the impression of belonging to a global minority,
while it is only true for the blue ones.
When \emph{all} agents are under the impression that they belong to the minority, then some of them must be mistaken,  
it must be some sort of illusion. 
 Clearly, the concepts of majority illusion and majority opposition are related. In the first part of this paper (Section \ref{sec:arbitrary-networks}),
we explore this very relation 
to obtain results about majority illusions.

\begin{table}[htp]
    \centering
        \caption{Possible combinations of local and global majority winners, and presence or absence of majority opposition and majority illusion. We assume w.l.o.g. that the color of the relevant individual (highlighted in the exemplary illustrations) is red, otherwise just swap `red' and `blue' everywhere. \xmark~ indicates absence of the opposition/illusion, $\checkmark$ indicates presence of the opposition/illusion, `weak' indicates the presence of a weak-majority opposition or a weak-majority illusion. References are included to situate the related fields of research.
        }
        \vspace{0.3cm}
    \begin{tabular}{|c|c|c|c|c|}
    \hline
         & \textbf{local majority} & majority opposition & \textbf{global majority} & majority illusion  \\ \hline
        \includegraphics[width = 25pt]{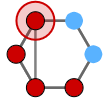}& \textcolor{red}{red} & \xmark   & \textcolor{red}{red} & \xmark\\ \hline
         \includegraphics[width = 25pt]{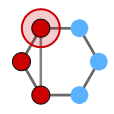} & \textcolor{red}{red} & \xmark  & tie & weak \\ \hline
         \includegraphics[width = 25pt]{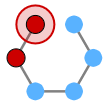} & \textcolor{red}{red} & \xmark  & \textcolor{blue}{blue} & $\checkmark$  \cite{Lerman2016,Grandi2022} \\ \hline
         \includegraphics[width = 25pt]{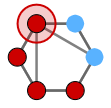}  & tie  & weak \cite{Anholcer2020,Bosek2019,Kreutzer2017} & \textcolor{red}{red} &  weak \\ \hline
         \includegraphics[width = 25pt]{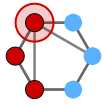} & tie & weak  \cite{Anholcer2020,Bosek2019,Kreutzer2017}  & tie &  \xmark \\ \hline
         \includegraphics[width = 25pt]{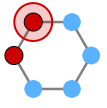}  & tie& weak  \cite{Anholcer2020,Bosek2019,Kreutzer2017}  & \textcolor{blue}{blue} &   weak\\ \hline
          \includegraphics[width = 25pt]{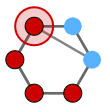}  & \textcolor{blue}{blue} &  $\checkmark$ & \textcolor{red}{red} &    $\checkmark$ \cite{Lerman2016,Grandi2022}   \\ \hline
         \includegraphics[width = 25pt]{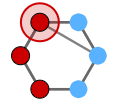}  & \textcolor{blue}{blue} & $\checkmark$ & tie &    weak  \\ \hline
         \includegraphics[width = 25pt]{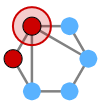} & \textcolor{blue}{blue} & $\checkmark$  & \textcolor{blue}{blue} &    \xmark  \\ \hline
    \end{tabular}
    \label{tab:options}
\end{table}

We will now define 
the notions described informally above.
We start by introducing some auxiliary terminology to be able to talk about which 
opinion
is prevalent in a network, be it locally or globally. 
Given a set $S$ of agents such that $|S|=n$ and 
a coloring $c$, a color $x$ is a \textit{majority winner} of $S$ (we write $M_S=x$) if $|\{i\in S\mid c_i=x\}|
\geq\frac{n}{2}$. When 
two colors are a majority winner among $S$ (there is a tie), we will write $M_S=tie$.
We say that an agent is under majority \emph{illusion} if both the agent's neighborhood and the entire network have 
only one majority winner (no tie) but they are different. This definition is equivalent to that in \cite{Grandi2022}.

\begin{definition}[majority illusion]\label{def:maj-ill-strict}
Given 
a colored graph $C = \langle V, E,c\rangle $, an agent $i\in V$ is under \emph{majority illusion}   (\texttt{m} illusion)
if $M_{N_i}\not = \text{tie}$  and $M_{V}\not = \text{tie}$ and $M_{N_i} \not = M_{V}$.  A graph is in a  \emph{Majority-majority illusion}  (\texttt{Mm} illusion) if more than half of the agents are under majority illusion. 
\end{definition}
For instance, the situation depicted by Figure \ref{fig:maj-ill-example} is an example of Majority-majority illusion. 
As in Definition \ref{def:maj-ill-strict}, we will use throughout the paper a capital letter to distinguish the global majority from the local majority, both in the full names and in the abbreviations of the illusions. 

We can generalize this strict definition to weaker cases. First, there exists a weaker type of disagreement between local and global majorities: the cases where exactly one of the two is a tie. 
Second, the majority of agents under an illusion can also be weak,  
when exactly half of the agents are under illusion.
The corresponding generalisations of majority illusion includes both types of weakening: 

\begin{definition}[weak versions of majority illusion]\label{def:maj-ill-weak}
Given 
a colored graph $C = \langle V, E,c\rangle $, agent $i\in V$ is under \emph{weak-majority illusion}  (\texttt{w} illusion) if $M_{N_i} \not = M_V$.
A graph is in a \emph{Majority-weak-majority illusion}  (\texttt{Mw} illusion) if more than half of the agents are under \texttt{w} illusion. A graph 
is in a  \emph{Weak-Majority-(weak-)majority illusion}  (\texttt{Wm / Ww} illusion) if at least half of the agents are under \texttt{m/w} illusion.
\end{definition}
To prevent confusion between the two different types of majorities involved in the illusions we study, 
we give in Table \ref{tab:naming_overview} an overview of the names we will use, with the corresponding abbreviations and definitions.
\begin{table}[t]
    \centering
    \caption{Overview of the various types of majority illusion with their respective abbreviations.}
    \label{tab:naming_overview}
    \begin{tabular}{|l|l|l|}
    \hline
        \textbf{name} & \textbf{abbr.}& \textbf{meaning} \\ \hline
        \textcolor{teal}{majority} illusion 
         &\textcolor{teal}{\texttt{m}} illusion& The local majority is different from the global majority \\
         &  & and they are both not a tie. (Definition \ref{def:maj-ill-strict})\\ \hline
       \textcolor{teal}{weak-majority} illusion 
        &\textcolor{teal}{\texttt{w}} illusion& The local majority is different from the global majority \\
         & & but one of them may be a tie. (Definition \ref{def:maj-ill-weak})\\ \hline
       \textcolor{orange}{Majority}-\textcolor{teal}{majority} illusion  &\textcolor{orange}{\texttt{M}}\textcolor{teal}{\texttt{m}} illusion&  \textit{more than }\textcolor{orange}{half of the nodes} are under \textcolor{teal}{majority} i. (Def. \ref{def:maj-ill-strict}) \\ \hline
       \textcolor{orange}{Majority}-\textcolor{teal}{weak-majority} i.  &\textcolor{orange}{\texttt{M}}\textcolor{teal}{\texttt{w}} illusion&  \textit{more than }\textcolor{orange}{half of the nodes} are under \textcolor{teal}{weak-majority} i. (Def. \ref{def:maj-ill-weak}) \\ \hline
       \textcolor{orange}{Weak-Majority}-\textcolor{teal}{majority} i.  &\textcolor{orange}{\texttt{W}}\textcolor{teal}{\texttt{m}} illusion&  \textit{at least }\textcolor{orange}{half of the nodes} are under \textcolor{teal}{majority} illusion. (Def. \ref{def:maj-ill-weak}) \\ \hline
       \textcolor{orange}{Weak-Majority}-\textcolor{teal}{weak-maj.} i.&\textcolor{orange}{\texttt{W}}\textcolor{teal}{\texttt{w}} illusion&  \textit{at least }\textcolor{orange}{half of the nodes} are under \textcolor{teal}{weak-majority} i. (Def. \ref{def:maj-ill-weak}) \\ \hline
    \end{tabular}
\end{table}

\medskip

Similarly, the notion of majority opposition comes in  strict and weak versions:
\begin{definition}[(weak) majority opposition]\label{def:maj-opp}
    Given a colored graph $C = \langle V, E,c\rangle $, agent $i\in V$ is under \emph{majority opposition} if $c_i \not = M_{N_i}$ and $M_{N_i}\not = tie$. Agent $i\in V$ is under \emph{weak majority opposition} if $c_i \not = M_{N_i}$.
\end{definition}

As indicated in Table~\ref{tab:options}, while it is the strict version of the majority illusion
that is studied by \cite{Lerman2016}\footnote{They use the strict version throughout the paper, except for Figure 1 where the illusion can be weak.}, and by \cite{Grandi2022}, it is the weak version of majority opposition
that is studied by \cite{Anholcer2020,Bosek2019,Kreutzer2017} in graph theory. 
As far as we know, neither the strict majority opposition nor the weak majority illusions have  been studied before.
Furthermore, note that, in the same network, different agents can be under a \texttt{w} illusion with respect to different opinions, since it is possible that exactly half of the nodes in the network are of one color and half of the nodes of 
another color. 

\medskip

Before proceeding, we introduce some extra terminology.
When an agent 
is under \texttt{m} illusion and all its 
neighbors all have the same color, we say that that agent is under \emph{unanimity illusion} (\texttt{u} illusion).
 Similarly, when all agents are under a \texttt{m/w} illusion, we will call it a \emph{Unanimous-(weak-)majority
 illusion} (\texttt{Um / Uw} illusion).
We say that \emph{an illusion is possible} in a graph $G=\langle V, E \rangle$ if there exists a coloring $c$ such that $C=\langle V, E, c\rangle$ witnesses the respective illusion.


\begin{remark}[(Ir)reflexivity]
In all existing literature about majority illusions so far \cite{Lerman2016, Grandi2022}, as well as in \cite{stewart2019-long} which is about information gerrymandering, 
social networks are taken to be  \emph{irreflexive}. Although one might argue that agents  know their own opinion, the reasoning here is mainly about how agents are influenced by others rather than by themselves. In influence models,  
there are various ways to deal with an agent's influence on themselves. In some models, agents are simply counted as  one of their neighbors \cite{Zollman2014, Degroot1974}. In others, they are not taken into account for their own influence at all \cite[Chapter 9]{Granovetter1978, Jackson2010}. And in yet other models, the network is defined as irreflexive but agents are still assumed to have a special type of influence on themselves \cite{Zhen2011, Seligman2014}. Since it is hard to justify one of these choices without knowing exactly the application of the model (and even then), we align here with the choice of the most relevant existing literature  \cite{stewart2019-long, Lerman2016, Grandi2022}, and consider networks that are irreflexive. 
Nevertheless, in Section~\ref{sec:reflexive}, we consider generalizations of majority illusions to networks that are not necessarily irreflexive.
\end{remark}

\section{Illusions in Arbitrary Networks}\label{sec:arbitrary-networks}
Our overall goal is to discover which social networks allow for 
illusions to occur. 
Since this is equivalent to asking which graphs can be colored in some  specific way,  
we build on existing results from vertex colorings to obtain results about majority 
illusions.
Recall that, in graph theory, a 
coloring is called \emph{proper} when no two neighbors are assigned the same color. 
The weaker notion of majority coloring \cite{Anholcer2020,Bosek2019,Kreutzer2017} 
is immediately relevant to us. In a majority coloring, \emph{all} vertices are in what we described as 
weak 
majority opposition: at least half of its neighbors are of a different color than its own. 
For coherence with the rest of the paper, we call this a \emph{weak} majority coloring here:
\begin{definition}[weak majority $2$-coloring]\label{def:weak_maj_2-col}
A 
weak
majority $2$-coloring of a graph $G=\langle V,E\rangle$ is a $2$-coloring $c$ 
such that, for each $i\in V:M_{N_i}\neq c(i)$. 
\end{definition}

A graph is called weak majority
$2$-colorable if there exists a weak majority 
$2$-coloring of it. 
Given a colored graph, we call \emph{monochromatic} the edges between nodes of the same color, and \emph{dichromatic} the ones between nodes of different colors.

The main result involving majority colorings is credited to  \cite{Lovász1966}  in the literature \cite{Anholcer2020,Bosek2019}: every graph is weak majority $2$-colorable\footnote{In a parallel line of research in combinatorics, majority colorings are called \emph{unfriendly partitions} \cite{Aharoni1990}, and shown to exist on all finite graphs \cite{Aharoni1990, Bernardi1987} but not on all infinite graphs \cite{Shelah1990}.}. 
The proof strategy for this result is commonly described as easy and relying on a simple `color swapping mechanism' that can only reduce the total number of monochromatic edges in the network.
However,  \cite{Lovász1966} itself focuses on multigraphs and is of a much wider scope. Therefore, to make the paper self-contained, 
we provide both a proof of the general result in Appendix \ref{app:proof-lovasz} and below the proof of a related lemma, Lemma \ref{lem:min-monochromatic}, which will be crucial to our main result, Theorem \ref{thm:maj-weak-maj}.

\begin{lemma}\label{lem:min-monochromatic}
Let $G=\langle V, E\rangle$ be a graph, and let $c$ be a $2$-coloring of $G$ that minimizes the number of monochromatic edges. Then, $c$ is a weak majority $2$-coloring of $G$. \end{lemma}
\begin{proof}
    Let $E_{M}$ be the set of monochromatic edges and $E_{D} = E\backslash E_{M}$  the set of dichromatic edges in graph $G$ colored by $c$. 
    Assume for contradiction that there is a node $i\in V$ that is an endpoint of strictly more monochromatic edges (we write $E_{M_i}$ for the set of such edges) than dichromatic edges ($E_{D_i}$): $|E_{M_i}|>|E_{D_i}|$. 
Consider now a second $2$-coloring $c'$ of $G$ that only differs from $c$ with respect to $i$'s color, i.e., $c'$ assigns the same color as $c$ did to all nodes except for $i$: $c_i\neq c'_i$. Let us write $E'_M$ for the new set of monochromatic edges, and $E'_{M_i}$ and $E'_{D_i}$ for the new sets of monochromatic and dichromatic edges from $i$. Given that $|E'_{M_i}|=|E_{D_i}|$ and
$|E'_{D_i}|=|E_{M_i}|$, we now have $|E'_{D_i}|>|E'_{M_i}|$ and $|E_{M_i}|>|E'_{M_i}|$. Given that no other edge of the graph is affected by this change, the total number of monochromatic edges is smaller with coloring $c'$ than it was with $c$: $|E_M|>|E'_M|$. But since we started by assuming that $c$ was such that $|E_M|$ was minimal, this is a contradiction. 
\end{proof}

We now use the existence of a majority coloring that minimizes the number of monochromatic edges to prove the following general result:
\begin{theorem}\label{thm:maj-weak-maj}
A \texttt{Mw} illusion is possible in any graph $G=\langle V,E \rangle$.
\end{theorem}
\begin{proof}
  Let $G=\langle V,E \rangle$ be a graph and let $c$ be a $2$-coloring of $G$ 
  that minimizes the total number of monochromatic edges. By Lemma \ref{lem:min-monochromatic}, $c$ is a weak majority $2$-coloring of $G$. 
    There are two cases: 
    \begin{itemize}
        \item $M_V\neq tie$. 
        Assume w.l.o.g. that  $M_V=red$. Since $c$ is a weak majority coloring, for any red vertex $i$, $M_{N_i}\in\{blue, tie\}$, and therefore $M_{N_i}\neq M_V$. Hence,  a majority of the nodes (the red ones) is under (possibly weak)  majority illusion: we have a \texttt{Mw} illusion. 
        \item $M_V = tie$. There are two cases:
        \begin{itemize}
        \item 
        If $|\{i\in V: M_{N_i}\in\{blue, red\}\}| > \frac{|V|}{2}$, we have a \texttt{Mw} illusion.
        \item
        Otherwise 
        (if $|\{i\in V: M_{N_i}=tie\}| \geq \frac{|V|}{2}$)
        choose a node $j$ with $M_{N_j}= tie$ and define a new coloring $c'$ that is equal to $c$ for all nodes except for $j$:  $c'_j\neq c_j$. Since $j$ has as many blue as red neighbors, this does not change the total number of monochromatic edges in the graph. Therefore,  $c'$ is also a coloring that minimizes this number. Hence, by Lemma \ref{lem:min-monochromatic}, $c'$ is also a weak majority $2$-coloring of $G$. Now, we have $M_V = c'_j$, and we can apply the logic of the first case: Assume w.l.o.g. that  $c'_j=red$. Since $c'$ is a weak majority coloring, for any red vertex $i$, $M_{N_i}\in\{blue, tie\}$. It follows that a majority of the nodes has $M_{N_i}\neq M_V$:
        we have a \texttt{Mw} illusion. 
        \end{itemize}
    \end{itemize} 
\end{proof}

One of the results in \cite{Grandi2022} is that checking whether or not a network allows for a \emph{Majority-majority} illusion is NP-complete\footnote{\cite{Grandi2022} does not actually speak about a \texttt{Mm} illusion, but about `at least a fraction of $q$ agents is under \texttt{m} illusion', with $q>\frac{1}{2}$. The fact that it then also holds for \texttt{Mm} illusion follows from that `more than half' is equivalent to `at least some fraction $q$ where $q$ is more than half' since we only have to consider rational numbers.}. 
Here, in stark contrast, we see  that there is no need for checking whether 
a network allows for a Majority-\emph{weak-}majority illusion, 
since Theorem \ref{thm:maj-weak-maj} shows that it is always the case.

\section{Illusions in Specific Network Classes}\label{sec:specific-networks}
 While the above  resolves the question of the existence of \emph{weak} majority illusions, we now aim to understand when the \emph{strict} version of the illusion can occur. In order to obtain results in that direction, we
 turn to some 
classes of graphs with well-known global properties.  
In some 
networks, namely two-colorable graphs and graphs where all nodes have odd degrees, it is easy to see whether or not illusions are possible, although these networks are unlikely to resemble real social networks. Nevertheless, we still mention those results, since they can be used as starting point for 
the systematic analysis of the types of graphs that allow for majority illusions. 
We then turn to consider the class of regular graphs, which, from the literature, seem to be promising for the prevention of illusions.

\subsection{Graphs with Odd Degrees}
On graphs in which all nodes have an odd degree, any \texttt{Mw} illusion, which are  guaranteed to exist by  Theorem \ref{thm:maj-weak-maj}, is always a stronger type of illusion. 
The intuition is that an agent with odd degree cannot see a tie in its neighborhood, which causes either \emph{all} agents to be in \texttt{w} illusion if there is a global tie, or, if there is no global tie, a majority of agents to be in a \texttt{m} illusion.
\begin{proposition}\label{thm:degree-odd}
For any graph $G=\langle V,E \rangle$ such that for all $i\in V,  d_i$ is odd, a \texttt{Mw} illusion in $G$ is either a \texttt{Uw} illusion or a \texttt{Mm} illusion.
\end{proposition}

\begin{proof}
 Let $G=\langle V,E \rangle$ be such that for all $i\in V$, $d_i$ is odd. By Theorem \ref{thm:maj-weak-maj}, there exists a coloring of $G$ that is a \texttt{Mw} illusion. Consider any such coloring $c$. 
Two cases:
\begin{itemize}
\item $M_V=tie$. For all $i\in V$, since $d_i$ is odd, $M_{N_i}\neq tie$ and therefore $M_{N_i}\neq M_V$: we have \texttt{Uw} illusion.
\item $M_V\neq tie$. Assume w.l.o.g. that $M_V = red$. Since $G$ is in a \texttt{Mw} illusion, 
$|\{i\in V: M_{N_i}\in\{blue, tie\}\}| > \frac{|V|}{2}$, but since for all $i$, $d_i$ is odd, this implies that all those vertices cannot have a tie: we have 
$|\{i\in V: M_{N_i}=blue\}| > \frac{|V|}{2}$,
a \texttt{Mm} illusion.
\end{itemize}

\end{proof}

Given a graph coloring we can define a `swappable node' as a node whose neighbors all have at least 2 (so 3 for odd degree) more nodes of one color than nodes of the other color. Then, a corollary of 
Proposition \ref{thm:degree-odd} is the following:

\begin{corollary}\label{cor:swappable_node}
    Let a graph $G=\langle V,E \rangle$ be such that for all $i\in V,  d_i$ is odd. 
    Let $c$ be a weak majority 2-coloring of $G$ that is a \texttt{Mw} illusion\footnote{Such coloring exists by the proof of Theorem \ref{thm:maj-weak-maj}.}.
    If $M_V=tie$ and there is at least one $j\in V$ that is `swappable', then a \texttt{Wm} illusion is possible in $G$.
\end{corollary}
\begin{proof}
    W.l.o.g. assume $c_j =  red$ and define $c'$ which is equal to $c$ for all nodes except that $c'_j=blue$. Since $c$ was a weak majority 2-coloring, all red nodes had more than half blue neighbors. Since $j$'s neighbors all had a margin of at least 2 and nothing except $j$'s color changed, all red nodes except $j$ still have more than half blue neighbors in $c'$. Hence, half of the nodes are under \texttt{m} illusion.
\end{proof}

 \subsection{2-Colorable Graphs}
  In the same way as we used the existence of a majority coloring to obtain results about the existence of majority illusions 
 we can also use the existence of a special type of weak majority 
 colorings, the proper
 colorings, to obtain results about majority illusions in $2$-colorable graphs.

\begin{lemma}\label{lem:2-col}
    Any proper 2-coloring of a graph $G=\langle V,E \rangle$ is either
    a \texttt{Mm} illusion or a \texttt{Uw} illusion. 
\end{lemma}
The idea of the proof is similar to that of 
Proposition \ref{thm:degree-odd}: no node can see a tie among its neighbors.
\begin{proof}
    Let $c$ be a proper $2$-coloring of $G$. 
    There are two cases:
    \begin{itemize}
    \item If $M_V\neq tie$, then w.l.o.g. assume that $M_V=red$. Since more than half the nodes are red and all red nodes have a  majority of blue neighbors,  we have a \texttt{Mm} illusion. 
    \item If $M_V=tie$, then all the nodes are under \texttt{w} illusion, since for all nodes, all neighbors are the other color. We have a \texttt{Uw} illusion.
    \end{itemize}
        
\end{proof}
Both cases used in the above proof are cases of Majority-\emph{weak}-majority illusions (which were already guaranteed to exist by Theorem \ref{thm:maj-weak-maj}), but we can also show the existence of the strict majority illusion in two different cases. 
First, when the number of nodes is odd,  there cannot be a global tie, so by using the first case in Lemma \ref{lem:2-col} we get the following proposition
:
\begin{proposition}\label{prop:2-col-maj-maj}
    In any properly $2$-colorable graph $G=\langle V,E \rangle$ with $|V|$ odd, a \texttt{Mm} illusion is possible. 
\end{proposition}
 \begin{proof}
    Let $c$ be a proper 2-coloring of $G$. Since $|V|$ is odd, $M_V\in \{red, blue\}$, we can use the first case of the proof of Lemma \ref{lem:2-col}: W.l.o.g. assume $M_V=red$. Since more than half of the nodes are red and all red nodes have only blue neighbors (since $c$ was a proper 2-coloring), we have a \texttt{Mm} illusion. 
\end{proof}
Second, when the color of a node can be swapped if needed, we can solve a tie:  
\begin{proposition}\label{prop:2col-strict}
  In any properly $2$-colorable graph $G=\langle V,E \rangle$ with some $i\in V$ such that for all $j\in N_i$ $d_j > 2$, a \texttt{Wm} illusion is possible. 
\end{proposition}
\begin{proof}
    Let $c$ be a proper coloring of $G$. 
    There are two cases:
    \begin{itemize}
    \item If $M_V\neq tie$, then 
    conformingly to 
    Lemma \ref{lem:2-col}, we have a \texttt{Mm} illusion.
    \item If $M_V=tie$, then 
    swap the color of node $i$:  let $c'$ assign the same colors as $c$ to all other nodes but $c'_i\neq c_i$. Now $M'_V= c'_i$. W.l.o.g. say $c_i=blue$ and $c'_i=red$. All of $i$'s neighbors are also red and have now exactly one red neighbor ($i$), and more than one blue neighbor. Therefore, all red nodes except for $i$ have more than half of their neighbors blue. Therefore, exactly $\frac{|V|}{2}$ of the nodes are in a situation of \texttt{m} illusion. %
    \end{itemize} 
    
\end{proof}

In \cite{Grandi2022}, the complexity of checking whether a network admits (what we call) a \texttt{Wm} illusion was left as an open problem. Propositions \ref{prop:2-col-maj-maj} and \ref{prop:2col-strict} show that by checking whether a graph is properly 2-colorable (which can be done in polynomial time \cite{Brown1996}) and whether there exists a node whose neighbors all have degree larger than $2$ or whether the number of nodes is odd, we can know that a graph admits a \texttt{Mm/Wm} illusion. 
Still this does not give us the complexity of checking whether a network admits a \texttt{Wm} illusion:  while this is a sufficient condition for a graph to allow for a \texttt{Mm/Wm} illusion, it is not a necessary one.


\subsection{Regular Graphs}
Some classes of networks seem promising for the prevention of illusions, in particular networks in which all nodes have similar degrees. 
In \cite{Centola2017}, theoretical analysis and experiments where human subjects were asked to perform estimation tasks are used to study the influence of network structure on the wisdom of crowds. The authors find a remarkable difference between centralized networks, where the degree distribution varies a lot between nodes, and decentralized (regular)  networks, in which all nodes have the same degree, regarding what social influence does to the accuracy of the estimates of individuals (when individual's estimates are based on a weighted average of their own belief and the beliefs of their neighbors).
They show that in decentralized networks, social influence significantly improves individual
accuracy and the group mean estimate. Furthermore, an overview of research about collective intelligence by Centola \cite{Centola2022} mentions several studies about decentralized networks in practical applications: in decentralized networks, political polarization and biases about climate change and immigration are reduced \cite{Becker2019,Guilbeault2018}, and social influence reduced biases about the risk of smoking \cite{Guilbeault2020}, as well as (implicit) race and gender biases in clinical settings \cite{Centola2021}.
Given that majority illusions also involve the wrong perception of a group, and  that decentralized/regular networks seem to be beneficial for group accuracy and bias reduction, 
we wonder whether they also are `good networks' in terms of the distortion we study: to what extent they do not allow for a majority of their nodes being under \texttt{m} illusion. According to Lerman's initial paper on majority illusions \cite{Lerman2016}, differences between the degrees of nodes and their neighbors are one of the main factors enabling majority illusion. Therefore, one would expect that regular networks, where all nodes have the same degree, make majority illusions less likely. 
Nevertheless, we will show that majority illusions (beyond the ones given by Theorem \ref{thm:maj-weak-maj}) 
are also possible in regular networks.

\medskip

A $d$-regular network 
 is a network in which all nodes have degree $d$.
We start by 
considering 
the 
simplest
cases where the network is made of one ore more cycles ($d=2$), and where the network is complete ($d=|V|-1$).
\begin{proposition}\label{prop:k=2}
    For any $2$-regular graph $G = \langle V,E\rangle$, 
    there exists no 2-coloring with which $G$ is under \texttt{Mm/Wm} illusion.
\end{proposition}
\begin{proof}
     Let $G = \langle V,E\rangle$ be any 2-regular graph. A node can only be under \texttt{m} illusion if both of its neighbors are of the minority color. Every minority-colored node can serve as a neighbor for at most two nodes. Thus, to give at least half of the nodes a \texttt{m} illusion, there must be at least $\frac{|V|}{2}$ nodes in the minority color, which is a contradiction with being a strict minority.
\end{proof}
However, a \texttt{Mw} illusion is possible, according to Theorem \ref{thm:maj-weak-maj}, and it is easy to find one (which we leave as an exercise to the reader).

\begin{proposition}
\label{prop:complete-weak-maj-maj}
    In a complete graph, no node can be under \texttt{m} illusion.
\end{proposition}
\begin{proof}
    Suppose towards a contradiction that there is a complete graph (i.e. a $d$-regular graph with $d=|V|-1$) $G = \langle V, E\rangle$ in which 
    there is a node under \texttt{m} illusion.
    Since there 
    is a node under \texttt{m} illusion, there cannot be a global tie. Without loss of generality assume $M_V=red$. Since all nodes observe all other nodes, all non-red nodes observe a majority of red neighbors, and all red nodes (a majority) observe either a majority of red nodes or a tie: no node is under \texttt{m} illusion, a contradiction.
\end{proof}

We know (by Theorem~\ref{thm:maj-weak-maj}) that a 
\texttt{M\emph{w}} illusion is possible on any complete graph
. 
We can go further and specify the types of colorings under which these graphs are in such an illusion. 

\begin{proposition}
\label{prop:complete-maj-weak-maj}
A complete $2$-colored graph $G=\langle V, E, c\rangle$ is in a  \texttt{Mw} illusion if and only if one of the following two holds:
    \begin{itemize}
        \item 
        the difference in numbers of nodes of each color is one; or
        \item 
        the number of nodes of each color is equal
        .
    \end{itemize}
\end{proposition}
Note that  if the second point holds (the number of nodes of each color is equal), then the graph is in a \texttt{Uw} illusion.

\begin{proof}
   If the difference in numbers of nodes of each color is one, assume w.l.o.g. that $M_V = red$. Then for all red nodes $r$, $M_r = tie$, so we have a \texttt{Mw} illusion.
    If the number of nodes of each color is equal, we have $M_V = tie$ but every node will observe a majority of the other color: we have a \texttt{Uw} illusion. 
    In all other cases, in which the difference in numbers of nodes of each color is greater than one, all nodes see the correct majority: no node is under \texttt{m/w} illusion.
\end{proof}

We return to the analysis of general regular graphs. 
The number of minority-colored neighbors needed for an illusion gives a restriction on the possible values of $d$ depending on $|V|$: 
\begin{proposition}
\label{prop:c}
If a 2-colored $d$-regular graph $G=\langle V,E,c \rangle$ with $|V|=n$ is in a \texttt{Mm/Wm} illusion, then $n$ and $d$ must satisfy
    \begin{itemize}
        \item $d\leq n-4$  if $n$ and $d$ are even; 
        \item  $d\leq n-3$  if one of $n$ and $d$  is even and one is odd.
    \end{itemize}
\end{proposition}

\begin{proof}
    This is a direct corollary of the more general Proposition \ref{prop:c-general} in 
    Section \ref{app:pq-illusions}. 
\end{proof}    

\begin{example}
    Consider a $d$-regular graph $G=\langle V, E\rangle $ with $|V|=6$ and $d=4$. For any node to be in a \texttt{m} illusion, at least $3$ of its neighbors have to have a different color than the global majority color. 
    Assume that the global majority color is red. Then there are at least $4$ red nodes and therefore only 2 nodes can be blue. Therefore, no node can have 3 or more blue neighbors.
\end{example}

The number of available edges of the minority color brings another requirement on the relative values of $|V|$ and $d$ for the strictest version.
 
\begin{proposition}
\label{prop:d}
    If a 2-colored $d$-regular graph $G=\langle V,E,c \rangle$ with $|V|=n$ is in a \texttt{Mm} illusion, then $n$ and $d$ must satisfy:
    \begin{itemize}
        \item  $n\geq \frac{2(3d+2)}{d-2}$ (assuming  $d > 2$) if  $n$ and $d$ are even;  
        \item  $n\geq \frac{2(3d+1)}{d-1}$  (assuming $d>1$) if  $n$ is even and $d$ is odd;
        \item  $n\geq \frac{3d+2}{d-2}$ (assuming $d >2$) if $d$ is even and $n$ is odd.
    \end{itemize}
\end{proposition}
    \begin{proof}
    If $G$ is in a \texttt{Mm} illusion, there are more than half of the nodes of one color. W.l.o.g., assume that this majority color is red, and that the minority color is blue.
    \begin{itemize}
        \item When $n$ and $d$ both are even, in order for a \texttt{Mm} illusion to exist, at least $\frac{n}{2}+1$ nodes have to be red. Nodes with an illusion have to have at least $\frac{d+2}{2}$ blue neighbors. Then, there have to be at least $\frac{n}{2}+1$ such nodes with an illusion. Thus there have to be at least $\frac{d+2}{2}(\frac{n}{2}+1)=\frac{(d+2)(n+2)}{4}$ edges to a blue node. Hence, there must be at least $\frac{(d+2)(n+2)}{4d}$ blue nodes because every blue node can have at most $d$ edges. Since at least $\frac{n}{2}+1$ nodes were red, there are at most $\frac{n}{2}-1$ left over to be blue, so this means that $\frac{(d+2)(n+2)}{4d}$ must be at most $\frac{n}{2}-1$. This is equivalent to $n\geq \frac{2(3d+2)}{d-2}$ assuming that $d>2$;
        \item When $n$ is even and $d$ odd, in order for a \texttt{Mm} illusion to exist, at least $\frac{n}{2}+1$ nodes have to be red. Nodes with an illusion have to have at least $\frac{d+1}{2}$ blue neighbors. Then, there have to be at least $\frac{n}{2}+1$ such nodes with an illusion. Thus there have to be at least $\frac{d+1}{2}(\frac{n}{2}+1)=\frac{(d+1)(n+2)}{4}$ edges to a blue node. Hence, there must be at least $\frac{(d+1)(n+2)}{4d}$ blue nodes because every blue node can have at most $d$ edges. Since at least $\frac{n}{2}+1$ nodes were red, there are at most $\frac{n}{2}-1$ left over to be blue, so this means that $\frac{(d+1)(n+2)}{4d}$ must be at most $\frac{n}{2}-1$. This is equivalent to $n+2\leq d(n-6)$, which means $n\geq \frac{2(3d+1)}{d-1}$ assuming that $d>1$;
        \item When $d$ is even and $n$ odd, in order for a \texttt{Mm} illusion to exist, at least $\frac{n+1}{2}$ nodes have to be red. Nodes with an illusion have to have at least $\frac{d+2}{2}$ blue neighbors. Then, there have to be at least $\frac{n+1}{2}$ such nodes with an illusion. Thus there have to be at least $\frac{d+2}{2}\cdot\frac{n+1}{2}$ edges to a blue node. Hence, there must be at least $\frac{(d+2)(n+1)}{4d}$ blue nodes. Since at least $\frac{n+1}{2}$ nodes were red, there are at most $\frac{n-1}{2}$ left over to be blue, so this means that $\frac{(d+2)(n+1)}{4d}$ must be at most $\frac{n-1}{2}$. This is equivalent to $n\geq \frac{3d+2}{d-2}$ assuming that $d>2$.
    \end{itemize} 
\end{proof}
    
    \begin{example}
        Consider a $d$-regular network with $|V|=6$ and $d=3$. Let us assume that red is the 
        global majority color, so we have at least 4 red nodes and at most 2 blue nodes. Then any node with a \texttt{m} illusion has at least 2 blue neighbors. Since for a \texttt{Mm} illusion there have to be at least 4 nodes with an illusion, there are at least $4\cdot 2=8$ edges to blue nodes necessary. However, since we have at most 2 blue nodes that each have only 3 edges, this is not possible.
    \end{example}
For any  $n$ and $d$ (with $d>2$ and $n$ or $d$ even) satisfying the above constraints, we can find a $d$-regular graph of size $n$ that has a \texttt{Mm} illusion. Note that this does not mean that for any $d$-regular graph of size $n$ 
we can find a coloring that gives a \texttt{Mm} illusion, because there exist many different regular graphs with the same $n$ and $d$. We only show that, for all combinations of $n$ and $d$ not excluded by our previous results,  there exists at least one such graph with the illusion, and that we know how to find it.
\begin{figure}[t]
    \centering
    \begin{subfigure}[t]{.25\textwidth}
        \centering\captionsetup{width=.8\linewidth}
        \includegraphics[width=0.6\textwidth]{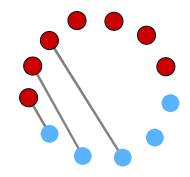}
        \caption{Color 7 nodes red and 5 nodes blue. Draw lines from red to blue nodes, }
    \end{subfigure}%
    \begin{subfigure}[t]{.25\textwidth}
        \centering\captionsetup{width=.8\linewidth}
        \includegraphics[width=0.6\textwidth]{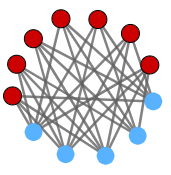}
        \caption{until each red node has $\lfloor{\frac{d}{2}+1}\rfloor=4$ connections to a blue node.}
    \end{subfigure}%
    \begin{subfigure}[t]{.25\textwidth}
        \centering\captionsetup{width=.8\linewidth}
        \includegraphics[width=0.6\textwidth]{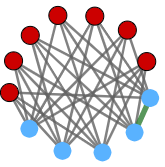}
        \caption{Add an edge (green) 
        such that all blue nodes have 6 edges.}
    \end{subfigure}%
    \begin{subfigure}[t]{.25\textwidth}
        \centering\captionsetup{width=.8\linewidth}
        \includegraphics[width=0.6\textwidth]{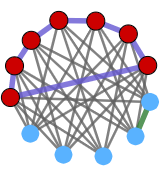}
        \caption{Draw a regular graph (purple) of the remaining edges between the red nodes.}
    \end{subfigure}%
\caption[short]{Example of the algorithm for Theorem \ref{thm:regular-maj-maj}, with $n=12$, $d=6$.
}
\label{fig:algorithm-12-6}
\end{figure}
\begin{theorem}\label{thm:regular-maj-maj}
    Let $n$ and $d$ be any two integers such that 
    $d>2$ and $d$ or $n$ is even. 
    If the conditions of Propositions \ref{prop:c} and 
    \ref{prop:d} are met, 
    there exists a $d$-regular graph $G=\langle V, E\rangle$ with $|V|=n$ in which a \texttt{Mm} illusion is possible. 
\end{theorem}

\noindent\emph{Proof sketch.}
    We prove this by construction: we give an algorithm that takes as input $n$ and $d$ and returns a regular graph with $n$ nodes of degree $d$ that has a \texttt{Mm} illusion. The algorithm and a proof that the algorithm outputs the desired graph are given in Appendix 
    \ref{app:algorithm}. See Figure \ref{fig:algorithm-12-6} for an example with 12 nodes of degree 6.

     Propositions \ref{prop:c} and \ref{prop:d} and Theorem \ref{thm:regular-maj-maj} together give necessary and sufficient conditions for $n$ and $d$ for the existence of a $d$-regular graph with $|V|=n$ nodes on which a \texttt{Mm} illusion is possible. 
     See Figure \ref{fig:ill_Thm2} for an illustration of the values of $n$ and $d$ that meet the conditions in Theorem \ref{thm:regular-maj-maj}.
\begin{figure}
    \centering
    \includegraphics[width=\linewidth]{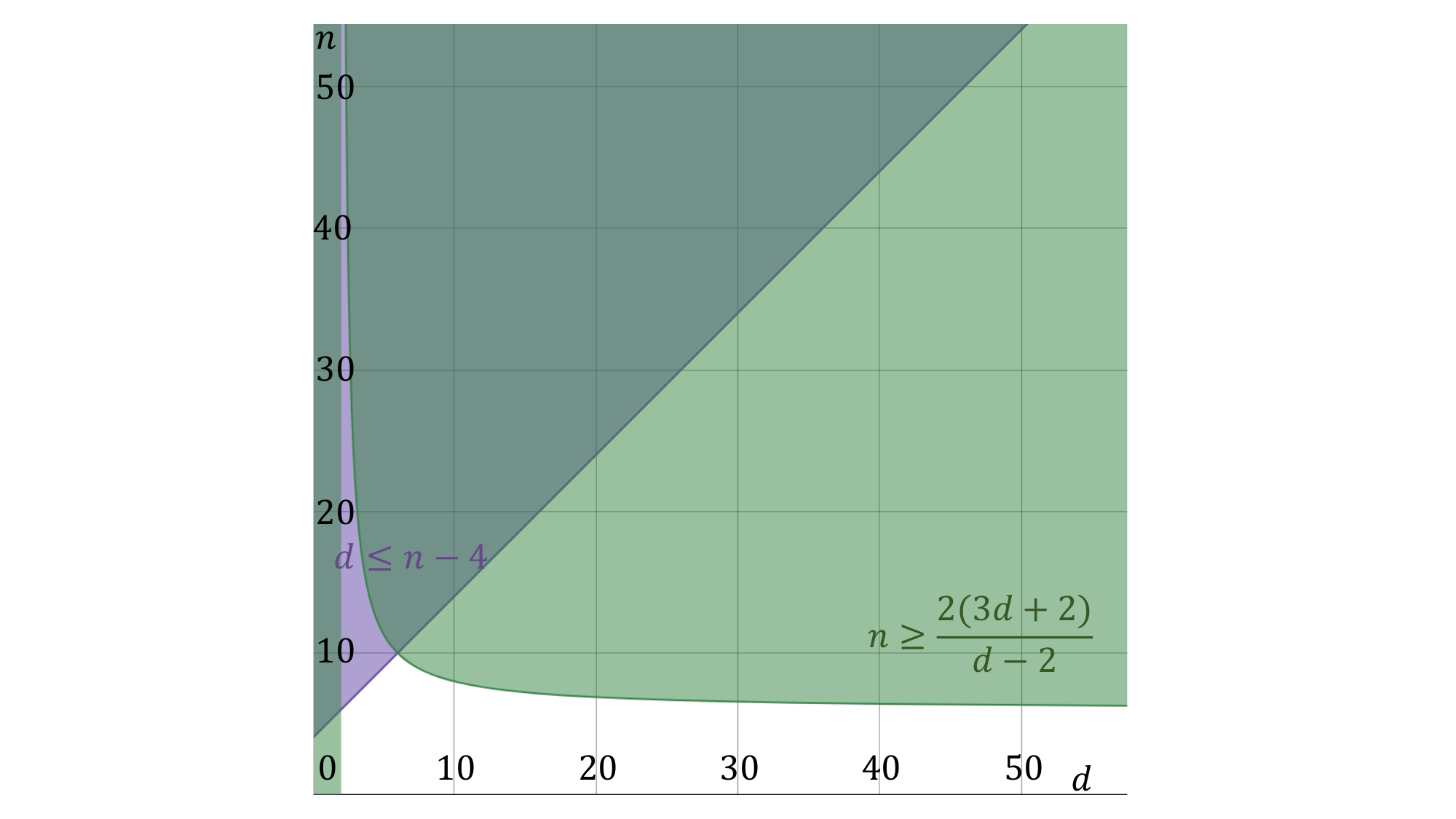}
    \caption{
    Illustration of the values of $n$ and $d$ for which the conditions in Theorem \ref{thm:regular-maj-maj} hold, for even $n$ and $d$ (the cases where one of $n$ and $d$ is odd are similar). The purple shaded area marks the values of $n$ and $d$ satisfying the condition in Proposition \ref{prop:c}, the light green shaded area those satisfying the condition in Proposition \ref{prop:d}, and the dark green area those satisfying the conditions in both propositions.}
    \label{fig:ill_Thm2}
\end{figure}

\begin{remark}
    Although we already know from Theorem \ref{thm:maj-weak-maj} that 
    \texttt{M\emph{w}} illusions are possible in any graph, the technique from Proposition \ref{prop:d} can be used to get a requirement on the relative values of $n$ and $d$ in $d$-regular graphs of $n$ nodes where \emph{specific} \texttt{Mw} illusions are possible. Namely, graphs where \texttt{Mw/Ww} illusions are possible with a 2-coloring that does not have a global tie, or where monochromatic \texttt{Mw/Ww} illusions are possible with a 2-coloring that has a global tie. We leave the details as an exercise for the reader.
\end{remark}
\medskip
In conclusion, contrary to the expectation that decentrality of a network would prevent it from being under majority illusion, we found that even \texttt{Mm} illusions are possible on some regular networks. The question remains whether they are less likely on more centralized networks, something that we will study in Section \ref{sec:simulations}.

\medskip

\begin{remark}
    Since any $2$-coloring of a graph is also a $k$-coloring for $k>2$
    , all results from Sections \ref{sec:arbitrary-networks} and \ref{sec:specific-networks} 
    giving sufficient conditions for the existence of illusions (Theorems \ref{thm:maj-weak-maj} and \ref{thm:regular-maj-maj}, and Propositions \ref{prop:2-col-maj-maj}, \ref{prop:2col-strict}
    ) 
    hold also in the case where the size of the set of colors $\mathcal{C}$ is larger than 2.
    
\end{remark}

\section{Computational Experiments}\label{sec:simulations}
So far, we used purely analytical methods to determine whether 
majority illusions are possible on precisely defined classes of graphs. 
We did not yet restrict our focus to 
types of graphs that would resemble \emph{real} social networks. 
While there is no strict graph-theoretical classification of such networks, there are certain graph-theoretical properties that are known to commonly occur in social networks, in terms of the nodes' degree and the distance between nodes. Typical social networks often come with a high clustering coefficient, a power-law-like degree distribution and small distances between the nodes. This type of networks are hard to treat analytically. We  
therefore take an empirical route and
perform computational experiments in which we generate random networks exhibiting these realistic properties to a greater or lesser extent, and examine \emph{how often} illusions occur and \emph{to which degree}. We study three types of networks: Erd\H{o}s-R\'{e}nyi networks, which we treat as a baseline against the two other more realistic networks we consider; Holme-Kim networks; and a friendship network obtained from Facebook data.

Erd\H{o}s-R\'{e}nyi networks \cite{ErdosRenyi} (also used in the first work on majority illusions \cite{Lerman2016}) are random graphs, where only the number of nodes $n$ and the probability of two nodes being connected by an edge $p_{edge}$ are specified. They have a Poisson degree distribution, and do not have the typical properties of real social networks. 
Holme-Kim networks \cite{HolmeKim_2002} are random networks that are, in contrast,  designed to capture at least two properties commonly observed in real social networks. They are scale-free, that is, the degree distribution in the network asymptotically follows a power-law; and they have the small-world property, which means that on average the distances in the network are low and the clustering coefficient is high. 
The idea is to use the Barabasi-Albert algorithm \cite{Barabsi2002} 
to create a graph with a powerlaw degree distribution, but to have an extra step when a node is added to form more triangles, in order to get a higher clustering coefficient. The number $m$ of edges to add for each node and the probability $p$ of adding a triangle after adding a random edge are input parameters of the algorithm.
Finally, we also run an experiment with a real Facebook friendship network.

For each of these network types, we study both the likelihood of the occurrence of majority illusions and their extent, that is, the fraction of nodes under illusion and the error of each node 
in their estimation of the proportion of blue and red nodes
.

\paragraph{Hypotheses}
We have a couple of expectations regarding the results of our experiments.
First (H1), we expect less illusions and smaller error (smaller difference for nodes in local and global proportion of red and blue) when the networks are more connected (higher clustering coefficient, higher $p_{edge}$ in Erd\H{o}s-R\'{e}nyi graphs, higher $m$ and $p$ in Holme-Kim graphs), because when nodes have more edges, they have more information about the network and therefore should be less likely to be under illusion.
Second (H2), 
if the global distribution of colors is close to 50/50, we intuitively expect more nodes to be under \texttt{m/w} illusion, since nodes are more likely to \emph{just} see the wrong majority.
Third (H3),  given the analysis in the line of work of Centola and others on network centralization \cite{Centola2022}, we expect a positive correlation between centrality of the network and the amount of majority illusions: in more centralized networks, we expect more illusions than in less centralized networks. 
For the same reason, since Erd\H{o}s-R\'{e}nyi networks asymptotically become regular networks for large $n$, we expect less illusions with larger $n$.

Finally (H4), given the results of \cite{Lerman2016} that in Erd\H{o}s-R\'{e}nyi networks where high-degree nodes tend to connect to low-degree nodes majority illusions are more likely, we expect to find 
a negative correlation between degree assortativity and the likelihood or strength of illusions too.
And since differences between degrees of nodes are larger in larger Holme-Kim networks, we also expect a positive correlation between the size of a Holme-Kim network and the likelihood or extent of illusions.

\subsection{Experimental Set-up}
For creating Erd\H{o}s-R\'{e}nyi networks, we use the implementation of the NetworkX Python package \cite{networkX}, the \texttt{erdos\_renyi\_graph} function. It takes as parameters the graph size $n$ and the probability of an edge between two nodes $p_{edge}$.
For creating Holme-Kim networks, we also use the implementation of the NetworkX package, the \texttt{powerlaw\_cluster\_graph} function. This function takes as parameters $n$, the number of nodes of the graph; $m$, the number of random edges to add for each new node; and $p$, the probability of adding a triangle after adding a random edge; and returns a semi-randomly generated graph with the given parameters. Note that if $p = 0$, the graph is a Barabasi-Albert graph \cite{Barabsi2002}. 

To test the occurrence of majority illusions on `real' social networks, we used an existing Facebook friendship network from Stanford \cite{Facebooknetwork}. The network consists of 4039 nodes and 88234 edges.\footnote{For more details on the dataset, see \href{https://snap.stanford.edu/data/ego-Facebook.html}{https://snap.stanford.edu/data/ego-Facebook.html}.}

For both the experiments with Erd\H{o}s-R\'{e}nyi and Holme-Kim networks, we generated random graphs for values of $n$ varying from 20 to 100 and 21 to 101 with steps of 20, to capture both even and odd sized networks. For the Erd\H{o}s-R\'{e}nyi networks, we used a value of $p_{edge}$ varying from 0.1 to 0.9 with steps of 0.2 (the cases where $p_{edge}=0$ and $p_{edge}=1$ are less interesting since they correspond respectively to the empty and the complete graph), and for the Holme-Kim networks we used for every value of $n$ a low, middle, and high value of $m$ (respectively 0.1, 0.5, and 0.9 times $n$), and $p$ varying from 0 to 1 with steps of 0.1. For every such combination of $n$ and $p_{edge}$ (respectively $n$, $m$, and $p$), we generated 5000 Erd\H{o}s-R\'{e}nyi (respectively Holme-Kim) graphs and colored them randomly in red and blue, with the probability $p_{blue}$ of a node being colored blue varying from 0.1 to 0.5 with steps of 0.1 (so 1000 random graphs for every value of $p_{blue}$).
Also the Facebook network was colored randomly 1000 times for every value of $p_{blue}$ between 0.1 and 0.5, in steps of 0.1.

For every such generated colored graph, we check whether it contains a \texttt{Mm/Mw/Wm/Ww} 
illusion:
we count the number of nodes under \texttt{m/w} illusion and  for each node we measure its error
.
The error is measured as the difference between the proportion of blue/red nodes globally and the proportion of blue/red nodes in the node's neighborhood, which describes \emph{how wrong} a node is about the proportion of blue and red, rather than just whether it is wrong or not about the majority.
As a global measure of error, we will use the mean squared error of a group of nodes. The reason for this is that if we view the nodes as `estimators' of the global proportion of blue/red, as we do in this context, it becomes natural to measure the quality of these estimators: `how wrong' the nodes are. The mean squared error is a standard way to assess the accuracy of estimators, penalizing positive and negative errors by the same amount, and larger errors more than smaller errors.
We also measure some other properties of the graph 
to see whether or not there is any correlation between those properties and the fraction of nodes under illusion. Some of these network properties are directly motivated from our hypotheses: 

\begin{itemize}
    \item \emph{Global fraction of blue/red nodes} (naturally correlated with $p_{blue}$).
    \item Network \emph{centrality}: a global measure of inequality of the local centralities (degree, eigenvector centrality \cite{Landau1895}, closeness centrality, and betweenness centrality) of nodes in the network.  Both global and local centrality measures (except for eigenvector centrality) are explained in \cite{Freeman1978}, and computed using the respective NetworkX local centrality functions. 
    \item \emph{Degree assortativity coefficient}: how much nodes are connected to nodes with similar degree. The value of the coefficient is between -1 (disassortativity: negative correlation between the degrees of neighbors) and 1 (assortativity: positive correlation between the degrees of neighbors). The coefficient is calculated using NetworkX's \texttt{degree\_assortativity\_coefficient}, which uses \cite[Equation (21)]{Newman2002Deg_assort}.
\end{itemize}
Others are standard network measures: the \emph{average path length}, the \emph{clustering coefficient},  the fraction of nodes in the \emph{largest component} of the network (only for Erd\H{o}s-R\'{e}nyi graphs, since  Holme-Kim graphs consist of only one component), and the \emph{homophily} \cite{McPherson2001, Easley2010-NCM} regarding the colors of the nodes\footnote{Homophily is measured as the difference between the probability of a two-colored edge in a random graph with the same number of blue and red nodes as in the current graph, and the actual fraction of two-colored edges in the current graph. A high homophily means that there are more links between same-colored nodes than expected if the links were formed randomly, and a low homophily means that there are less links between same-colored nodes than expected.}.

All experiments were coded in Python 3.10.5, using the NetworkX package \cite{networkX}, and run on 
an Intel Core i7 processor running at 3GHz with 16 GB of RAM, using Windows 10 as operating system. Statistical analyses were performed in R.
All code can be found on \href{https://github.com/MaaikeLos/Majority_Illusions}{https://github.com/MaaikeLos/Majority\_Illusions}.

\subsection{Results}
We report below the most interesting results. For a more complete analysis, we refer to the R code. Some additional plots that are helpful for a deeper understanding of our results are provided in 
Appendix \ref{app:experiment_extras}.
\paragraph{General observations}
Figure \ref{fig:med_frac_nodes_ill} shows for the three different graph types the median and interquartile range of the fraction of nodes that are under \texttt{m/w} illusion per network, over all parameter values together. We observe that \texttt{m} illusions are scarce both in Erd\H{o}s-R\'{e}nyi and Holme-Kim networks. Furthermore, we see that in Holme-Kim networks, \texttt{w} illusions occur slightly more often (median 2.5 percent of all nodes) than in Erd\H{o}s-R\'{e}nyi graphs (median 0 percent of the nodes), and that in the Facebook network nodes are more frequently under \texttt{m} (4 percent) and \texttt{w} (7.5 percent) illusion than in both random graph types. This could point to a tendency of more illusions in real networks than in random graphs. However, since the parameters of the different graph types are incomparable, we should be careful concluding such relation. 

\begin{figure}[h]
    \centering
    \includegraphics[width=0.8\linewidth]{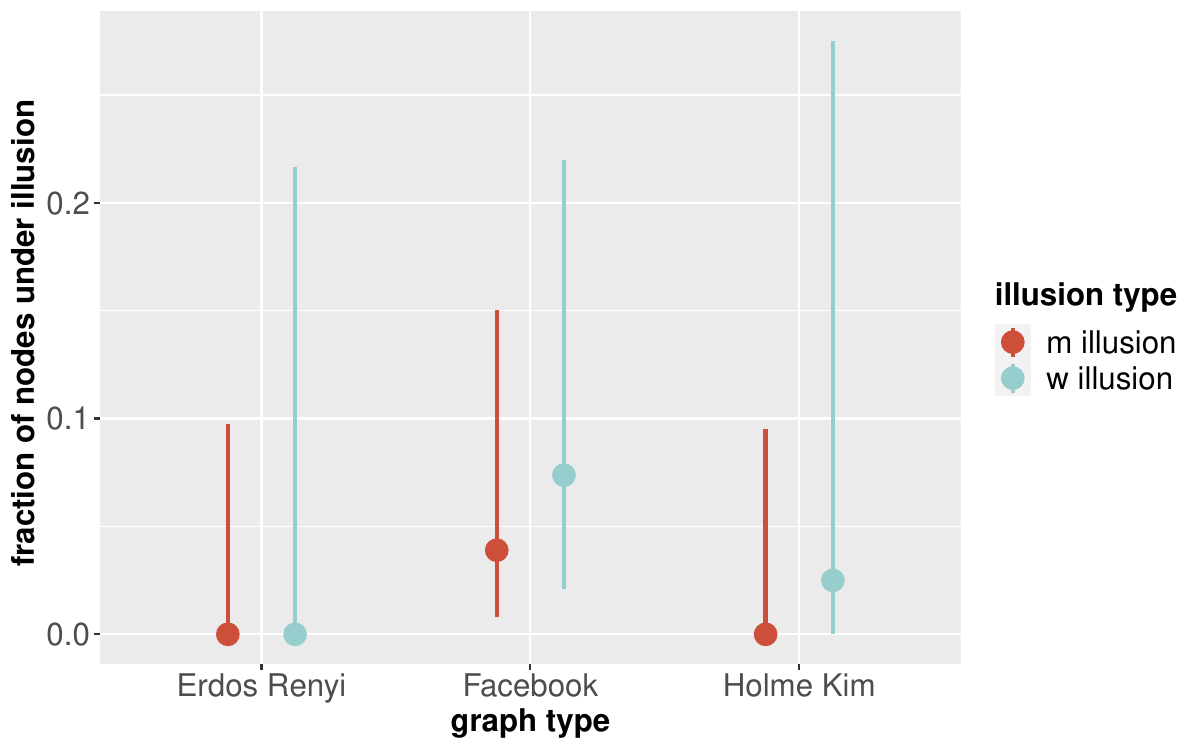}
    \caption{Median and interquartile range of fraction of nodes under \texttt{m/w} illusion, per graph type.}
    \label{fig:med_frac_nodes_ill}
\end{figure}

Figure \ref{fig:WMwmi} shows the fraction of networks in the experiment under different types of illusion. 
We find that in Erd\H{o}s-R\'{e}nyi networks it is very unlikely that at least half of the nodes are under \texttt{m} illusion, and that in approximately 7 percent of the networks at least half of the nodes are under \texttt{w} illusion. In Holme-Kim networks illusions occur more often: approximately 6 percent of the networks is in a \texttt{M/W}\texttt{m} illusion, and around 15 percent has a \texttt{M/W}\texttt{w} illusion. The Facebook network is in between the two random network types: for around 2 percent of the colorings the network is in a \texttt{Mm} illusion, while in approximately 13 percent there is a \texttt{Mw} illusion.
\begin{figure}
    \centering
    \includegraphics[width=0.8\linewidth]{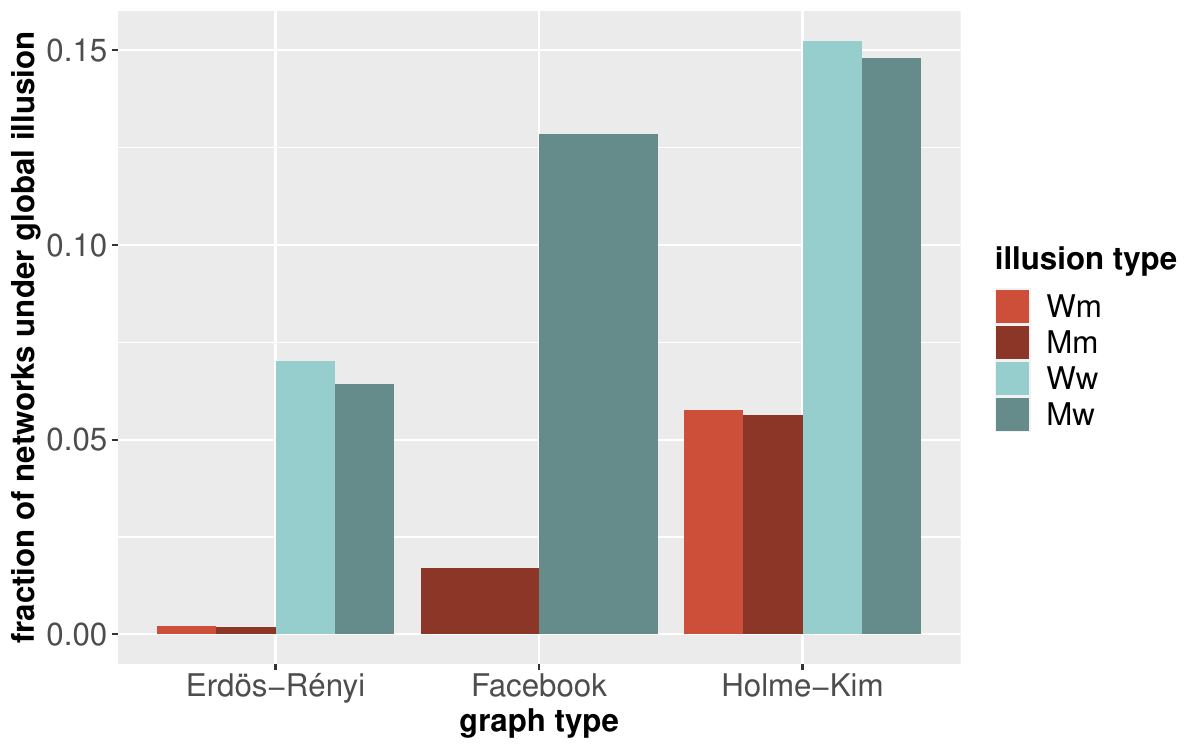}
    \caption{Fraction of networks under illusion per graph type.}
    \label{fig:WMwmi}
\end{figure}

Figure \ref{fig:MSE} shows the mean squared error of different groups of nodes for the three network types, which gives us a quantitative measure of the illusion of nodes. Indeed, nodes with \texttt{m} illusion have a larger error than nodes only under \texttt{w} illusion, which again have a larger error than nodes without illusion. This effect appears to be strongest in the Facebook network, where the mean squared error of nodes with \texttt{m} illusion is approximately 20 times higher than that of nodes with no illusion, while in Erd\H{o}s-R\'{e}nyi and Holme-Kim networks the mean squared error of nodes with \texttt{m} illusion is respectively 10  and 4 times as high as that of nodes with no illusion.

\begin{figure}
    \centering
    \includegraphics[width=0.8\linewidth]{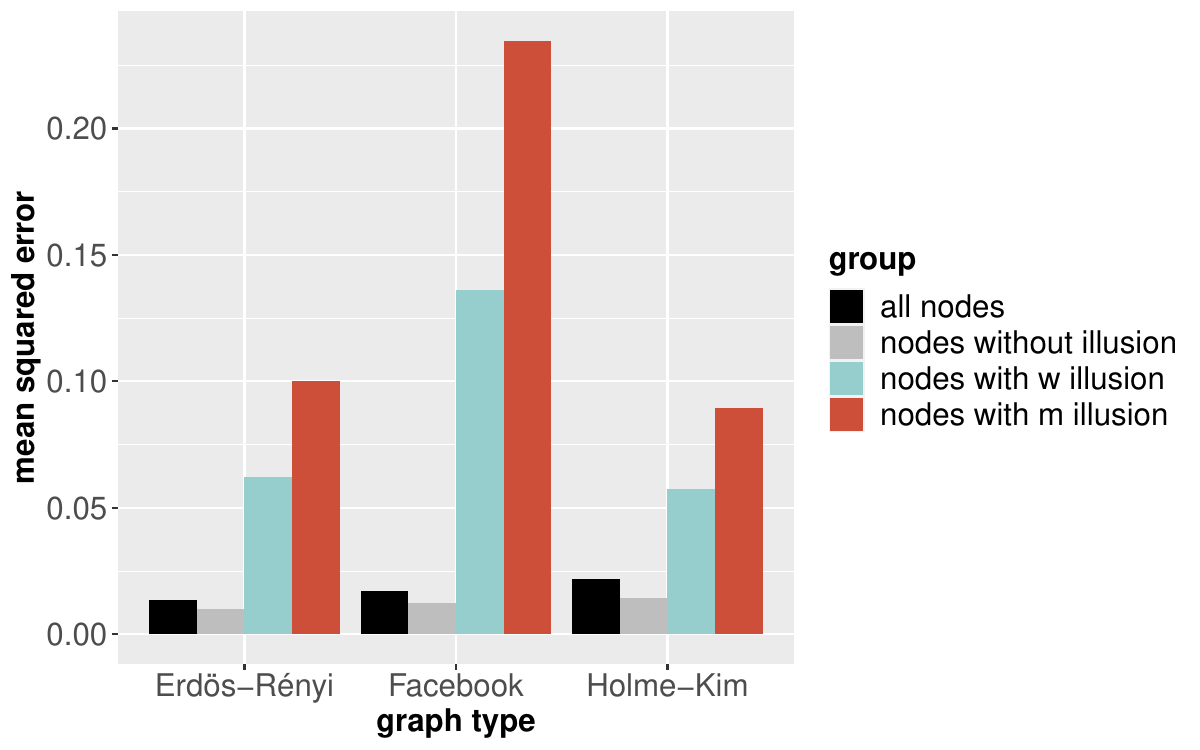}
    \caption{Mean squared error of different groups of nodes
    , per graph type.}
    \label{fig:MSE}
\end{figure}

We consider the pairwise Pearson correlation between the fraction of nodes under illusion, the mean squared error of nodes, and the other graph properties. The full correlation matrices 
can be found in Appendix \ref{app:experiment_extras}, in 
Figure \ref{fig:cormatER} (Erd\H{o}s-R\'{e}nyi graphs), Figure \ref{fig:cormatHK} (Holme-Kim graphs), and \ref{fig:cormatFB} (Facebook network). Table \ref{tab:correlation_table} shows a summary.
\begin{table}[t]
    \caption{Pearson correlation coefficient between the fraction of nodes under \texttt{m/w} illusion and the mean squared error (MSE) of nodes, and different parameters and graph properties. }\label{tab:correlation_table}
    \centering
    \begin{tabular}{|r||c|c||c|c||c|c|}
    \hline
     & \multicolumn{2}{c||}{Erd\H{o}s-R\'{e}nyi graphs} & \multicolumn{2}{c||}{Holme-Kim graphs}& \multicolumn{2}{c|}{Facebook network} \\ \hline
         & \textbf{frac.   ill.} & \textbf{MSE} & \textbf{frac.   ill.} & \textbf{MSE }& \textbf{frac.  ill.} & \textbf{MSE }  \\ \hline \hline
        $p_{edge}$ &  -0.38 & -0.6 & - & - & - & - \\ \hline
        $m$  &  - & - & -0.08 & -0.14
        \footnotemark[2] & - & -\\ \hline
        $p$  &  - & - & 0 & 0.01& - & - \\ \hline
        prop. blue/red globally & 0.63 & 0.15
        \footnotemark[1]& 0.57 & 0.2 & 0.88 &0.88\\ \hline
        $n$ & -0.17 &-0.32 & -0.16 & -0.34& - & - \\ \hline
        clustering coefficient & -0.38 & -0.6 & -0.2 & -0.34 & - & -\\ \hline
        degree centrality & 0.09 & 0.09 & 0.07 & 0.17 & - & - \\ \hline
        betweenness centrality & 0.38 & 0.81 & 0.25 & 0.5 & - & - \\ \hline
        degree assortativity & -0.12 & -0.24 & -0.08 & -0.18 & - & - \\ \hline
    \end{tabular}
    \footnotetext[1]{Note that although the general MSE has a small positive correlation with the proportion of blue/red globally, the MSE of nodes with illusion has a quite strong negative correlation with the proportion of blue/red globally. See Figure \ref{fig:cormatER}.}
    \footnotetext[2]{Figure \ref{fig:MSE_HK} shows that this is not straightforward negative correlation, but that MSE is high for low and high values of m, and low for middle values of m.}
\end{table}
From these Pearson correlations we can learn the following about our initial hypotheses. 
\paragraph{Graph parameters}
 In Erd\H{o}s-R\'{e}nyi graphs, as expected in H1, the more connected a network is (the higher $p_{edge}$\footnote{And thus the higher the global clustering coefficient, which is completely dependent on $p_{edge}$.}), the lower the fraction of nodes under illusion and the mean squared error of nodes. 
 In Holme-Kim graphs, however, we did not find such negative correlation between either $m$ (the number of random edges added for each new node) or the probability $p$ of adding a triangle after adding a random edge (the parameter that determines the average clustering coefficient of the network), and the fraction of nodes under illusion / the mean squared error
 . However, there is a small negative correlation between the \emph{global} clustering coefficient and the fraction of nodes under illusion / the 
 mean squared error. 
 
 In all simulations we observed that when the global proportions of blue and red were closer to 0.5 than to 0 and 1,
\texttt{Mm/Mw/Wm/Ww} 
illusions are more likely to happen, 
on average more nodes are under illusion,
and nodes have a larger error.
This is in line with the hypothesis H2 that in networks where the two colors are occurring equally often, it is easier for a node to be under illusion than in networks where there are more nodes of one color than of the other.

Both for Erd\H{o}s-R\'{e}nyi networks and for Holme-Kim networks, we observe a trend that in larger networks, 
illusions (both on the network level and on the node level) are slightly less likely and nodes have a smaller error. For Erd\H{o}s-R\'{e}nyi networks, this is in line with our hypothesis\footnote{And, indeed, there is a negative correlation between the network size and the degree centrality (see Figure~\ref{fig:cormatER}, however not between degree centrality and the amount or size of illusions.} (H3), for Holme-Kim networks it is not\footnote{Moreover, we see almost no correlation between network size and degree assortativity (see Figure~\ref{fig:cormatHK}), and degree assortativity is only lightly correlated with amount and size of illusions.} (H4). However, further experiments are necessary to know whether this trend continues for networks larger than 100 nodes.

\paragraph{Centrality}
We considered different measures of network centrality, of which degree centrality is the most commonly used one.  We expected in H3 a positive correlation between network centrality and likelihood and strength of illusions.
Indeed, we found a positive correlation between some of the centrality measures and the fraction of nodes under illusion / the mean squared error of nodes in Erd\H{o}s-R\'{e}nyi graphs, but surprisingly not for \emph{degree centrality}, the most straightforward measure of centrality which is considered in the line of literature by Centola et al \cite{Centola2017, Centola2021, Centola2022}. In Holme-Kim graphs, we found substantial correlation only between \emph{betweenness} centrality (the variance in nodes' local betweenness centrality, which is 
the number of 
shortest paths between pairs of nodes in the network 
on which
the node lies) and the amount / size of illusion.
Our hypothesis H3 is therefore neither directly confirmed nor rejected, and further analysis of the connection between the different types of centrality and majority illusions is necessary for a clear result.

\paragraph{Degree assortativity}

We found, in line with H4, a small negative correlation between the degree assortativity coefficient of an Erd\H{o}s-R\'{e}nyi graph and both the fraction of nodes under illusion in the graph and their error. 
However, since we observed that  the correlation is more complex
, we cannot accept H4 straightforwardly. With a degree assortativity coefficient close to 0, the fraction of nodes under illusion and the mean squared error are close to 
0, with lower or higher degree assortativity the fraction of nodes under illusion and the mean squared error are higher (see for an illustration Figure \ref{fig:deg_assort} in Appendix \ref{app:experiment_extras}). This means that if there is no correlation between nodes' degree and the degree of their neighbors, illusions are less likely and smaller than when there is either positive or negative correlation between the degrees of nodes and their neighbors. This is different from 
our hypothesis H4 based on 
the result by Lerman \cite{Lerman2016} (that higher degree assortativity would correlate with less/smaller illusions).
In Holme-Kim graphs, the correlation is even smaller, 
and we did not find a clear pattern in the relation, but note that we have almost no data for positive values of the degree assortativity coefficient: apparently the degree assortativity is mostly negative in Holme-Kim graphs. Therefore, more elaborate experiments would be necessary to study the relation between degree assortativity and majority illusions in Holme-Kim graphs.

\paragraph{Further observations}
We wondered whether there would be a relation between the \emph{homophily}
regarding nodes' colors, and the size and amount of majority illusions in a graph. One would expect that homophily decreases the probability and size of illusions, since it makes the majority of nodes more likely to see the correct majority. In Holme-Kim graphs we found such negative correlation but only small. Also, with homophily around 0, there are less and smaller illusions that with either higher or lower homophily (see for an illustration Figure \ref{fig:homophily} in Appendix \ref{app:experiment_extras}). However, since in our experiments homophily was only a measured value rather than an input variable, we do not have a good representation of graphs with different homophily values. Conducting experiments in which networks are deliberately colored to have different degrees of homophily would be necessary to form any conclusion about the relation between homophily and majority illusions. This is an interesting direction for further research.

Next, in Erd\H{o}s-R\'enyi graphs, we observe a quite high positive correlation (0.79) between the \emph{average path length} in the graph and the mean squared error of the nodes. We observed small positive correlations (between 0.2 and 0.35) between the average path length and the fraction of nodes under illusion in Erd\H{o}s-R\'enyi and Holme-Kim graphs, and the mean squared error in Holme-Kim graphs. We do not see a direct explanation for this, but note that the average path length is negatively correlated to the clustering coefficient and positively correlated to centrality measures.
Similarly, the negative correlation between the \emph{fraction of nodes in the largest component} of the graph and the mean squared error and fraction of nodes under illusion in Erd\H{o}s-R\'enyi graphs could possibly be explained by its negative correlation with centrality measures.

\section{Generalizations of Majority Illusions}\label{sec:generalizations}
Although  
the majority illusion is the only illusion 
discussed in the literature, other illusions are worth studying too. After all, the majority opinion around us is not the only thing that can influence us. While we might still want to assume that we are primarily influenced by the \emph{dominant} opinion in our surroundings, dominant can be understood in different ways, `majority' is only one way. 
Another way to define dominance is by means of a threshold: if we are not that easily convinced maybe a majority is not enough to influence us, and we only adopt an opinion if more than, for example, 2/3 of the people around us shows the opinion. In this way we can define $q$-illusions, where $q$ is the threshold of the illusion.
This is a generalisation of the \texttt{m} illusion: with a threshold of 1/2 the threshold is normal majority. 
Another generalization arises when there are more than 2 options. Then a natural conception of dominance is plurality: the option that occurs most often. One is under plurality illusion if the option that one observes most often in their neighborhood is different from the option that occurs most often in the population. Again this boils down to majority when there are only two options, then the option that occurs most often is also the majority option.
Theoretically, a way to define dominance is equivalent to a voting rule: which option would win if every node would vote for their option? Therefore, we can give a general definition of illusions based on voting rules: if, according to this voting rule, another option wins among my direct neighbors than in the total population, I am under illusion of this voting rule. And if we choose as voting rule `majority', we are back to \texttt{m} illusions.  
In this second part of this article we explore the general definition of illusions, and more specifically the illusions for quota rules and plurality. We make a start with the study of which networks allow for which kinds of illusions.

\paragraph{Preliminaries}
In this section we will consider generalizations of the definition of \texttt{Mm} illusions on the same class of graphs we studied so far: irreflexive, symmetric, finite, simple graphs.
In the definition of `Majority-majority illusions', the first `Majority' refers to the fraction of agents that is under illusion, the second `majority' refers potentially to the method with which 
opinions are aggregated.
Hence, we see the first as a quotum (the number of agents necessary to be under illusion before we say the network as a whole is under illusion) which we can generalize to arbitrary quota, and the second as a voting rule, which we can generalize to arbitrary voting rules. The intuition behind this second generalization is that agents can be influenced by their neighbors in different ways, not only by adapting the opinion that at least a majority of their neighbors have.

As mentioned before, many of the previous results also hold on graphs with more than two colors, but since many voting rules are similar or even equivalent to majority in the case with only 2 colors, we here explicitly generalize to any number $k$ of colors. Hence, we generalize the concept of \texttt{Mm} illusions on 2-colored graphs to $p$-$\mathcal{R}$ illusions on $k$-colored graphs for any fraction $p$ (not just $\frac{1}{2}$), any voting rule $\mathcal{R}$ (not just majority), and any number of colors $k$ as follows. Given a set $S$ of agents and a coloring $c$, a voting rule $\mathcal{R}: S\to 2^{\mathcal{C}}$ is a function outputting a set of winning colors. When several colors are winning, we call this a `tie'. 

\begin{definition}[General definition for illusions] \label{def:general_illusions}
    Given a $k$-colored graph $ C = \langle V, E, c\rangle$,  agent $i\in V$ is under
    \begin{itemize}
        \item $\mathcal{R}$-illusion for a voting rule $\mathcal{R}$ if $\mathcal{R}(N_i)\cap \mathcal{R}(V) = \emptyset$ and not  $\mathcal{R}(N_i) =\mathcal{R}(V) = \emptyset$;
        \item weak-$\mathcal{R}$-illusion if $\mathcal{R}(N_i) \not = \mathcal{R}(V)$.
    \end{itemize}
    A $k$-colored graph $ C = \langle V, E, c\rangle$ is in a  
    \begin{itemize}
        \item $p$-(weak-)$\mathcal{R}$-illusion if more than a $p$-fraction of the group is under (weak-)$\mathcal{R}$-illusion;
        \item weak-$p$-(weak-)$\mathcal{R}$-illusion if at least a $p$-fraction of the group is under (weak-)$\mathcal{R}$-illusion.
    \end{itemize}
\end{definition}
In this paper we considered so far $\mathcal{R}\in \{ \mathrm{unanimity, majority}\}$,
and we will consider in this section  $\mathcal{R}\in \{\mathrm{plurality}, q\text{-}\mathrm{quota}\}$
. One could study illusions with any other voting rule, but in the context of opinion diffusion we think those are the most natural rules to consider. 
Furthermore, these two rules are natural extensions of the majority rule: with only two colors plurality is equivalent to majority, while for 
$q=\frac{1}{2}-\frac{1}{2|V|}$ the $q$-quota rule is equivalent to majority\footnote{One would expect majority to be equivalent to the $q$-quota rule for $q=\frac{1}{2}$, but since we consider ties in majority as cases where there are two winners instead of no winners, majority is equivalent to the $q$-quota rule for $q=\frac{1}{2}-\epsilon$ where $\epsilon = \frac{1}{2|V|}$, such that $q|V|=\frac{|V|-1}{2}$. In this way, if exactly half of the nodes are one color, still that color is a $q$-quota winner, what it would not be with $q=\frac{1}{2}$.}.
\begin{remark}
    In the case where there are only two colors and the rule is majority, Definition \ref{def:general_illusions} reduces to Definition \ref{def:maj-ill-strict}. Mind however that one cannot use directly the phrasing of Definition \ref{def:maj-ill-strict} for more colors, because there can be local ties without overlap with the global majority (or vice versa). 
    For example, in case $M_{N_i}=\{green, blue\}$ (half of the neighbors are green and half of them are blue) and $M_V=\{red\}$,  
    agent $i$ is completely mistaken about the global majority, and therefore is under strict \texttt{m} illusion, as captured by  Definition \ref{def:general_illusions}. However, 
    if one were tempted to apply Definition \ref{def:maj-ill-strict} as it is stated but for more colors, 
    one would think that agent $i$ is not under strict \texttt{m} illusion, because  $M_{N_i}=tie$. 
    Since we only considered two colors in the first part of this paper, we chose to use the more readable Definition \ref{def:maj-ill-strict} and avoid introducing unnecessary details, but we could have defined \texttt{m} illusions for any number of colors: 

    \begin{definition*}[Definition \ref{def:maj-ill-strict} for multiple colors]
        Given a colored graph $C=\langle V, E, c\rangle$, an agent $i\in V$ is under majority illusion (\texttt{m} illusion) if $M_{N_i}$ is not a tie 
        containing any winner in  $M_V$, and $M_V$ is not a tie 
        containing any winner in  $M_{N_i}$, and $M_V\not = M_{N_i}$. A graph is in a \emph{Majority-majority illusion} (\texttt{Mm} illusion) if more than half of the agents are under majority illusion.
    \end{definition*}
    Clearly, this definition is equivalent to Definition \ref{def:general_illusions} for $\mathcal{R}=M$ (and to Definition \ref{def:maj-ill-strict} for only two colors). 
    
\end{remark}

\begin{remark}
    Instead of defining just `weak' and `strong' illusions as in Definition \ref{def:general_illusions}, one could argue that weak illusions come in many gradations. For example, if there is a lot of overlap between the local and global winner set, you are `less wrong' than if there is only very little overlap. Also in some situations it might matter whether your local winner set is a subset or a superset of the global winner set (whether there are more candidates that you think that win but do not globally win, or more candidates that you think loose that actually globally win).
\end{remark}
We can also generalize the definition for opposition and graph colorings  to arbitrary number of colors $k$ and arbitrary voting rule $\mathcal{R}$:
\begin{definition}[Weak $\mathcal{R}$ $k$-coloring]\label{def:weak_R_coloring}
    Given a colored graph $C=\langle V, E, c \rangle$, an agent $i\in V$ is under \emph{weak $\mathcal{R}$ opposition} if there is a color $c'$ such that $c'\not = c_i$, and $c'\in \mathcal{R}(N_i)$ (we could also just say: $\mathcal{R}(N_i)\not = \{c_i$\}). A \emph{weak $\mathcal{R}$ $k$-coloring} of a graph is a $k$-coloring such that all the nodes are under weak $\mathcal{R}$ opposition: for each $i\in V: \mathcal{R}(N_i)\neq \{c_i\}$.
\end{definition}

\subsection{Quota Rule Illusions ($q$-illusions)}\label{app:pq-illusions}
As a first example of illusions with another voting rule, we will consider
quota rules, which generalize majority. 
A color 
$c$ is a winner according to the $q$-quota rule iff more than a $q$-fraction of all nodes have color $c$. 

Even though all results except for Proposition \ref{prop:complete-p-q} in this section hold directly for more colors than just 2, for the ease of presentation we use just 2 colors here 
and consider $q\geq \frac{1}{2}$. For more than two colors, also lower values of $q$ could be considered. 
Below we specify Definitions \ref{def:general_illusions} and \ref{def:weak_R_coloring} for quota rules for $q\geq \frac{1}{2}$ These are not different from Definitions \ref{def:general_illusions} and \ref{def:weak_R_coloring}, but just easier to work with for quota rules.

\begin{definition}[$q$ illusion]\label{def:q-ill}
Given a 2-colored graph $ C = \langle V, E, c\rangle$, an agent $i\in V$ is under \emph{$q$ illusion} for $q\geq \frac{1}{2}$ if, for some $x\in \{red, blue\}$, 
\begin{itemize}
    \item $|\{j\in N_i \mid c_j = x\}|>q\cdot d_i$, but  $|\{j\in N \mid c_j = x\}|<q\cdot |V|$; or
    \item $|\{j\in N_i \mid c_j = x\}|<q\cdot d_i$, but  $|\{j\in N \mid c_j = x\}|>q\cdot |V|$.
\end{itemize}
\end{definition}


\begin{definition}[weak-$q$ illusion]\label{def:weak-q-ill}
Given a 2-colored graph $ C = \langle V, E, c\rangle$, an agent $i\in V$ is under \emph{weak-$q$ illusion} for $q\geq \frac{1}{2}$ if, for some $x\in \{red, blue\}$, 
\begin{itemize}
    \item $|\{j\in N_i \mid c_j = x\}|>q\cdot d_i$, but  $|\{j\in N \mid c_j = x\}|\leq q\cdot |V|$; or
    \item $|\{j\in N_i \mid c_j = x\}|\leq q\cdot d_i$, but  $|\{j\in N \mid c_j = x\}|>q\cdot |V|$.
\end{itemize}
\end{definition}

For $q=\frac{1}{2}$, Definitions \ref{def:q-ill} and \ref{def:weak-q-ill} reduce to Definitions \ref{def:maj-ill-strict} and \ref{def:maj-ill-weak}.
Additionally, note that for $q>\frac{1}{2}$,  Definitions \ref{def:q-ill} and \ref{def:weak-q-ill} are equivalent.

In the same network, different agents can be under a $q$ illusion with respect to different colors if $q >\frac{1}{2}$ (or a weak $q$ illusion if $q \geq\frac{1}{2}$), since then for both colors there can be less than (or equal to) $qn$ of the nodes in the network of that color, while in the different neighborhoods there can be more than $q$ of different colors.
We call a network illusion where all agents that are under illusion have an illusion of the same color (or all see a tie) a \emph{monochromatic} illusion, and the general case where different agents can have illusions of different colors a \emph{polychromatic} illusion or just an illusion (so the polychromatic illusion is more general and includes also the monochromatic one). 
Note that for $q =\frac{1}{2}$ a strict illusion can only be monochromatic.

We can also make the first `majority' in `majority-$q$ illusion' an arbitrary fraction of agents $p$ instead of exactly $\frac{1}{2}$, and we can define both the weak and strong version of this to study cases where either \emph{at least} $p$ of the agents are under illusion or \emph{more than} $p$ of the agents are under illusion. 
\begin{definition}[(weak)-$p$-(weak)-$q$ illusion]
A graph $ C = \langle V, E, c\rangle$ is in a \emph{$p$-(weak)-$q$ illusion} if more than a $p$ proportion of the agents is under (weak-)$q$ illusion. A graph is in a \emph{weak-$p$-(weak)-$q$ illusion} if at least a proportion $p$ of the agents is under (weak-)$q$ illusion. 
\end{definition}
We examine the possibility of $p$-$q$ illusions on complete graphs.
  
\begin{proposition}
\label{prop:complete-p-q}
    If $G=\langle V,E \rangle$ is a complete graph with $|V|= n$, 
    \begin{itemize}
        \item  a $p$-$q$ illusion is possible in $G$ iff there is an integer $x$
    , such that either 
        \begin{itemize}
            \item $q(n-1)<x\leq qn$ and $q(n-1)+1\geq n-x$ and $n-x>pn$; or 
            \item $q(n-1)+1\geq x > qn$ and $qn\geq n-x$ and $x>pn$
            . 
        \end{itemize} 
        \item A weak-$p$-$q$ illusion is possible in $G$ iff there is an integer $x$ such that either
        \begin{itemize}
            \item $q(n-1)<x\leq qn$ and $q(n-1)+1\geq n-x$ and $n-x\geq pn$; or 
            \item $q(n-1)+1\geq x > qn$ and $qn\geq n-x$ and $x\geq pn$
            .
        \end{itemize} 
        \item  A $p$-weak-$q$ illusion is possible in $G$ iff there is an integer $x$ such that either 
        \begin{itemize}
            \item $q(n-1) < x \leq qn$ and $n-x > pn$; or 
            \item $q(n-1) +1 \geq x > qn$ and $x > pn$
            .
        \end{itemize}
        \item  A weak-$p$-weak-$q$ illusion is possible in $G$ iff there is an integer $x$ such that  either
        \begin{itemize}
            \item $q(n-1) < x \leq qn$ and $n-x \geq pn$; or 
            \item $q(n-1) +1 \geq x > qn$ and $x \geq pn$
            .
        \end{itemize}
    \end{itemize}
    If there exists such $x$, we can generate the respective illusions by coloring exactly $x$ arbitrary nodes in the color of the illusion.
\end{proposition}

The proof of Proposition \ref{prop:complete-p-q}  consists of writing out and analysing the conditions under which the network is under the respective illusion, and can be found in Appendix \ref{app:propcompletepq}.

Proposition \ref{prop:c} can be written in a sligtly more general way:
\begin{proposition}\label{prop:c-general}\footnote{We can make a similar theorem for general $q$
, but then we only get that $d\leq n-2$ for the possibility of an illusion, which we already knew because for $d+1=n$ we have a complete network in which a (strict) illusion is not possible.}
        If a 2-colored $d$-regular graph $G=\langle V,E,c \rangle$ with $|V|=n$ is in a (weak-)$p$-majority illusion (for any $p>0$), then $n$ and $d$ must satisfy
        \begin{itemize}
            \item $d\leq n-4$  if $n$ and $d$ are even; 
            \item $d\leq n-3$  if one of $n$ and $d$ is even and one is odd.
        \end{itemize}
    \end{proposition}

\begin{proof}
    Assume 
    $G$ is in a (weak-)$p$-majority illusion. 
    Then there is at least one agent under \texttt{m} illusion, so more than $\frac{d}{2}$ of its neighbors have color $c$, but less than $\frac{n}{2}$ of the agents in the entire network have color $c$. Hence, this agent has at least $\frac{(d+1)}{2}$ neighbors of color $c$ if $d$ is odd and at least $\frac{(d+2)}{2}$ if $d$ is even. 
    But in total, there are less than $\frac{n}{2}$ nodes of color $c$, so at most $\frac{n-2}{2}$ if $n$ is even and at most  $\frac{n-1}{2}$ if $n$ is odd. Hence we must have: 
    \begin{itemize}
        \item if $n$ and $d$ are even:  $\frac{(d+2)}{2} \leq \frac{n-2}{2}$, so  $d\leq n-4$; 
        \item if $n$ is even and $d$ is odd: $\frac{(d+1)}{2} \leq \frac{n-2}{2}$, so  $d\leq n-3$;
        \item if $n$ is odd and $d$ is even:   $\frac{(d+2)}{2} \leq \frac{n-1}{2}$, so $d\leq n-3$.
    \end{itemize}
\end{proof}

Proposition \ref{prop:d} can be generalized for arbitrary $p$ and $q$ as follows:
   \begin{proposition}\label{prop:d-general}
    If a 2-colored $d$-regular graph $G=\langle V,E,c \rangle$ with $|V|=n$ is in a monochromatic $p$-$q$ illusion (for $0\leq p\leq 1$ and $0\leq q \leq 1$), then for any $n>1$, $d>1$, $0 < q\leq 1$, $p$ must satisfy $0\leq p<\frac{dn}{(n+1)(d+1)}$.
    \end{proposition}

\begin{proof}
    Assume 
    $G$ is in a monochromatic $p$-$q$ illusion, and assume that the color of illusion is $c$. Then more than $np$ of the nodes are under $q$ illusion, so at least $(n+1)p$ nodes are under illusion\footnote{The exact value can differ depending on the divisibility of $n$ by $\frac{1}{p}$, or maybe on the common divisors of $n$ and $p$, but it always has the lower bound of $(n+1)p$.}. One of the following two holds for all nodes under illusion (since the illusion is monochromatic):
    \begin{itemize}
        \item more than $qd$ of the neighbors are colored $c$ and less than $qn$ nodes in the network are colored $c$. Then at least $(d+1)q$ of the neighbors are colored $c$. Hence, there have to be at least $(n+1)p(d+1)q$ edges to a $c$-colored node, and since every node has $d$ edges, there have to be at least $\frac{(n+1)p(d+1)q}{d}$ nodes of color $c$. However, in the total network there are less than $qn$ nodes of color $c$, so $\frac{(n+1)p(d+1)q}{d} < qn$, which, for $q>0$, boils down to $p<\frac{dn}{(n+1)(d+1)}$. 
        \item less than $qd$ of the neighbors are colored $c$ and more than $qn$ nodes in the network are colored $c$. Then more than $(1-q)d$ of the neighbors are the other color $c'$, so at least $(1-q)(d+1)$ are colored $c'$.  This means that there have to be at least $(n+1)p(1-q)(d+1)$ edges to $c'$-colored nodes, so there have to be at least $\frac{(n+1)p(1-q)(d+1)}{d}$ nodes of color $c'$. However, since in total more than $qn$ nodes of the network are $c$, less than $(1-q)n$ nodes are colored $c'$. Therefore,  $\frac{(n+1)p(1-q)(d+1)}{d}<(1-q)n$, so $p<\frac{dn}{(n+1)(d+1)}$.
    \end{itemize}
\end{proof}

In conclusion, we defined $p$-$q$ illusions as a generalization of \texttt{Mm} illusions, we gave an example of a graph type (complete graphs) on which $p$-$q$ illusions are possible with certain constraints, and we showed that some of the propositions about \texttt{Mm} illusions can be generalized to $p$-$q$ illusions.

\subsection{Plurality Illusions}
The next voting rule 
we 
consider 
is plurality ($P$), 
which is a generalization of 
majority for 
more than 2 colors. 
With the plurality voting rule, the color that occurs most often wins. Plurality is not resolute: if there are multiple colors that occur equally often but more often than all other colors, they are all plurality winners. 
We write $P_{N_i}$ for the set of plurality winners in the neighborhood of agent $i$, and $P_{V}$ for the set of global plurality winners. 

Definitions \ref{def:plurality_illusion} and \ref{def:weak-plur-coloring} are specific cases of the general Definitions \ref{def:general_illusions} and~\ref{def:weak_R_coloring}.
\begin{definition}[Plurality illusion]\label{def:plurality_illusion}
    Given a colored graph $C=\langle V, E, c \rangle$, 
    an agent $i\in V$ is under \emph{plurality illusion} if 
    $P_{N_i}\cap P_V = \emptyset$ and not $P_{N_i} = P_V = \emptyset$. 
    An agent $i$ is under \emph{weak-plurality illusion} if $P(N_i) \not = P(V)$.
\end{definition}

`A plurality of agents is under illusion' does not have a clear meaning (since the property of being under illusion is binary), so something like `plurality-plurality illusion' is not well-defined. 
Instead, for the amount of agents under illusion we consider $q$-quota rules as in the previous subsection, with $q=\frac{1}{k}$. 
Intuitively, $\frac{1}{k}$-quota rules are connected to the plurality rule: if a color is a plurality winner among $k$ colors, it is also a weak-$\frac{1}{k}$-quota winner.
\begin{definition}[$\frac{1}{k}$-plurality illusion]
     A $k$-colored graph $C=\langle V, E, c \rangle$ with $|V|=n$ is in a \emph{$\frac{1}{k}$-plurality illusion} if more than $\frac{1}{k}$ of the agents are under plurality illusion.
\end{definition}
The weak versions are defined according to Definition \ref{def:general_illusions}.


\begin{remark}
    With only two colors, the definition of $\frac{1}{k}$-plurality illusions is equivalent to that of \texttt{Mm} illusions. Indeed, if $k=2$, any agent that is under plurality illusion must have $P_{N_i}\cap P_V = \emptyset$ (so $P_{N_i}$ and $P_V$ cannot be ties and must be different), so must be under \texttt{m} illusion. Furthermore, \emph{more than $\frac{1}{2}$} of the agents are under plurality illusion iff \emph{a majority} of agents is under plurality illusion.
\end{remark}

Analogous to weak majority 2-colorings (Definition \ref{def:weak_maj_2-col}), we can define weak plurality $k$-colorings (this is Definition \ref{def:weak_R_coloring} with $\mathcal{R} = P$).
\begin{definition}[Weak plurality $k$-coloring]\label{def:weak-plur-coloring}
    Given a colored graph $C=\langle V, E, c \rangle$, an agent $i\in V$ is under \textit{weak plurality opposition} if there is a color $c'$ such that $c'\not = c_i$, and a weak plurality (there may be a tie) of $i$'s neighbors is colored $c'$. A \emph{weak plurality coloring} of a graph is a coloring such that all the nodes are under weak plurality opposition: for each $i\in V: P_{N_i}\neq \{c_i\}$.
\end{definition}
\begin{example}
The 3-colored graph $C = \langle V, E, c \rangle$ with $V = \{A, B, C, D, E\}$ as shown in Figure \ref{fig:ex:plurality_illusion} can illustrate the notions defined above. In this graph, the global plurality winner $P_V =\{red\}$. Nodes $A$ and $B$ are under weak plurality illusion because $P_{N_A} = P_{N_B} = \{blue, green, red\} \not = P_V$. Node $D$ is under plurality illusion because $P_{N_D}=\{blue, green\}$ and $\{blue, green\}\cap P_V = \emptyset$. Therefore, more than $\frac{|V|}{k}=\frac{5}{3}$ nodes are under weak plurality illusion: $C$ is under $\frac{1}{3}$-plurality illusion. Furthermore, $c$ is a weak plurality coloring since for all $i\in V: P_{N_i}\not = \{c_i\}$.
     \begin{figure}
        \centering
        \includegraphics{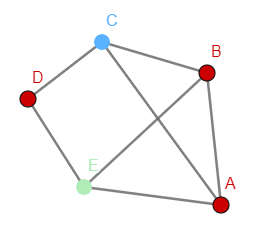}
        \caption{Example network with $\frac{1}{k}$-plurality illusion.}
        \label{fig:ex:plurality_illusion}
    \end{figure}
\end{example}

With the machinery defined above, we can generalize Theorem \ref{thm:maj-weak-maj} for $\frac{1}{k}$-plurality illusions. First we generalize Lemma \ref{lem:min-monochromatic}:

\begin{lemma}\label{lem:weak-plurality-coloring}
Let $G = \langle V,E\rangle$ 
be a graph, and let $c$ be a $k$-coloring of $G$ that minimizes the number of monochromatic edges. Then, $c$ is a weak plurality $k$-coloring of $G$.
\end{lemma}
\begin{proof}
    Let $E_{M}$ be the set of monochromatic edges 
    in graph $G$ colored by $c$. 
    Assume for contradiction that there is a node $i\in V$ that is an endpoint of strictly more monochromatic edges (we write $E_{M_i}$ for the set of such edges) than edges to any other color $d$. We call the set of edges from $i$ to color $d$ $E_{d_i}$, so we have that $|E_{M_i}|>|E_{d_i}|$ for any other color $d$.
Consider now a second $k$-coloring $c'$ of $G$ that only differs from $c$ with respect to $i$'s color, i.e., $c'$ assigns the same color as $c$ did to all nodes except for $i$: $c(i)\neq c'(i)$. Let us write $E'_M$ for the new set of monochromatic edges, and $E'_{M_i}$ and $E'_{d_i}$ for the new sets of monochromatic edges and edges to a color $d$ from $i$. Given that 
$|E_{c'(i)_i}|=|E'_{M_i}|$ and 
$|E_{M_i}|>|E_{c'(i)_i}|$ (
by construction
) we now have $|E_{M_i}|>|E'_{M_i}|$. Given that no other edge of the graph is affected by this change, the total number of monochromatic edges is smaller with coloring $c'$ than it was with $c$: $|E_M|>|E'_M|$. But since we started by assuming that $c$ was such that $|E_M|$ was minimal, this is a contradiction. 
\end{proof}

With the help of Lemma \ref{lem:weak-plurality-coloring} we can prove Theorem \ref{thm:weak_plurality} in a very similar way to Theorem \ref{thm:maj-weak-maj}:

\begin{theorem}\label{thm:weak_plurality}
    In any graph $G = \langle V,E \rangle$, a 
    $\frac{1}{k}$-weak-plurality illusion is possible.
\end{theorem}

\begin{proof}
  Let $G=\langle V,E \rangle$ be a graph and let $c$ be a $k$-coloring of $G$ 
  that minimizes the total number of monochromatic edges. By Lemma \ref{lem:weak-plurality-coloring}, $c$ is a weak plurality $k$-coloring of $G$. 
    There are two cases: 
    \begin{itemize}
        \item $|P_V|=1$ (there is only one plurality winner, no tie). 
        Assume w.l.o.g. that  $P_V=\{red\}$, so $|\{i\in V: c_i = red\}|>\frac{|V|}{k}$. Since $c$ is a weak plurality coloring, for any red vertex $i$, $P_{N_i}\neq \{red\}$, and therefore $P_{N_i}\neq P_V$. Hence,  more than $\frac{|V|}{k}$ of the nodes (all the red ones) are under (possibly weak)  plurality illusion: we have a $\frac{1}{k}$-weak-plurality illusion. 
        \item $|P_V| >1$. 
        There are two cases:
        \begin{itemize}
        \item  If $|\{i\in V: P_{N_i}\neq P_V\}| > \frac{|V|}{k}$, we have a $\frac{1}{k}$-weak-plurality illusion.
        \item  Otherwise,  
        (if $|\{i\in V: P_{N_i}= P_V\}| \geq \frac{|V|}{k}$)
        there are two cases:
            \begin{itemize}
                \item If there is a node $j$ with $P_{N_j} = P_V$ and $c_j\in P_V$: choose any such $j$ and define a new coloring $c'$ that is equal to $c$ for all nodes except for $j$:  $c'_j\neq c_j$. Since $j$ has the same number of neighbors of any color in $P_V$, this does not change the total number of monochromatic edges in the graph. Therefore,  $c'$ is also a coloring that minimizes this number. Hence, by Lemma \ref{lem:weak-plurality-coloring}, $c'$ is also a weak plurality $k$-coloring of $G$. Now, we have $P_V = \{c'_j\}$, and we can apply the logic of the first case: Assume w.l.o.g. that  $c'_j=red$. Since $c'$ is a weak plurality coloring, for any red vertex $i$, $M_{N_i}\neq red$. It follows that more than $\frac{|V|}{k}$ of the nodes has $P_{N_i}\neq P_V$:        we have a $\frac{1}{k}$-weak-plurality illusion. 
                \item If there is no such node, then for all nodes $j$ with $c_j\in P_V$  we have $ P_{N_j}\neq P_V$, so they are all under (weak) plurality illusion. But these are at least $2\frac{|V|}{k}$ nodes: $\frac{1}{k}$-weak-plurality illusion.
            \end{itemize}
        \end{itemize}
    \end{itemize}
\end{proof}


The results on \texttt{m} illusions with 2 colors in regular graphs (Propositions \ref{prop:c} and \ref{prop:d}) can be generalized to plurality illusions with multiple colors as well, but since the constraints become rather unreadable when the number of colors is to be added as a parameter, we decide not to include the details here. 

We finish this section with two propositions on plurality illusions in the most extreme types of regular graphs: simple cycles and complete graphs.
\begin{proposition}
\label{prop:cycles_plurality}
    In any connected 2-regular graph $G =\langle V, E\rangle$ (a simple cycle), a $\frac{1}{k}$-plurality illusion is possible iff $n>k\geq 3$ (where $k$ is the number of colors).
\end{proposition}
\begin{proof}
    By construction: Let $G =\langle V, E\rangle$ be a simple cycle, with $n\geq 4$, and let $n>k\geq 3$. Take a color $c$ and start by coloring one node $i$ in this color. Then walk around the circle while coloring every second node with the same color $c$, until there are just more than $\frac{1}{k}$ nodes colored $c$. Because of the constraints on $n$ and $k$, this is always possible. Divide the not-yet-colored nodes as evenly as possible over the remaining colors (the position of these nodes in the circle is irrelevant). 
   Now $P_V = \{c\}$, but every $c$-colored node (of which there are more than $\frac{1}{k}$) has only non-$c$-colored neighbors, and is therefore under plurality illusion: we have $\frac{1}{k}$-plurality illusion.
    Note that for $k=2$, we are in the case of \texttt{Mm} illusions, for which we already know that they are not possible on simple cycles, and for $k=n$ all nodes have a separate color, so there is a global tie over all colors and every node sees exactly two of those. 
\end{proof}

\begin{proposition}
\label{prop:complete_plurality}
    In any complete $k$-colored graph $C=\langle V, E, c\rangle$, no agent can be under plurality illusion.
\end{proposition}
\begin{proof}
    Since any node $i$ is connected to all nodes in $V$ except itself, we have that for all colors $c\in P_V\backslash\{c_i\}$, $c\in P_{N_i}$. Furthermore, if $c_i\in P_V$ but $c_i\notin P_{N_i}$, then there is another color $c'\not = c_i$ such that $|\{j\in V: c_j = c_i\}|-1 < |\{j\in V: c_j = c'\}|$
    , so $|\{j\in V: c_j = c_i\}| \leq |\{j\in V: c_j = c'\}|$: $c_i$ is not the only color in $P_V$. Hence, there must be at least one color $c'$ that is both in $P_V$ and in $P_{N_i}$: $i$ is not under plurality illusion.
\end{proof}

\begin{remark}
    With multiple colors, it could be interesting to study the case where colors can be divided into two groups (e.g. light/dark blue and light/dark red), because in real voting scenarios parties can sometimes be divided into two categories (e.g. left and right). 
    We could then study the relation between multiple-color illusions (plurality, for example) and majority-illusions.
    Our expectation is that there is less illusion with less colors, because people that were wrong about the winner when the categories were more precise can be correct if categories are combined. We leave this as a direction for future research.
\end{remark}

\subsection{Not Necessarily Irreflexive Networks}\label{sec:reflexive}

Instead of generalizing the definition of \texttt{Mm/Mw/Wm/Ww} illusions, we could also instead generalize the class of graphs. In this 
subsection 
we briefly consider a generalization to not necessarily irreflexive networks. We formalize the intuitions mentioned here in 
Appendix \ref{app:reflexive}. 

Adding reflexive loops does not influence the possibility of illusions as much as one would initially expect. However, it is slightly harder to have a \texttt{m} illusion on a graph with reflexive edges, since nodes have just more information about the true distribution of colors in the graph. If we have a strict illusion, adding any number of reflexive loops to the network
will preserve at least a weak illusion. For nodes with a high degree, intuitively adding an extra edge does not have a large influence on whether or not the node is under illusion. A node does not need many more than half of its neighbors to be of the minority color, to still be under illusion once a reflexive edge is added. 
Note also that weak illusions that are not strict are already on the edge of being no illusion because there is a tie involved: a difference of only one node (which can be caused by a reflexive edge) can change the tie into the correct majority. 
In this light it is not surprising that some results 
on irreflexive graphs
still hold if the class of graphs is extended to graphs that can have reflexive edges, but some do not. In particular, Lemma \ref{lem:2-col} and Propositions \ref{prop:2-col-maj-maj}, \ref{prop:2col-strict}, \ref{prop:k=2}, \ref{prop:complete-weak-maj-maj}\footnote{We assume `complete' means fully reflexive in this case. In fact, also \texttt{w} illusions are not possible on complete graphs, since all nodes see everything.}, \ref{prop:c}, \ref{prop:d}, 
\ref{prop:c-general}, and \ref{prop:d-general} still hold on graphs with reflexive edges. Lemmas \ref{lem:min-monochromatic} and \ref{lem:weak-plurality-coloring}, Theorems \ref{thm:maj-weak-maj} and \ref{thm:weak_plurality}, and Propositions  \ref{thm:degree-odd}, \ref{prop:complete-maj-weak-maj}, \ref{prop:complete-p-q},  \ref{prop:cycles_plurality}, and \ref{prop:complete_plurality} on the contrary, need the assumption that the graph is irreflexive. In Table \ref{tab:results_reflexive} in Appendix \ref{app:reflexive} we give an overview of which type of illusions hold in which class of graphs.

\section{Conclusion and Outlook}

\begin{table}[t]
\caption{The (im)-possibility of majority illusions (using two colors) on different classes of (irreflexive simple) graphs. \checkmark~ indicates that the illusion is possible on all graphs of the class, \xmark~  indicates that the illusion is not possible on any graph in the class, \checkmark / \xmark~ indicates that the illusion is possible on some but not all graphs of the class.  References to the relevant results are given. 
Note that for graphs with only odd-degree nodes and 2-colorable graphs the \texttt{Mw} illusion is always either a \texttt{Mm} illusion or \texttt{Uw} illusion, conform 
Proposition \ref{thm:degree-odd} and Lemma \ref{lem:2-col}, and that for complete graphs with an even number of nodes $|V|$, the \texttt{Mw} illusion is always a \texttt{Uw} illusion, conform Proposition \ref{prop:complete-maj-weak-maj}. 
}
   \centering
    \begin{tabular}{|p{0.41\linewidth}|p{0.10\linewidth}| p{0.15\linewidth}|p{0.12\linewidth}|}
    \hline
        \textbf{Class of graphs}  &  \multicolumn{1}{c|}{\textbf{\texttt{Mw} illusion}} & \textbf{\texttt{Wm} illusion}& \textbf{\texttt{Mm} illusion}  \\ \hline
         All graphs & \multicolumn{1}{c|}{\multirow[c]{8}{*}{\checkmark (Thm. \ref{thm:maj-weak-maj})}} & \multicolumn{2}{c|}{\checkmark / \xmark} \\ \cline{1-1} \cline{3-4}
         2-colorable graphs with $|V|$ odd  & & \multicolumn{2}{c|}{\checkmark~ (Prop. \ref{prop:2-col-maj-maj})} \\ \cline{1-1}  \cline{3-4}
         2-colorable graphs with $i\in V: \forall j\in N_i: d_j>2$  & & \multicolumn{1}{c|}{\multirow{2}{*}{\checkmark~ (Prop. \ref{prop:2col-strict})}}& \multicolumn{1}{c|}{\multirow{2}{*}{\checkmark / \xmark}}\\ \cline{1-1} \cline{3-4}
          2-regular graphs  & & \multicolumn{2}{c|}{\xmark~ (Prop. \ref{prop:k=2})}\\ \cline{1-1} \cline{3-4}
         Complete graphs   &  & \multicolumn{2}{c|}{\xmark~ (Prop. \ref{prop:complete-weak-maj-maj})} \\     \hline
    \end{tabular}
    \label{tab:results}
\end{table}


\paragraph{Conclusions} 
We studied weak and strong versions of the majority illusion
and some of its generalizations using analytical and computational methods.
Table \ref{tab:results} summarizes our analytical findings. 
Using 
 results 
about majority colorings, our main result shows that no network  
is immune to Majority-weak-majority (\texttt{Mw}) illusion (Theorem \ref{thm:maj-weak-maj}). The result indicates that one cannot exclude the possibility of illusions by only controlling how the network is wired.

We subsequently strengthened this result by showing that some specific classes of graphs are not immune to \emph{stronger} types of illusions either. The results on 2-colorable graphs (Propositions \ref{prop:2-col-maj-maj} and \ref{prop:2col-strict}) show that stronger illusions are even possible on all graphs in some  classes. 
 Even though these classes are admittedly artificial, they do provide insights into structural features of networks that impact the majority illusion. Similarly, the results on complete graphs (Proposition \ref{prop:complete-weak-maj-maj}) and on $2$-regular graphs (Proposition \ref{prop:k=2}) reveal some of the contours of the relation between connectivity within a graph and majority illusions: one naturally expects that when  agents have many connections (that is, a lot of information about other agents in the population), they are less likely to be under majority illusion. Indeed, in the limit case of complete graphs, strict illusions are \emph{not} possible (Proposition \ref{prop:complete-weak-maj-maj}). This conclusion aligns with the results in Section \ref{sec:simulations} for Erd\H{o}s-R\'{e}nyi graphs, where we observe that in more connected graphs, the fraction of nodes under illusion is lower.  At the other extreme, when connections are very few ($2$-regular graphs), we are able to make the same observation: strict illusions are not possible (Proposition \ref{prop:k=2}). It therefore appears that, for illusions to be possible, the network should neither exhibit too much nor too little connectivity.  

We also provided an algorithm to find a $d$-regular graph of size $n$ with a Majority-majority (\texttt{Mm}) illusion, when it 
exists.
With computational simulations we studied the likelihood and size of majority illusions on different types of random 
networks, given the parameters of the graph, and on a real example network. 
We proved that some of the results about the possibility or impossibility of majority illusions are generalizable to quota rule illusions and plurality illusions.


\medskip

\paragraph{Future work} A natural direction for further research is to broaden the scope of our study by considering other classes of graphs. One interesting class is that of directed networks, since many real social media networks are directed. 
Furthermore, in this paper we only considered finite graphs. Theoretically it could be interesting to study majority illusions on infinite graphs as models of unbounded networks, possibly with the notion of majority over infinite sets introduced by \cite{PacuitSalame2004}, or by using the results on unfriendly partitions by \cite{Aharoni1990, Bernardi1987, Shelah1990}. 

A different direction is to expand on the results in Section \ref{sec:simulations} on the likelihood of illusions in certain types of random graphs by studying asymptotic properties of those graphs in the line of \cite{ErdosRenyi}: if the number of nodes is large, can we theoretically analyse the probability of majority illusions to occur?

 Last but not least, it would be interesting to measure the impact of 
 majority/plurality/$q$-quota illusions on specific social phenomena. For instance, how do they affect opinion diffusion dynamics in a population? How do they interact with polling effects? And how do they relate to better known types of illusions, such as the above-mentioned `friendship paradox' \cite{Feld1991}?

\medskip
\noindent\textbf{Acknowledgments. }
We thank the anonymous reviewers of EUMAS 2023 and of JAAMAS for their helpful comments.
We also thank Naomï Broersma for pointing out some mistakes in an earlier version of the algorithms in Appendix \ref{app:algorithm}.

\section*{Declarations}


\textbf{Funding} Zo\'{e} Christoff acknowledges support from the project \href{https://zoechristoff.com/veni-nwo-research-project-2021-2024-social-networks-and-democracy/}{ Democracy on Social Networks} (VENI project number Vl.Veni.201F.032) financed by the Netherlands
Organisation for Scientific Research (NWO). 
Davide Grossi acknowledges support by the \href{https://hybrid-intelligence-centre.nl}{Hybrid Intelligence Center}, a 10-year program funded by the Dutch
Ministry of Education, Culture and Science through the Netherlands
Organisation for Scientific Research (NWO) and by the European Union under the Horizon Europe project \href{https://perycles-project.eu/}{Perycles} (Participatory Democracy that Scales).
\smallskip
\begin{center}
\includegraphics[width=0.5\textwidth]{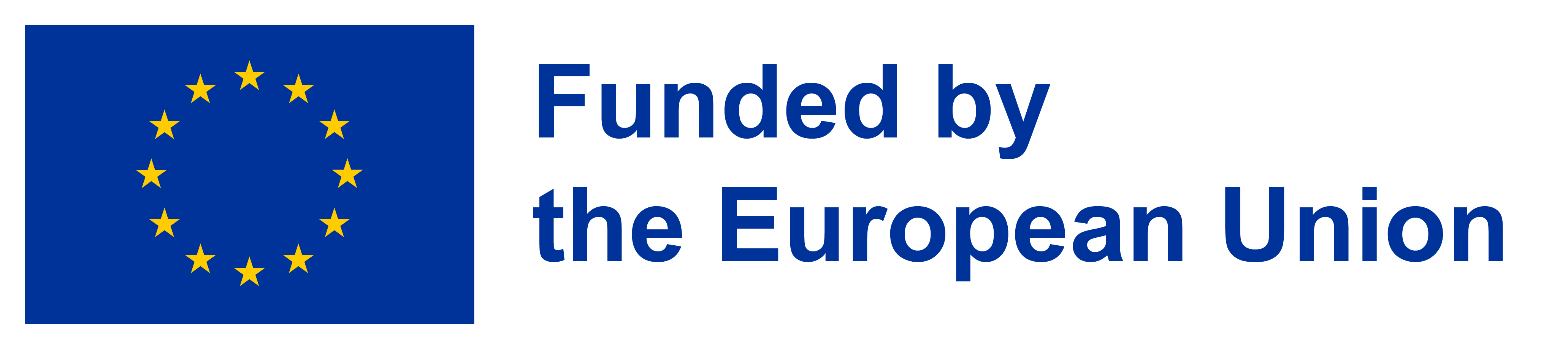}
\end{center}
\smallskip
\textbf{Competing interests} The authors have no competing interests to declare that are relevant to the content of this article.\\
\textbf{Ethics approval} Not applicable\\
\textbf{Consent to participate} Not applicable\\
\textbf{Consent for publication} Not applicable\\
\textbf{Availability of data and materials} The data generated in the computational experiments is available as supplementary material to this paper. Additionally, it is available upon request from the authors.\\
\textbf{Code availability} All code used for the computational experiments and data analysis is available on \href{https://github.com/MaaikeLos/Majority_Illusions}{https://github.com/MaaikeLos/Majority\_Illusions}\\
\textbf{Authors' contributions} All authors contributed to the study conception and design and the theoretical results. Computational experiments and data analysis were performed by Maaike Venema-Los. The manuscript was written by all authors together, and all authors read and approved the final manuscript.

\begin{appendices}

\section{Constructive Argument for Majority-2-Colorability}\label{app:additional}\label{app:proof-lovasz}

In 
Section \ref{sec:arbitrary-networks}, we mentioned that the result that every graph is weak majority 2-colorable is attributed to Lovász \cite{Lovász1966}, but described there in a much more general way, focusing on multigraphs. To make this paper self-contained, we provide a direct proof here. 
\begin{proposition}[attributed to \cite{Lovász1966}] 
Every graph is weak majority $2$-colorable.
\end{proposition}
\begin{proof}
We construct a weak majority 2-coloring following Algorithm \ref{Alg:maj2col}:
Consider a graph $G=\langle V, E \rangle$ and an arbitrary $2$-coloring $c$ of $G$. Then, search for a node with more monochromatic edges than dichromatic ones. If such a node does not exist, we are done, and our initial coloring $c$ is already a majority $2$-coloring. If such a node does exist, swap its color. 
Even though this might make another node increase its number of monochromatic edges, the total number of monochromatic edges in the graph can only decrease by such a step. 
We can proceed with this until there is no node to swap color of anymore, i.e. all nodes have at least as many dichromatic edges as monochromatic ones. We can be sure that at some point this will be the case, since every step can only strictly decrease the total number of monochromatic edges in the graph, and this number has a lower bound.
\begin{algorithm}
\caption{Weak majority 2-coloring}\label{Alg:maj2col}
 \hspace*{\algorithmicindent} \textbf{Input:} graph $G$\\
 \hspace*{\algorithmicindent} \textbf{Output:} weak majority $2$-coloring of $G$
 \begin{algorithmic}[1]
\State $c\gets$ random 2-coloring of $G$
\While{$\exists$ node $i$  with $E_M(i)>E_D(i)$}
    \State $c_i \gets \bar{c_i}$
\EndWhile
\State \Return{$c$}
\end{algorithmic}
\end{algorithm} 
\end{proof}

See Figure \ref{fig:maj-2-col} for an illustration of the algorithm. 
\begin{figure}[ht]
    \centering
    \begin{subfigure}{.25\textwidth}
        \centering\captionsetup{width=.8\linewidth}
        \includegraphics[width=0.8\textwidth]{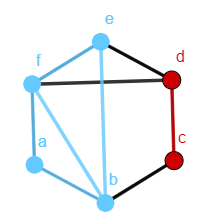}
        \caption{Initial coloring. Some nodes have more di- than monochromatic edges, for example $a$. Total nr. of monochromatic edges: 6}
    \end{subfigure}%
    \begin{subfigure}{.25\textwidth}
        \centering\captionsetup{width=.8\linewidth}
        \includegraphics[width=0.8\textwidth]{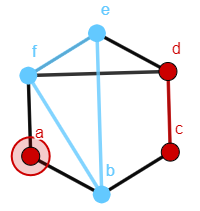}
        \caption{Result of swapping  node $a$'s color.  
        Node $e$ still has too many monochromatic edges. Total nr. of monochromatic edges: 4}
    \end{subfigure}%
    \begin{subfigure}{.25\textwidth}
        \centering\captionsetup{width=.8\linewidth}
        \includegraphics[width=0.8\textwidth]{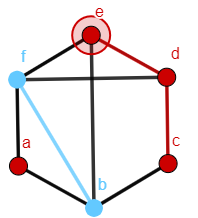}
        \caption{Result of swapping $e$'s color.  Even though this is bad for $d$, the total number of bad edges decreases. Total nr. of monochromatic edges: 3}
    \end{subfigure}%
    \begin{subfigure}{.25\textwidth}
        \centering\captionsetup{width=.8\linewidth}
        \includegraphics[width=0.8\textwidth]{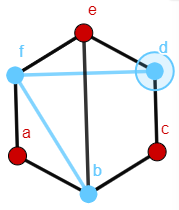}
        \caption{Result of swapping $d$'s color. Total nr. of monochromatic edges: 2. Now the graph is in a weak majority-2-coloring}
    \end{subfigure}%
    \caption[short]{Algorithm \ref{Alg:maj2col} executed on an example network.}
    \label{fig:maj-2-col}
\end{figure}

\section{Algorithm for Finding a $k$-Regular Graph with a Majority-Majority Illusion.}\label{app:algorithm}

\subsection{Pseudocode}
The pseudocode of the algorithm can be found in Algorithm \ref{Alg:main}. For an illustration of the algorithm, see Figure \ref{fig:algorithm-12-6}.
Note that this algorithm will not always return a connected graph. However, if required one can easily change the resulting graph in a connected graph by rewiring edges between red nodes: in every connected component break the a between two red nodes and then add links between the red nodes of different component. Proof that this always works is left as an exercise to the reader.


\begin{algorithm}
\caption{Finding a \texttt{Mm} illusion coloring}\label{Alg:main}
 \hspace*{\algorithmicindent} \textbf{Input:} number of nodes $n$, degree $k$\\
 \hspace*{\algorithmicindent} \textbf{Output:} $k$-regular graph of $n$ nodes, with \texttt{Mm} illusion
 \begin{algorithmic}[1]

\State $n_{red}\gets \lfloor{\frac{n}{2}+1}\rfloor$
\State $n_{blue}\gets \lceil{\frac{n}{2}-1}\rceil$
 \If{$n$ is even}
     \If{$k$ is even}
         \State $k_{blue} \gets \frac{k}{2}-3$
         \State $k_{red} \gets \frac{k}{2}-1$
     \Else[$k$ is odd]
         \State $k_{blue} \gets \frac{k-5}{2}$
         \State $k_{red} \gets \frac{k-1}{2}$
     \EndIf
     \If{$k_{blue}<0$}
         \State $k_{blue}\gets 0$
     \EndIf
 \Else[$n$ is odd]
     \State $k_{blue} \gets \frac{k}{2}-2$
     \State $k_{red} \gets \frac{k-2}{2}$
 \EndIf
 \Comment{Now $n_{blue}$ and $n_{red}$ are the numbers of blue and red nodes,  and $k_{blue}$ and $k_{red}$ are the degrees of the regular subgraphs built in the last step.}
 \State $N\gets $ set of $n$ nodes
 \State $E\gets \emptyset$ \Comment{set of edges starts empty}
\State $G \gets \langle N,E\rangle$
\State $R \gets $ any subset of $n_{red}$ nodes of $N$
\State $B \gets N\backslash B$
\State $E\gets $ add\_initial\_edges$(E, B, R, k)$\Comment{Algorithm \ref{alg:rounds}}
\State $E\gets $ add\_extra\_blue\_edges$(E, B, k, k_{blue})$\Comment{Algorithm \ref{alg:extra-blue}}
\If{$k_{red}$ is even \textbf{or} $n_{red}$ is even}
    \State $E\gets $add\_regular\_subgraph$(E, R, k_{red})$\Comment{Algorithm \ref{alg:add-regular-subgraph}}
\EndIf
\If{$k_{blue}$ is even \textbf{or} $n_{blue}$ is even}
    \State $E \gets$ add\_regular\_subgraph$(E, B, k_{blue})$
\EndIf
\If{ $k_{red}$ is odd \textbf{and} $n_{red}$ is odd}\Comment{(Hence also $k_{blue}$ and $n_{blue}$ are odd)}
    \State $k_{red_{temp}}\gets  k_{red} -1$ \Comment{Make regular graphs with one less edge per node}
    \State $E\gets$ add\_regular\_subgraph$(E, R, k_{red_{temp}})$
    \If{ $k_{blue} >0$}
        \State $k_{blue_{temp}}\gets  k_{blue} -1$
        \State $E\gets$ add\_regular\_subgraph$(E, B, k_{blue_{temp}})$
    \EndIf
    \State $blue_1 \gets$ blue node with least amount of edges
    \State $red_c \gets$ a red node that $blue_1$ is not yet connected to
    \State $E\gets E\cup \{(blue_1, red_c)\}$ \Comment{add an edge between ($blue_1$ and $red_c$)}
    \If{$k_{blue}>0$}
        \For{node in $B\backslash \{blue_1\}$ }
            \If{number of neighbors of node $<k$}
                \State $next \gets$ first blue node in $B\backslash \{blue_1\}$ that is not yet connected to $node$ and has less than $k$ edges
                \State $E\gets E\cup \{(node, next)$\}
            \EndIf
        \EndFor
    \EndIf
    \For{node in $R\backslash \{red_c\}$}
        \If{number of neighbors of node $<k$}
            \State $next \gets$ first red node in $R\backslash \{red_c\}$ that is not yet connected to $node$ and has less than $k$ edges
            \State $E\gets E\cup \{(node, next)$\}
        \EndIf
    \EndFor
\EndIf
\State $G \gets \langle N, E \rangle$
\State \Return{G}
 \end{algorithmic}
\end{algorithm}

\begin{algorithm}
\caption{add\_initial\_edges}\label{alg:rounds}
 \hspace*{\algorithmicindent} \textbf{Input:} set of edges $E$, blue nodes $B$, red nodes $R$, degree $k$\\
 \hspace*{\algorithmicindent} \textbf{Output:} $E$ with some edges added between blue and red nodes
 \begin{algorithmic}[1]
\State $x\gets 0$ \Comment{value to avoid double edges}
\If{$k$ is even}
     \State $ n_{edges} \gets \lfloor|R|\cdot\frac{k+2}{2}\rfloor$ \Comment{$n_{edges}$ is the nr of necessary edges}
\Else
    \State $n_{edges} \gets \lfloor|R|\cdot\frac{k+1}{2}\rfloor$
\EndIf
\For{$0\leq i <  n_{edges}$}
    \State $index_{red} \gets  i \mod |R|$
    \State $node_{red} \gets R[index_{red}]$

    \State $index_{blue} \gets (x + i) \mod |B|$
    \State $node_{blue} \gets B[index_{blue}]$

    \If{($node_{red}, node_{blue})\in E$}
        \State $x\gets 1$ \Comment{(edge already exists, we shift one position)}
        \State $index_{blue} = (x + i) \mod |B|$
        \State $node_{blue}\gets B[index_{blue}]$
    \EndIf
    \State $E\gets E\cup \{(node_{red}, node_{blue})\}$
\EndFor
\State \Return{$E$}
\end{algorithmic}
\end{algorithm}

\begin{algorithm}

\caption{add\_extra\_blue\_edges}\label{alg:extra-blue}
\hspace*{\algorithmicindent} \textbf{Input:} set of edges $E$, blue nodes $B$, degree $k$, degree $k_{blue}$\\
\hspace*{\algorithmicindent} \textbf{Output:} $E$ with extra blue edges
 \begin{algorithmic}[1]
\State sort $B$ such that ones missing most edges come first
\For{node in $B$}
    \State $index_{next} \gets (B.index (node)+1) \mod |B|$
    \State $next \gets B[index_{next}]$
    \If{(node, next\_node)$\notin E \wedge  |\text{neighbors(node)}| < k-k_{blue} \wedge |\text{neighbors}(next)| <k-k_{blue}$}\Comment{The node has more than $k_{blue}$ open edges (which is the nr to be left over after this), so less than $k-k_{blue}$ neighbors (the next-node check is only for the case where one node has to be left open (n odd, k even))}
        \State $E\gets E\cup \{(node, next)\}$
    \EndIf
\EndFor
\State \Return{$E$}
\end{algorithmic}
\end{algorithm}

\begin{algorithm}
\caption{add\_regular\_subgraph}\label{alg:add-regular-subgraph}
\hspace*{\algorithmicindent} \textbf{Input:} set of edges $E$, colored nodes $C$, degree $k_{color}$\\
\hspace*{\algorithmicindent} \textbf{Output:} $E$ with added regular subgraph between nodes of $C$
\begin{algorithmic}[1]
\For{node in $C$}
    \If{$|C|$ is even}
        \State $index_{start}\gets ( C.index(node)+\lfloor\frac{|C|}{2}\rfloor) \mod |C|$ \Comment{The index of the node opposite to the current node}
    \Else
        \State $index_{start}\gets C.index(node)+\frac{|C|}{2}$ \Comment{this will be a half number}
    \EndIf
    \If{$|C|$ is even \textbf{ and } $k_{color}$ is odd}
        \State $E\gets E\cup (node, C[index_{start}])$ \Comment{add an edge to the node opposite}
    \EndIf
    \State $r\gets \lfloor\frac{k_{color}}{2}\rfloor$ \Comment{$r$ is the range to both sides}
    \If{$|C|$ is odd \textbf{ and } $k_{color}$ is odd}
        \State $r\gets \frac{k_{color}-1}{2}$
    \EndIf
    \For{$1\leq i \leq r$}
        \State $index_{minus} \gets \lceil index_{start} - i\rceil \mod |C|$ 
        \State $index_{plus} \gets \lfloor(index_{start} + i)\rfloor  \mod |C|$
        \State $E\gets E\cup \{(node, C[index_{minus}]), (node, C[index_{plus}])\}$
    \EndFor
\EndFor
\State \Return{$E$}    
\end{algorithmic}
\end{algorithm}

\subsection{
Proof of Correctness of the Algorithm
}
\begin{proposition}
If $n$ and $d$ are  two integers such that  $d>2$ and $d$ or $n$ is even, which satisfy the conditions of Propositions \ref{prop:c} and \ref{prop:d}, Algorithm \ref{Alg:main} finds a
    $d$-regular 2-colored graph $C=\langle V, E, c\rangle$ with $|V|=n$ which is under \texttt{Mm} illusion.
\end{proposition}

\begin{proof}
 We go by cases over the combinations of the possible parities of $|V|=n$ and $k$.
 \begin{itemize}
     \item  \textbf{$n$ and $k$ are even}
 For the case where $n$ and $k$ are both even: Let there be $n_{red}=\frac{n}{2} + 1$ red nodes and $n_{blue}=\frac{n}{2}-1$ blue nodes.  We assume that the requirements $n-4\geq k$ (and $n\geq \frac{2(3k+2)}{k-2}$) hold.

 We start by just connecting every red node to $\frac{k+2}{2}$ blue nodes, and do this as evenly spread out over the blue nodes as possible by just going in rounds over the red nodes and over the blue nodes until we are done. After this, all red nodes have enough blue neighbors to be in \texttt{m} illusion, and there are in total $\frac{n+2}{2}\frac{k+2}{2}$ blue edges used, so there are $k(\frac{n}{2}-1) - \frac{n+2}{2}\frac{k+2}{2}$ blue edge-ends left over, evenly spread over the blue nodes. We can write this as $k(\frac{n}{2}-1)-(\frac{n}{2}-1+2)\frac{k+2}{2} = (\frac{n}{2}-1)(k-\frac{k+2}{2})-(k+2) = (\frac{n}{2}-1)(\frac{k-2}{2})-(k+2)$. Hence, every blue node has $\frac{k-2}{2}$ edge-ends left over, but $k+2$ of them have one less left over. But, $k+2$ may exceed the number of blue nodes $\frac{n}{2}-1$, namely when $n<2k+6$, and then there must be some nodes that have two less (because we cannot spread out the $k+2$ evenly over the $\frac{n}{2}-1$. However, $k+2$ cannot be bigger than $2(\frac{n}{2}-1)$ (since then we would have $k>n-4$, which is contradicted by our first assumption), so there are no blue nodes with more than 2 less open ends than  $\frac{k-2}{2}$.
 We can therefore approach it from the other side, if we start with stating that every blue node has $\frac{k-2}{2}-2 = \frac{k-6}{2}$ edge-ends left open, then $(\frac{n}{2}-1)2-(k+2)=n-k-4$ have another extra left open (and if $(\frac{n}{2}-1)2-(k+2) >\frac{n}{2}-1$ some have two extra, that is, when $n>2k+6$).
 So, we have three kinds of open ends left over. 
 \begin{enumerate}[(a)]
     \item $\frac{n}{2}+1$ red nodes that all have $k_{red}=\frac{k}{2}-1$ open ends;
     \item $\frac{n}{2}-1$ blue nodes that all have $k_{blue}=\frac{k-6}{2}=\frac{k}{2}-3$ open ends 
     \item another $n-k-4$ open ends evenly divided over the blue nodes (where each blue node has at most two ends of this kind).
 \end{enumerate}

Now we can link those open edges to make a regular graph:
\begin{enumerate}[i]
    \item we can tie the ends in (c) to each other (possible since $n-k-4$ is even, and no blue is yet connected to another blue);
    \item if $k_{blue} = \frac{k}{2}-3$ (and hence also  $k_{red}=\frac{k}{2}-1$) is even: 
    \begin{enumerate}
        \item make a regular graph between the remaining blue edge ends ($\frac{n}{2}-1$ nodes with all $\frac{k}{2}-3$ open ends);
        \item make a regular graph between the remaining red edge ends ($\frac{n}{2}+1$ nodes with all $\frac{k}{2}-1$ open ends);
    \end{enumerate} 
    Where 2 (a) is possible because any blue node is connected to at most two other blue nodes as a result of 1, so there are at least $(\frac{n}{2}-1)-3$ other blue nodes left over (also subtracting 1 for the node itself). Since $n-2\geq k$, we also have $(\frac{n}{2}-1)-3\geq \frac{k}{2}-3$. If $k_{blue}<0$, we just skip this step, because then all blue nodes already have enough edges due to step (i). Clearly, 2(b) is always possible, since there are no edges between any red nodes yet.
    \item if $k_{blue}$ and hence also  $k_{red}$ are odd, 
    \begin{enumerate}
        \item make a regular graph between the remaining blue edge ends leaving one edge per node open ($\frac{n}{2}-1$ nodes with all $\frac{k}{2}-4$ ends);
        \item make a regular graph between the remaining red edge ends, leaving one edge per node open ($\frac{n}{2}+1$ nodes with all $\frac{k}{2}-2$ ends);
    \end{enumerate}
    Then we are left with exactly one open end for all nodes, but since everything so far was spread out as evenly as possible over the network, it is possible to divide the nodes into pairs that are not yet connected to each other.
\end{enumerate}

The cases where one of $k$ and $n$ is odd work very similar:
\item \textbf{$n$ is even, $k$ is odd}
When $n$ is even and $k$ is odd, we still have $n_{red}=\frac{n}{2}+1$ red nodes and $n_{blue}=\frac{n}{2}-1$ blue. We assume that, as given in Proposition \ref{prop:c}, $k\leq n-3$, and, as given in Prop. \ref{prop:d}, $n\geq \frac{2(3k+1)}{k-1}$.
Then we proceed as in Algorithm 2 above: we tie every red node to $\frac{k+1}{2}$ blue nodes, as evenly as possible. Then, there are still $k(\frac{n-2}{2})-\frac{n+2}{2}\frac{k+1}{2} = \frac{k-1}{2}\frac{n-2}{2}-(k+1)$ blue edge-ends left over. Since these are as evenly spread over the blue nodes as possible, every blue node has at least $\frac{k-1}{2}-2=\frac{k-5}{2}$ open ends, and then there are another $\frac{n-2}{2}\cdot 2 -(k+1) = n-k-3$ blue edge-ends left over. (if $n-k-3>\frac{n}{2}-1$, some blue nodes have two extra, otherwise those can be divided over $n-k-3$ blue nodes that all have one extra). Then, all red nodes have a \texttt{m} illusion, and we still have the following open ends to fill the regular graph: 
\begin{enumerate}[a)]
    \item $\frac{n}{2}+1$ red nodes, all with $k_{red}=\frac{k-1}{2}$ open ends;
    \item $\frac{n}{2}-1$ blue nodes, all with $k_{blue}=\frac{k-5}{2}$ open ends to start with;
    \item another $n-k-3$ blue open ends, at most 2 per node.
\end{enumerate}
Since $n-k-3$ is even, and no blue is yet connected to another blue, we can just connect the open ends in (c) to each other. 
Then, if $\frac{k-1}{2}$ and $\frac{k-5}{2}$ are even, just make two regular graphs between the two, otherwise leave one open and connect those, just as in the case with both $n$ and $k$ even.
\item \textbf{$n$ is odd, $k$ is even}
When $n$ is odd and $k$ is even, we have $n_{red}=\frac{n+1}{2}$ red nodes, $n_{blue}=\frac{n-1}{2}$ blue nodes, and a node needs to be connected to $\frac{k}{2}+1$ blue nodes to have an illusion. We assume that $k\leq n-3$ 
conformingly to Proposition \ref{prop:c}, and that $n\geq \frac{3k+2}{k-2}$ 
conformingly to Proposition \ref{prop:d}. 
Just as in the other cases, we start by tying every red node to $\frac{k}{2}+1$ blue nodes, which leaves us with $k(\frac{n-1}{2})-\frac{n+1}{2}\frac{k+2}{2} = \frac{k-2}{2}\frac{n-1}{2}-\frac{k+2}{2}$ blue ends open. 
It is not possible that $\frac{k+2}{2}>\frac{n-1}{2}$, since that would imply that $k>n-3$, so the remaining $\frac{k+2}{2}$ can be divided over the blue nodes. What this means is that every blue node has $\frac{k-2}{2}-1 = \frac{k}{2}-2$ open ends at least, and $\frac{n=-1}{2}-\frac{k+2}{2} = \frac{n-k-3}{2}$ blue nodes have one more open end (so those have $\frac{k-2}{2}$).
This means that we still have the following open ends to fill the regular graph: 
\begin{enumerate}[(a)]
    \item $\frac{n+1}{2}$ red nodes, all with $k_{red}=\frac{k-2}{2}$ open ends;
    \item $\frac{n-1}{2}-\frac{n-k-3}{2} = \frac{k+2}{2}$ blue nodes that have $k_{blue}=\frac{k}{2}-2$ open ends; and
    \item $\frac{n-k-3}{2}$ blue nodes that have  $\frac{k}{2}-1$ open end.
\end{enumerate}
If the number of nodes mentioned in (c) is even, we can connect them to each other and make regular graphs between the remaining red resp. blue nodes. 
If the number mentioned in (c) is odd, then either both  $\frac{n+1}{2}$ ($n_{red}$) and $\frac{k-2}{2}$ ($k_{red}$) are even or both $\frac{n-1}{2}$ ($n_{blue}$) and $\frac{k}{2}-2$ ($k_{blue}$) are even, but not both. That means that we can make a regular graph of the remaining edges of either the blue nodes or the red nodes, but not both. However, of the one that cannot form a regular graph, we tie one node to one of the nodes in $c$, which makes the remaining number in $c$ even, and of the rest, we make an almost regular graph.
 \end{itemize}
    
\end{proof}

 \paragraph{Faster algorithm (only for nice values of $n$ and $k$).} 
 If $n\equiv 2\mod 4$ and $n\leq 2k-2$, we can connect all red nodes to all blue nodes. Then all red nodes still have $k-\frac{n}{2}+1$ open ends, and all blue nodes still have $k-\frac{n}{2}-1$ open ends. Since $n\equiv 2\mod 4$ and $k$ is even, both $k-\frac{n}{2}+1$ and $k-\frac{n}{2}-1$ are even. Hence, we can make two regular graphs with the remaining open ends, one between the red nodes, and one between the blue nodes.
 Note that if $n\equiv 0\mod 4$, both the numbers of open ends of each color and the number of nodes of each color is odd, so we cannot apply this algorithm.

\section{Proof of Proposition \ref{prop:complete-p-q}}\label{app:propcompletepq}
\noindent\textbf{Proposition \ref{prop:complete-p-q}} ($p$-$q$ illusions on complete graphs)\textbf{.}
\emph{
    If $G=\langle V,E \rangle$ is a complete graph with $|V|= n$, 
    \begin{itemize}
        \item  a $p$-$q$ illusion is possible in $G$ iff there is an integer $x$
    , such that either 
        \begin{itemize}
            \item $q(n-1)<x\leq qn$ and $q(n-1)+1\geq n-x$ and $n-x>pn$; or 
            \item $q(n-1)+1\geq x > qn$ and $qn\geq n-x$ and $x>pn$
            .
        \end{itemize} 
        \item A weak-$p$-$q$ illusion is possible in $G$ iff there is an integer $x$ such that either
        \begin{itemize}
            \item $q(n-1)<x\leq qn$ and $q(n-1)+1\geq n-x$ and $n-x\geq pn$; or 
            \item $q(n-1)+1\geq x > qn$ and $qn\geq n-x$ and $x\geq pn$
            .
        \end{itemize} 
        \item  A $p$-weak-$q$ illusion is possible in $G$ iff there is an integer $x$ such that either 
        \begin{itemize}
            \item $q(n-1) < x \leq qn$ and $n-x > pn$; or 
            \item $q(n-1) +1 \geq x > qn$ and $x > pn$
            .
        \end{itemize}
        \item  A weak-$p$-weak-$q$ illusion is possible in $G$ iff there is an integer $x$ such that  either
        \begin{itemize}
            \item $q(n-1) < x \leq qn$ and $n-x \geq pn$; or 
            \item $q(n-1) +1 \geq x > qn$ and $x \geq pn$
            .
        \end{itemize}
    \end{itemize}
    If there exists such $x$, we can generate the respective illusions by coloring exactly $x$ arbitrary nodes in the color of the illusion.
}

\begin{proof}
    For a $p$-weak-$q$ illusion to be possible in $G$, we need that for more than a $p$-fraction of the nodes, more than $q\cdot k$ of its neighbors has a color $c$, but that at the same time at most $q\cdot n$ of the set of all nodes has color $c$, which can only happen if the node itself does not have color $c$; or that at most $qk$ of the neighbors has color $c$ while more than $qn$ of all nodes has color $c$, which can only happen if the node itself has color $c$. Hence, if we call the total number of $c$-colored nodes in the network $x$, we need to have that $q(n-1)<x\leq qn$ (if the node is not $c$) or $q(n-1)+1\geq x>qn$ (if the node is $c$). 
    If there is $x$ such that $q(n-1)<x\leq qn$, in fact all nodes that are not color $c$ are under weak-$q$ illusion, and if there is $x$ such that $q(n-1)+1\geq x>qn$, all nodes that are color $c$ are under weak-$q$ illusion. Hence, since $x$ represented the number of nodes of color $c$, there is a $p$-weak-$q$ illusion iff there is an integer $x$ such that $q(n-1)<x\leq qn$ and $n-x>pn$ (or $n-x\geq pn$ for weak-$p$-weak-$q$ illusion), or there is an integer $x$ such that $q(n-1)+1\geq x>qn$ and $x>pn$ (or $x\geq pn$ for weak-$p$-weak-$q$ illusion).
    
    For the strict $q$ illusions the requirements are slightly more complicated. If a node $i$ has color $c$ and there are $x$ nodes of color $c$, 
    the only option for $i$ to be under $q$ illusion with color $c$ is that $c$ is a $q$-quota winner globally but not locally. Then the other color cannot be a $q$-quota winner globally, but it can either be or be not a $q$-quota winner locally. 
    Therefore, we need that $q(n-1)+1\geq x>qn$ ($c$ is not a winner locally but it is a winner globally), and either $q(n-1)<n-x\leq qn$ (the other color is a winner locally but not globally) or both $q(n-1)+1\geq n-x$ and $qn\geq n-x$ (the other color is neither a winner locally nor globally). These requirements combine to the requirement
        \begin{equation}\label{eq:1}
            q(n-1)+1\geq x>qn \text{ and } qn\geq n-x.
        \end{equation}
    If $i$ has the other color than $c$, it can only be under $q$-illusion with color $c$ if $c$ is a $q$-quota winner locally but not globally. Then the other color cannot be a $q$-quota winner locally, but it can either be or be not a $q$-quota winner globally. Therefore, we need that $q(n-1)< x\leq qn$ ($c$ is a winner locally but it is not a winner globally), and either $q(n-1)\geq n-x > qn$ (the other color is not a winner locally but it is globally) or both $q(n-1)+1\geq n-x$ and $qn\geq n-x$ (the other color is neither a winner locally nor globally). These requirements combine to the requirement 
    \begin{equation}\label{eq:2}
    q(n-1) < x \leq qn \text{ and } q(n-1)+1\geq n-x.         
    \end{equation}
    If there is an $x$ that satisfies equation \ref{eq:1}, all nodes of color $c$ are under $q$ illusion, so if \ref{eq:1} and $x>pn$, there is a $p$-$q$ illusion (and if \ref{eq:1} and $x\geq pn$, there is a weak-$p$-$q$ illusion).
    If there is an $x$ that satisfies equation \ref{eq:2}, all nodes that are not color $c$ are under $q$ illusion, so if \ref{eq:2} and $n-x>pn$, there is a $p$-$q$ illusion (and if \ref{eq:2} and $n-x\geq pn$, there is a weak-$p$-$q$ illusion).
\end{proof}

\section{
Not Necessarily Irreflexive Networks}\label{app:reflexive}
In this appendix we study concisely the possibility of majority illusions on graphs that are not necessarily irreflexive%
: graphs in which nodes can have an edge to themselves. 
In this section, if we mention a graph $G$ or a colored graph $C$, we do not assume that it is irreflexive as in the 
main paper. A self-loop counts as only one edge for a node, for example, a node with only an edge to itself has degree 1. 

As mentioned in 
Section \ref{sec:reflexive}, adding reflexive loops to a graph only lightly influences the possibility of illusions on the graph. Here we give some formal statements displaying the size of this effect.

First, 
if we have a strict illusion, 
adding any number of reflexive loops to the network 
will preserve at least a weak illusion:

\begin{proposition}\label{prop:strict_weak}
    If a \texttt{Mm/Wm} illusion is possible on an irreflexive network $G'=\langle V, E'  \rangle$, then on the network $G =\langle V, E \rangle$ where $E = E' \cup R$ and $R$ is a set of reflexive edges on $G$, a \texttt{Mw/Ww} illusion is possible.
\end{proposition}
\begin{proof}
    If a \texttt{Mm/Wm} illusion is possible on $G'$, then at least / more than half of the nodes in $V$ are under \texttt{m} illusion. Assume without loss of generality that $M_V = red$ (it cannot be a tie, since there are nodes under \texttt{m} illusion). When we add only reflexive edges to $G'$ to obtain $G$, all the nodes $i$ that were under \texttt{m} illusion in $G'$ (so that had $M_i=blue$), get at most one more red `neighbor', so now $M_i\in\{blue, tie\}$, so they are all under \texttt{w} illusion.
\end{proof}
Note that this does not hold in the other direction: if a graph with reflexive edges is under \texttt{Mw/Ww} illusion and we remove the reflexive edges, the graph is not necessarily under \texttt{Mm/Wm} illusion.

For nodes with a high degree, intuitively adding an extra edge does not have a large influence on 
whether or not the node is under illusion. A node does not need many more than half of its neighbors to be of the minority color, to still be under illusion once a reflexive edge is added. This idea is captured in the following proposition.

\begin{proposition}\label{prop:degree_irreflexive}
    Let  $C'=\langle V, E', c\rangle$ be 
    an irreflexive 2-colored graph 
    where $M_V \not = tie$.
    If a node $i$ with degree $d'_i$ has more than a $q$-fraction of neighbors of the global minority color for $q = \frac{1}{2} + \frac{1}{d'_i}$ if $d'_i$ is even or $q = \frac{1}{2}+\frac{3}{2d'_i}$ if $d'_i$ is odd, then on any colored graph $C = \langle V, E, c \rangle$ where $E$ is $E'$ with any number of reflexive edges added, $i$ is under \texttt{m} illusion.    
\end{proposition}
\begin{proof}
    Take any irreflexive colored graph $C'=\langle V, E', c\rangle$, and consider any node $i$ with even degree $d'_i$ that has more than $q$ neighbors of the global minority color for $q = \frac{1}{2} + \frac{1}{d'_i}$ if $d'_i$ is even or $q = \frac{1}{2}+\frac{3}{2d'_i}$ if $d'_i$ is odd.
    \begin{itemize}
        \item \textbf{Case $d'_i$ is even:} 
        In $C'$, $i$ has more than a fraction of $ \frac{1}{2} + \frac{1}{d'_i}$ neighbors of the global minority color, assume w.l.o.g. this is red. 
        Now consider a graph $C = \langle V, E, c \rangle$ where $E$ is $E'$ with any number of reflexive edges added. Clearly, the total number of red nodes has not changed, so also in $C$, red is the minority color. In this graph, $i$ can have a loop to itself, so it's degree $d_i$ is either $d'_i$ or $d'_i +1$. We do not know whether $i$ itself is red or not, but we do know that its neighbors have the same colors as in $C'$, so $i$ has more than $\frac{d'_i}{2} + \frac{d'_i}{d'_i} = \frac{d'_i}{2} +1$  red neighbors. 
        If $d_i =d'_i$, then $d_i$ is even and $i$ has more than $\frac{d_i}{2} +1 $ red neighbors while less than half of all nodes in $C$ are red: $i$ is under \texttt{m/w} illusion. 
        If $d_i = d'_i+1$, then $d_i$ is odd and $i$ has more than $\frac{d_i-1}{2} +1 = \frac{d_i +1}{2}$ red neighbors while less than  half of all nodes in $C$ are red: $i$ is under \texttt{m} illusion. 
        
        \item \textbf{Case $d'_i$ is odd:} 

        In $C'$, more than a fraction of $\frac{1}{2}+\frac{3}{2d'_i}$ of $i$'s  neighbors are of the global minority color, assume w.l.o.g. this is red. 
        Now consider a graph $C = \langle V, E, c \rangle$ where $E$ is $E'$ with any number of reflexive edges added. Clearly, the total number of red nodes has not changed, so also in $C$, red is the minority color. In this graph, $i$ can have a loop to itself, so it's degree $d_i$ is either $d'_i$ or $d'_i +1$. We do not know whether $i$ itself is red or not, but we do know that its neighbors have the same colors as in $C'$, so $i$ has more than $\frac{d'_i}{2}+\frac{3d'_i}{2d'_i} = \frac{d'_i}{2} + \frac{3}{2}$ red neighbors. 
        If $d_i =d'_i$, then $d_i$ is odd and $i$ has more than $\frac{d_i}{2} + \frac{3}{2} = \frac{d_i +1}{2} +1$ red neighbors while less than half of all nodes in $C$ are red: $i$ is under \texttt{m} illusion. 
        If $d_i = d'_i+1$, then $d_i$ is even and $i$ has more than $\frac{d_i-1}{2} + \frac{3}{2} = \frac{d_i+2}{2}$ red neighbors while less than  half of all nodes in $C$ are red: $i$ is under \texttt{m} illusion. 
    \end{itemize}   
\end{proof}
The same intuition works for \texttt{w} illusions:
\begin{corollary}
    Let $C'=\langle V, E', c\rangle$ be an irreflexive 2-colored graph, let $i$ be a node with degree $d'_i$ and let $q$ be a fraction such that $q = \frac{1}{2} + \frac{1}{d'_i}$ if $d'_i$ is even and $q = \frac{1}{2}+\frac{3}{2d'_i}$ if $d'_i$ is odd.
    If $i$ has more than a $q$-fraction of neighbors of a color that is a tied global winner, or, if $M_V \not = tie$, at least a $q$-fraction of neighbors of the minority color,
    then on any colored graph $C = \langle V, E, c \rangle$ where $E$ is $E'$ with any number of reflexive edges added, $i$ is under \texttt{w} illusion. 
\end{corollary}

\noindent\emph{Proof sketch.}
    Analogously to the proof of Proposition \ref{prop:degree_irreflexive}, but with weak instead of strict inequalities. \\


If nodes can have reflexive edges, there is the possibility that a node sees all nodes in the graph, and therefore cannot be under any kind of illusion. In general, we can define precisely when nodes see too much to be under \texttt{m/w} illusion (Proposition \ref{lem:minority}), and when the fact that a node can be under illusion does not leave enough nodes for the majority color to be majority (Proposition \ref{lem:reflexive_mwmi}). 

\begin{proposition}
\label{lem:minority}
    Let $G = \langle V, E, c\rangle$ be a 2-colored graph where $M_V \not = tie$, and let $N_{min}\subseteq V = \{i\in V: c_i \not = M_V\}$ be the set of nodes in the minority color.
    Then, if for more than half of the nodes $i\in V$, $d_i\geq 2\cdot |N_{min}|+1$, the graph cannot be under \texttt{Mm/Mw/Wm/Ww} illusion. If for at least half of the nodes $i\in V$, $d_i\geq 2\cdot |N_{min}|+1$, the graph cannot be under \texttt{Mm/Mw} illusion. 
\end{proposition}
\begin{proof}
    Any node $i$ with degree $d_i\geq 2\cdot |N_{min}|+1$ has at most $|N_{min}|$ neighbors in the minority color, so at least $|N_{min}|+1$ neighbors in the majority color, and therefore cannot be under \texttt{m/w} illusion.
    If this is the case for more than half of the nodes, the graph cannot be under \texttt{Wm/Ww} illusion. If this is the case for at least half of the nodes, the graph cannot be under \texttt{Mm/Mw} illusion.
\end{proof}

\begin{proposition}\label{lem:reflexive_mwmi}
    On a graph $G=\langle V, E \rangle$ where for at least half of the nodes $i$, $\lceil \frac{d_i+1}{2}\rceil > \lceil\frac{|V|}{2}\rceil$, a \texttt{Mm/Mw} illusion is not possible. 
\end{proposition}
\begin{proof}
    Take any graph $G = \langle V, E, \rangle$ 
    with a subset of the nodes $V'\subseteq V$ such that $|V'|\geq \frac{|V|}{2}$ and for all $i\in |V'|$, 
    $\lceil \frac{d_i+1}{2}\rceil > \lceil\frac{|V|}{2}\rceil$, and suppose for a contradiction that there exists a coloring $c$ of $G$ with which $G$ is in \texttt{Mw} illusion. This means that there is a set of nodes $I\subseteq V$ with $|I| > \frac{|V|}{2}$ such that all the nodes in $I$ are under \texttt{w} illusion. 
    There are two cases:
    \begin{itemize}
        \item $M_V \not = tie$. W.l.o.g. assume that $M_V = red$. Then all the nodes $j\in I$ have $M_{N_j}\in\{blue, tie\}$, so for all  $j\in I$, at least $\lceil \frac{d_j+1}{2}\rceil$ of $j$'s neighbors are blue, so there are at least $\lceil \frac{d_j+1}{2}\rceil$ blue nodes in the graph.   
        Since $|I| > \frac{|V|}{2}$ and $|V'|\geq \frac{|V|}{2}$, there is at least one node $i'\in I\cap V'$, so $\lceil \frac{d_{i'}+1}{2}\rceil > \lceil\frac{|V|}{2}\rceil$ and there are at least $\lceil \frac{d_{i'}+1}{2}\rceil$ blue nodes in the graph.  Therefore, there are at least $ \lceil\frac{|V|}{2}\rceil$ blue nodes in the graph,  which is a contradiction with the assumption that $M_V = red$.
        \item $M_V = tie$. Then all the nodes $i$ under \texttt{w} illusion have $M_{N_i}\in\{blue, red\}$, so for all  $j\in I$, at least $\lfloor \frac{d_j+2}{2}\rfloor$ of $j$'s neighbors are blue or at least $\lfloor \frac{d_j+2}{2}\rfloor$ of $j$'s neighbors are red.
        Again,  since $|I| > \frac{|V|}{2}$ and $|V'|\geq \frac{|V|}{2}$, there is at least one node $i'\in I\cap V'$. For this node $i'$, $\lceil \frac{d_{i'}+1}{2}\rceil > \lceil\frac{|V|}{2}\rceil$ and there are either at least $\lfloor \frac{d_{i'}+2}{2}\rfloor$ blue nodes in the graph or at least $\lfloor \frac{d_{i'}+2}{2}\rfloor$ red nodes in the graph.
        Since $\lfloor \frac{d_{i'}+2}{2}\rfloor \geq \lceil \frac{d_{i'}+1}{2}\rceil $, either the total numer of red nodes in the graph or the total number of blue nodes in the graph is greater than $\lceil\frac{|V|}{2}\rceil$, which is a contradiction with the assumption that $M_V = tie$. 
    \end{itemize}
\end{proof}
   
Even though Proposition \ref{lem:reflexive_mwmi} seems rather general, it is only applicable to some specific cases.    Let us consider what $\lceil \frac{d_i+1}{2}\rceil > \lceil\frac{|V|}{2}\rceil$ means for the possible parities of $d_i$ and $|V|$.
    \begin{itemize}
        \item $d_i$ is even and $|V|$ is even: $\lceil \frac{d_i+1}{2}\rceil > \lceil\frac{|V|}{2}\rceil$ becomes $\frac{d_i+2}{2} > \frac{|V|}{2}$, so $d_i+2>|V|$. Since $d_i$ and $|V|$ are both even, this can only happen if $d_i = |V|$.
        \item $d_i$ is even and $|V|$ is odd: $\lceil \frac{d_i+1}{2}\rceil > \lceil\frac{|V|}{2}\rceil$ becomes $\frac{d_i+2}{2} > \frac{|V|+1}{2}$, so $d_i+1>|V|$, which implies  $d_i = |V|$. This is not possible, since one of them is odd and one is even.
        \item $d_i$ is odd and $|V|$ is even: $\lceil \frac{d_i+1}{2}\rceil > \lceil\frac{|V|}{2}\rceil$ becomes $\frac{d_i+1}{2} > \frac{|V|}{2}$, so $d_i+1>|V|$, which implies  $d_i = |V|$. This is not possible, since one of them is odd and one is even.
        \item $d_i$ is odd and $|V|$ is odd: $\lceil \frac{d_i+1}{2}\rceil > \lceil\frac{|V|}{2}\rceil$ becomes $\frac{d_i+1}{2} > \frac{|V|+1}{2}$, so $d_i>|V|$, which is clearly impossible.
    \end{itemize}
    Hence, the only possibility for $\lceil \frac{d_i+1}{2}\rceil > \lceil\frac{|V|}{2}\rceil$ is that both $d_i$ and $|V|$ are even, and they are equal: $i$ is connected to \emph{all} nodes. 
    This is also true in the other direction, if for more than half of the nodes $i$,  $d_i = |V|$ and they are both even, a \texttt{Mm/Mw} illusion is not possible. On this idea and the idea expressed in Proposition \ref{lem:minority} is the following proposition based, however, it is slightly more general (not just about even degrees and not only about cases where there is no tie). 
    We can prove the proposition in an easier way without both previous propositions:
\begin{proposition}\label{prop:refl_halfnodes}
On a 
graph 
$G = \langle V, E \rangle$ 
where at least half of the nodes are connected to all nodes, a \texttt{Mm/Mw} illusion is not possible.
\end{proposition}
\begin{proof}
 Take a 
 graph $G$ where at least half of the nodes are connected to all nodes. Any node $i$ that is connected to all nodes cannot be under \texttt{m/w} illusion, since it sees exactly the correct proportion of red and blue nodes, so $G$ cannot be in \texttt{Mm/Mw} illusion.
\end{proof}

We observe something more specific about the edge cases where exactly half or just less than half nodes are connected to all nodes:
\begin{proposition}\label{prop:exact_half_nodes}
    If exactly $\frac{|V|}{2}$ nodes are connected to all nodes ($|V|$ is even), a \texttt{Wm} illusion is possible, and if exactly $\frac{|V|-1}{2}$ nodes are connected to all nodes ($|V|$ is odd), a \texttt{Mw} illusion is possible. 
\end{proposition}
\begin{proof}
    In both cases, color all the nodes that are connected to all nodes blue and all other nodes red: then all red nodes are connected to all blue nodes but not to all red nodes, so more than / at least half of their neighbors are blue: they are under illusion.
\end{proof}
\begin{remark}
    Proposition \ref{prop:exact_half_nodes} might seem to contradict Proposition \ref{prop:refl_halfnodes} but it does not: in the even case, the illusion here is a \texttt{Wm} illusion, while Proposition \ref{prop:refl_halfnodes} is about \texttt{Mw} illusions. In the odd case, the number of nodes here is just below half, while in Proposition \ref{prop:refl_halfnodes} it is just more than half.
\end{remark}

As mentioned in 
Section \ref{sec:reflexive}, some results on irreflexive graphs still hold on graphs that can have reflexive edges, but some do not.
In Table \ref{tab:results_reflexive}, a summary is given of the classes of graphs on which majority illusions are possible and on which they are not.
  \begin{table}[h]
\caption{The (im)-possibility of majority illusions (using two colors) on different classes of graphs. \checkmark~ indicates that the illusion is possible on all graphs of the class, \xmark~  indicates that the illusion is not possible on any graph in the class, \checkmark / \xmark~ indicates that the illusion is possible on some but not all graphs of the class. 
References to the relevant results are given. We do not assume that graphs are irreflexive here.
Note that for irreflexive graphs with only odd-degree nodes and 2-colorable graphs the \texttt{Mw} illusion is always either a \texttt{Mm} illusion or \texttt{Uw} illusion, conform 
Proposition \ref{thm:degree-odd} and Lemma \ref{lem:2-col}, and that for irreflexive complete graphs with an even number of nodes $|V|$, the \texttt{Mw} illusion is always a \texttt{Uw} illusion, conform Proposition \ref{prop:complete-maj-weak-maj}. 
}
   \centering
    \begin{tabular}{|p{0.21\linewidth}|p{0.21\linewidth}| p{0.15\linewidth}|p{0.12\linewidth}|}
    \hline
        \multirow{1}{*}{\textbf{Class of graphs}}  &  \multicolumn{1}{c|}{\multirow{1}{*}{\textbf{\texttt{Mw} illusion}}} & \textbf{\texttt{Wm} illusion}& \textbf{\texttt{Mm} illusion}  \\ \hline
         Irreflexive graphs & \multicolumn{1}{c|}{\multirow[c]{11}{*}{\checkmark (Thm. \ref{thm:maj-weak-maj})}} & \multicolumn{2}{c|}{\checkmark / \xmark}      \\  \cline{1-1} \cline{3-4}
         2-colorable graphs with $|V|$ odd  & & \multicolumn{2}{c|}{\multirow{2}{*}{\checkmark~ (Prop. \ref{prop:2-col-maj-maj})}} \\ \cline{1-1}  \cline{3-4}
         2-colorable graphs with $i\in V: \forall j\in N_i: d_j>2$  & & \multicolumn{1}{c|}{\multirow{3}{*}{\checkmark~ (Prop. \ref{prop:2col-strict})}}& \multicolumn{1}{c|}{\multirow{3}{*}{\checkmark / \xmark}}\\ \cline{1-1} \cline{3-4}
         Irr. 2-regular graphs  & & \multicolumn{2}{c|}{\xmark~ (Prop. \ref{prop:k=2})}\\ \cline{1-1} \cline{3-4}
         Irr. complete graphs   &  & \multicolumn{2}{c|}{\xmark~ (Prop. \ref{prop:complete-weak-maj-maj})} \\     \hline
           2-regular graphs  & \multicolumn{1}{c|}{\checkmark/\xmark}
         & \multicolumn{2}{c|}{\xmark~ (Prop. \ref{prop:k=2})}\\ \hline
           complete graphs  & \multicolumn{3}{c|}{\xmark}\\ \hline
         Graphs with at least $\frac{|V|}{2}$ nodes connected to all nodes  & \multicolumn{1}{c|}{\multirow{3}{*}{\xmark (Prop. \ref{prop:refl_halfnodes})}} & \multirow{3}{*}{\checkmark/\xmark} & \multirow{3}{*}{\xmark (Prop. \ref{prop:refl_halfnodes})} \\ \hline
         Graphs with exactly $\frac{|V|}{2}$ nodes connected to all nodes  & \multicolumn{1}{c|}{\multirow{3}{*}{\checkmark/\xmark}} & \multirow{3}{*}{\checkmark (Prop. \ref{prop:exact_half_nodes})} &\multirow{3}{*}{\checkmark/\xmark} \\ \hline
         Graphs with exactly $\frac{|V|-1}{2}$ nodes connected to all nodes  & \multicolumn{1}{c|}{\multirow{3}{*}{\checkmark (Prop. \ref{prop:exact_half_nodes}) }} & \multicolumn{2}{c|}{\multirow{3}{*}{\checkmark/\xmark}} \\ \hline
    \end{tabular}
    \label{tab:results_reflexive}
\end{table}

\section{Additional Plots of the Experiments in Section \ref{sec:simulations}}\label{app:experiment_extras}

This appendix gives some extra plots on the experiments in Section \ref{sec:simulations}. The full data can be found on \href{https://github.com/MaaikeLos/Majority_Illusions}{https://github.com/MaaikeLos/Majority\_Illusions}, as well as the R code to generate the plots in this section and more plots to study the effect of the different parameters in detail.

First, we give the full Pearson correlation matrices for Erd\H{o}s-R\'{e}nyi graphs, Holme-Kim graphs, and the Facebook network in respectively Figure \ref{fig:cormatER}, Figure \ref{fig:cormatHK}, and Figure \ref{fig:cormatFB}.

Figures \ref{fig:median_frac_illusion_ER} and \ref{fig:median_frac_illusion_HK} show the median fraction of nodes under illusion (with the interquartile range) in Erd\H{o}s-R\'{e}nyi (respectively Holme-Kim) graphs with different input parameters, which allows us to see the interaction effects between the different parameters.
Figure \ref{fig:MSE_HK} shows the mean squared error of nodes in Holme-Kim graphs with different input parameters. Similar plots for Erd\H{o}s-R\'{e}nyi graphs and for the fraction of networks under \texttt{Mm/Mw/Wm/Ww} illusion can be obtained using the R code on Github.

Figure \ref{fig:deg_assort} (a) shows the median fraction of nodes under illusion (and the interquartile range) versus \emph{degree assortativity} in Erd\H{o}s-R\'{e}nyi graphs. Figure \ref{fig:deg_assort} (b) shows the mean squared error of nodes versus degree assortativity in Erd\H{o}s-R\'{e}nyi graphs.


Figure \ref{fig:homophily} (a) 
shows the median fraction of nodes under illusion (and the interquartile range) versus \emph{homophily} in 
Holme-Kim graphs. Figure \ref{fig:homophily} 
(b) shows the mean squared error of nodes versus the homophily in 
Holme-Kim
graphs
.


\begin{figure}
    \centering
    \includegraphics[width=\linewidth]{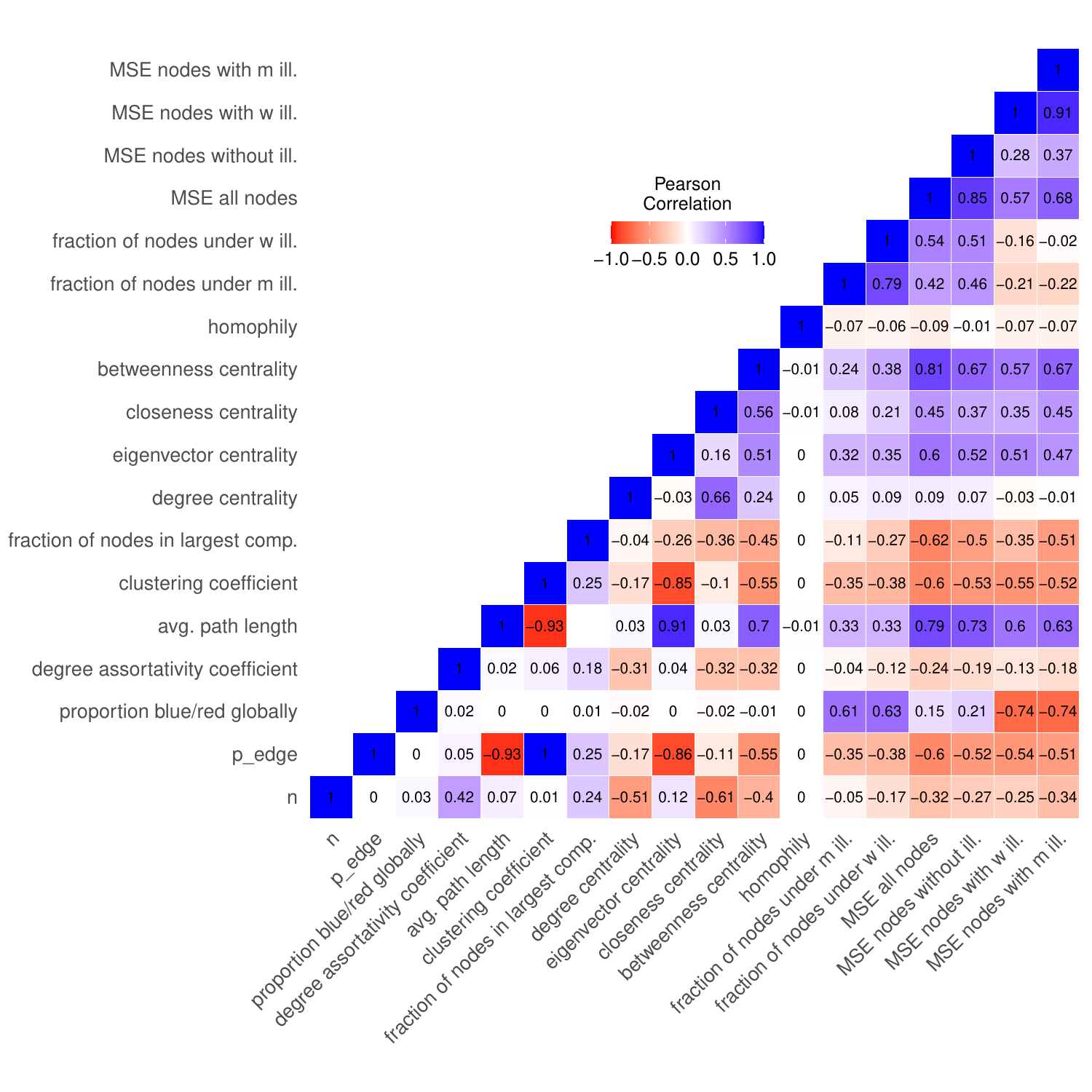}
    \caption{Correlation matrix for Erd\H{o}s-R\'{e}nyi graphs. Blue shades correspond to positive correlation, red shades to negative correlation, the intensity of the color indicates the strength of the correlation.}
    \label{fig:cormatER}
\end{figure}
\begin{figure}
    \centering
    \includegraphics[width=\linewidth]{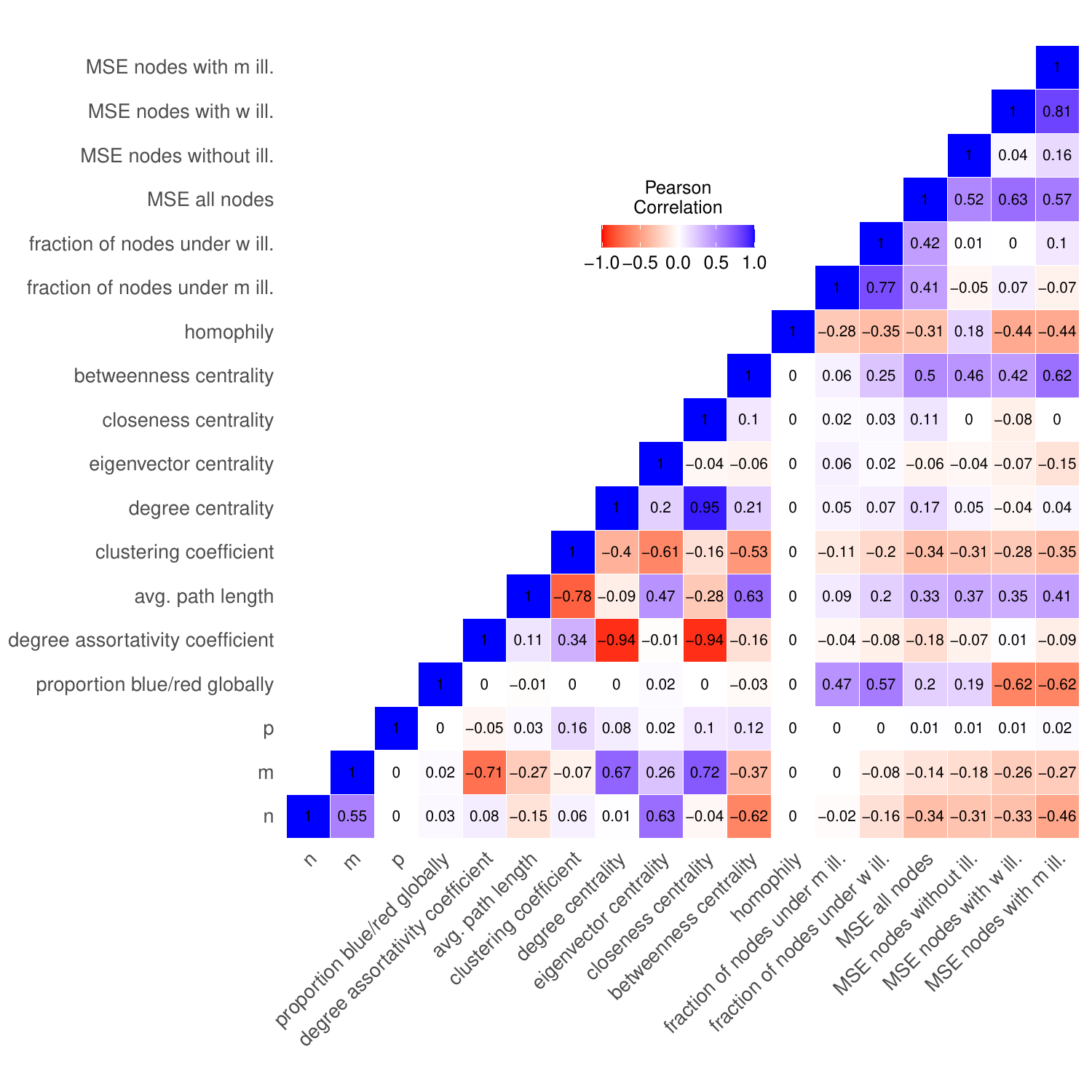}
    \caption{Correlation matrix for Holme-Kim graphs. Blue shades correspond to positive correlation, red shades to negative correlation, the intensity of the color indicates the strength of the correlation.}
    \label{fig:cormatHK}
\end{figure}
\begin{figure}
    \centering
    \includegraphics[width=0.65\linewidth]{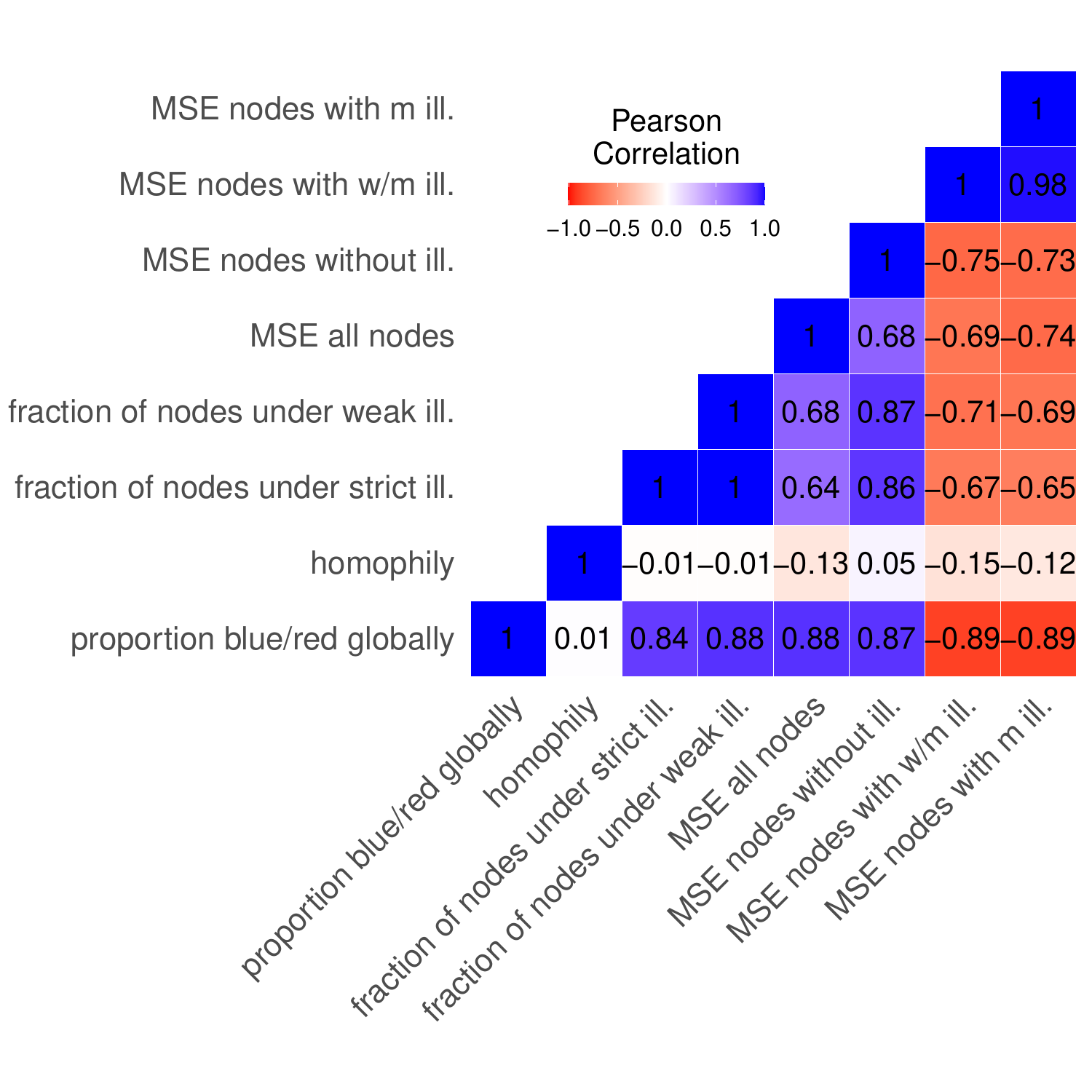}
    \caption{Correlation matrix for the Facebook network. Blue shades correspond to positive correlation, red shades to negative correlation, the intensity of the color indicates the strength of the correlation.}
    \label{fig:cormatFB}
\end{figure}

\begin{figure}
    \centering
    \includegraphics[width=\linewidth]{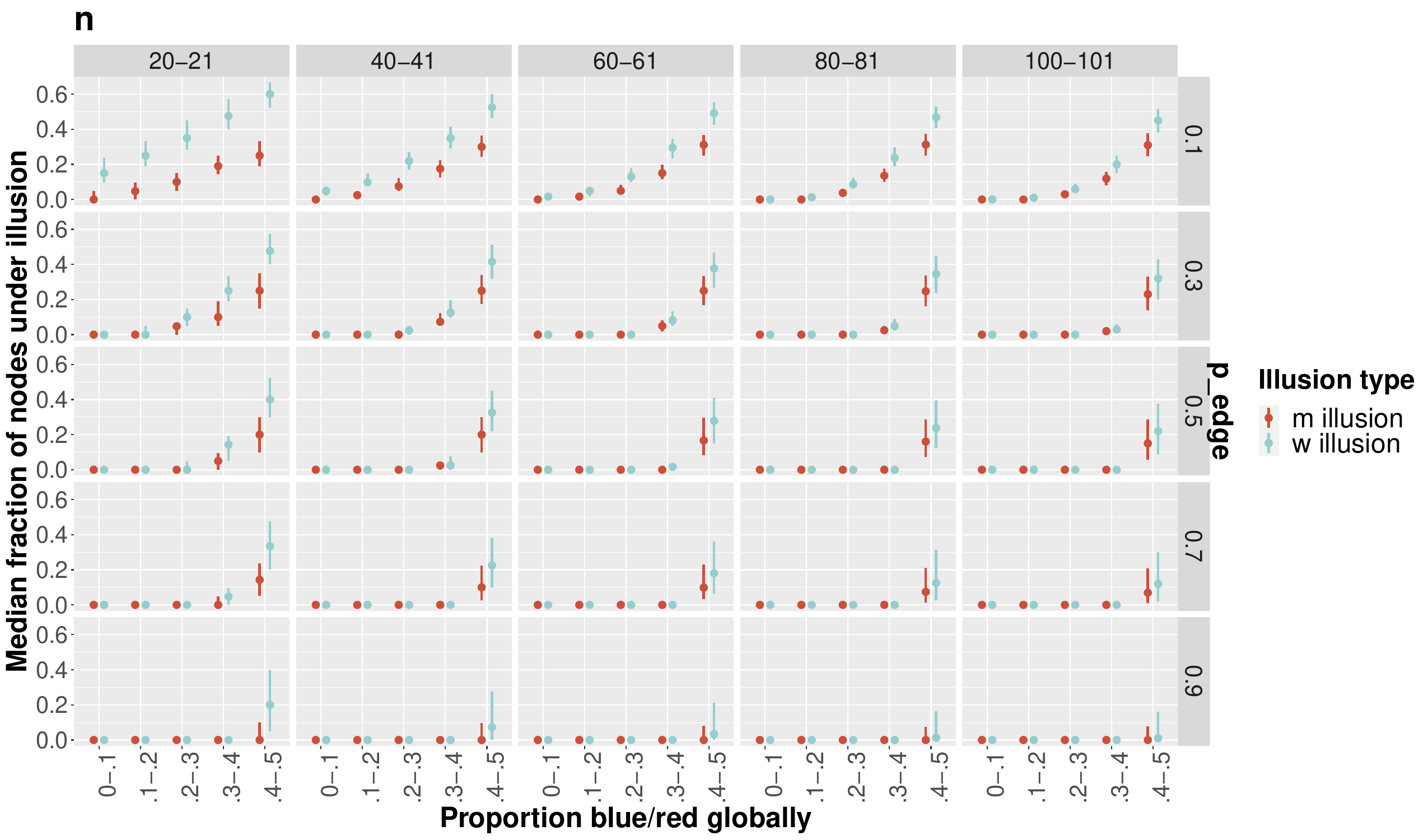}
    \caption{Median fraction of nodes under illusion (and interquartile range) in Erd\H{o}s-R\'enyi graphs with different parameters}
    \label{fig:median_frac_illusion_ER}
\end{figure}
\begin{figure}
    \centering
    \includegraphics[width=\linewidth]{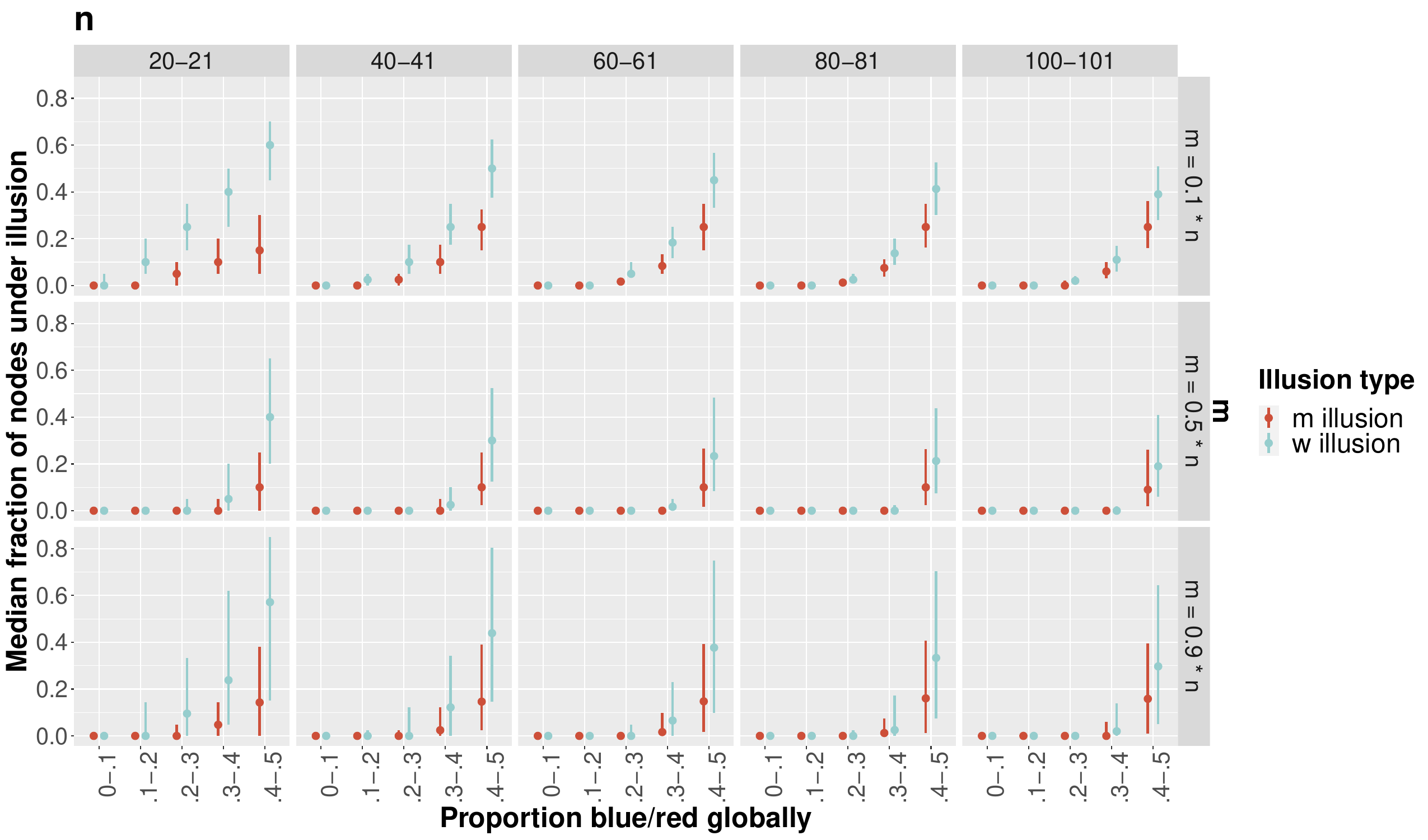}
    \caption{Median fraction of nodes under illusion (and interquartile range) in Holme-Kim graphs with different parameters}
    \label{fig:median_frac_illusion_HK}
\end{figure}

\begin{figure}
    \centering
    \includegraphics[width=\linewidth]{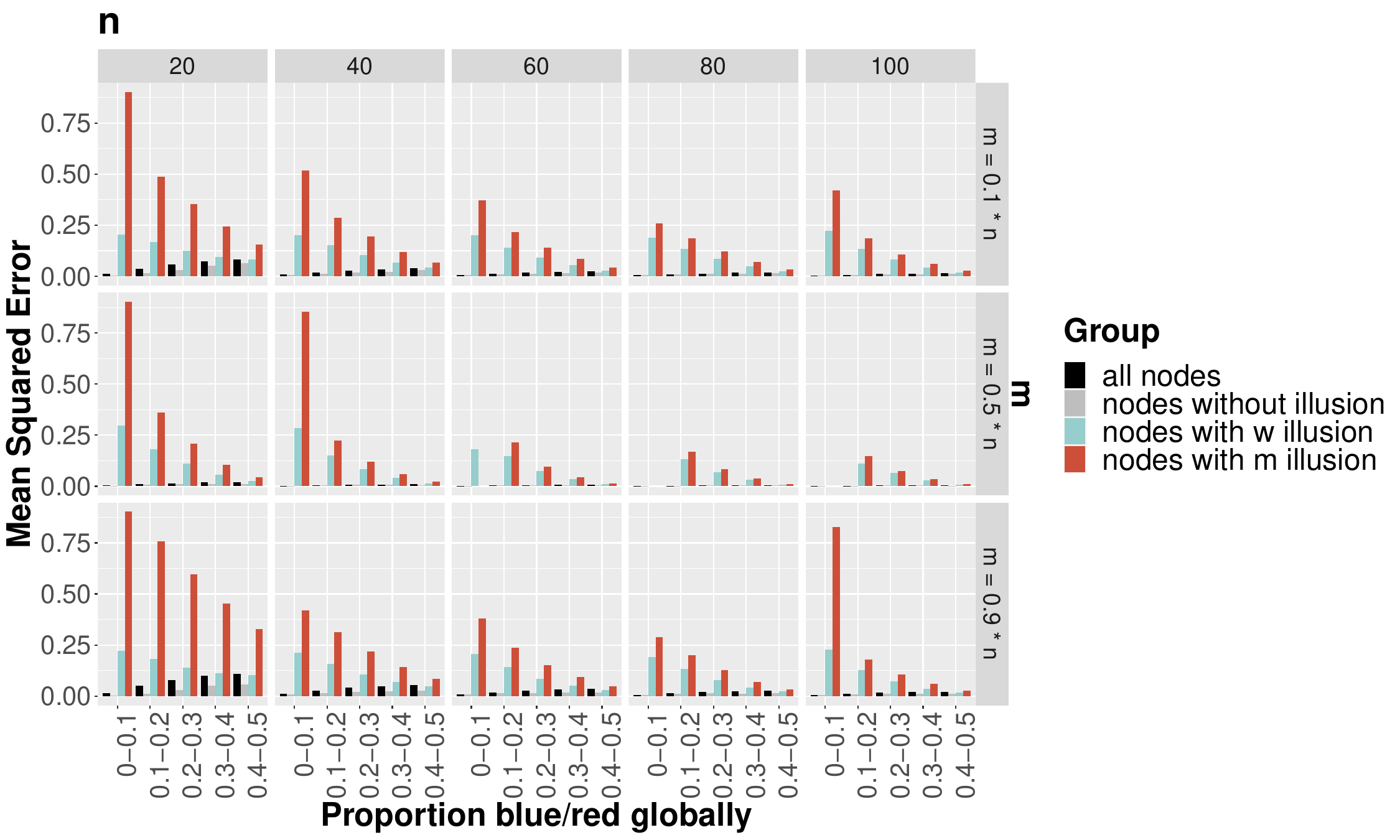}
    \caption{Mean squared error in Holme-Kim graphs with different parameters}
    \label{fig:MSE_HK}
\end{figure}

\begin{figure}
    \centering
    \begin{subfigure}{\textwidth}
    \centering\captionsetup{width=\linewidth}
        \includegraphics[width=\textwidth]{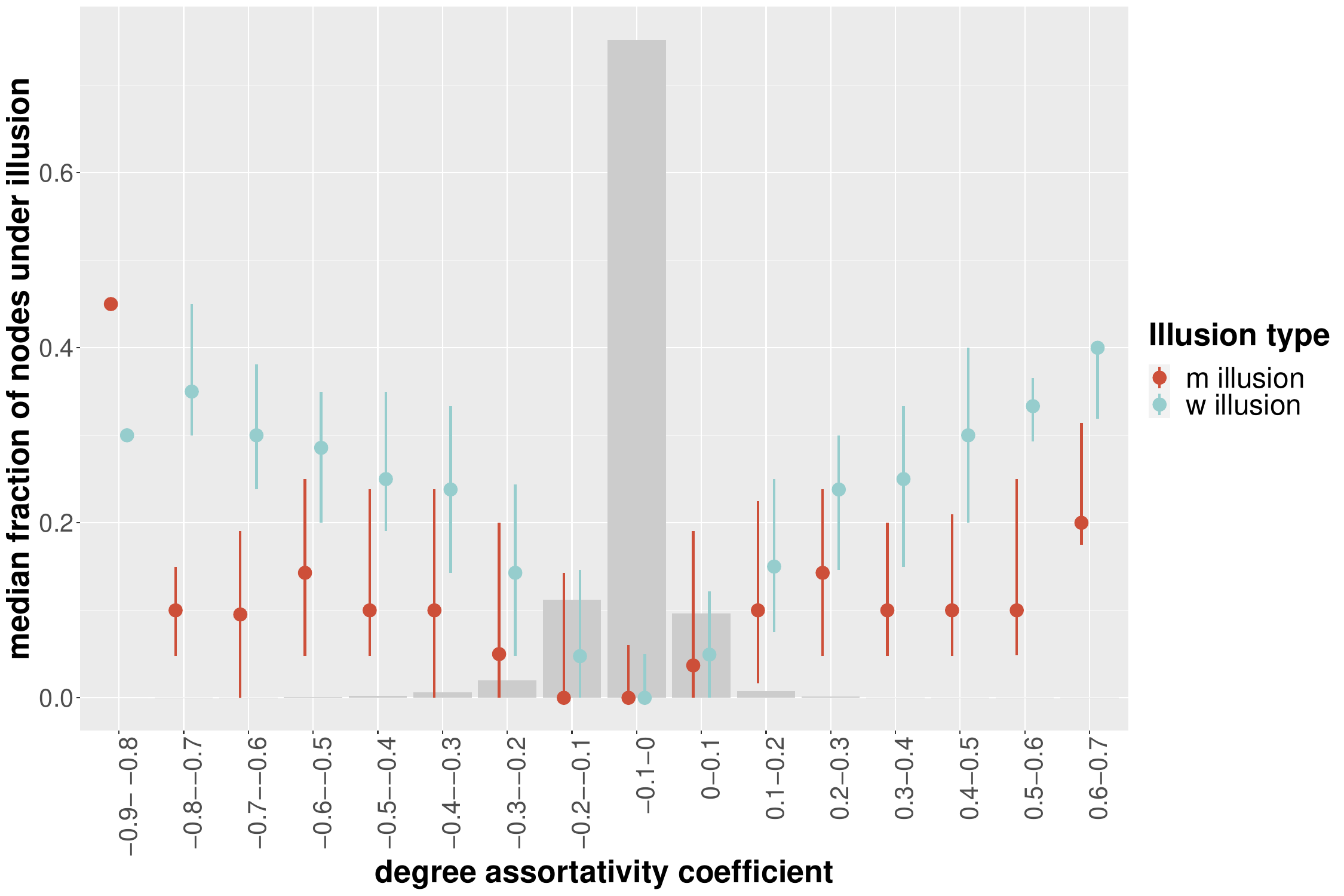}
        \caption{Median fraction of nodes under illusion versus degree assortativity coefficient.}
    \end{subfigure}\\
    \begin{subfigure}{\textwidth}
    \centering\captionsetup{width=\linewidth}
    \includegraphics[width=\textwidth]{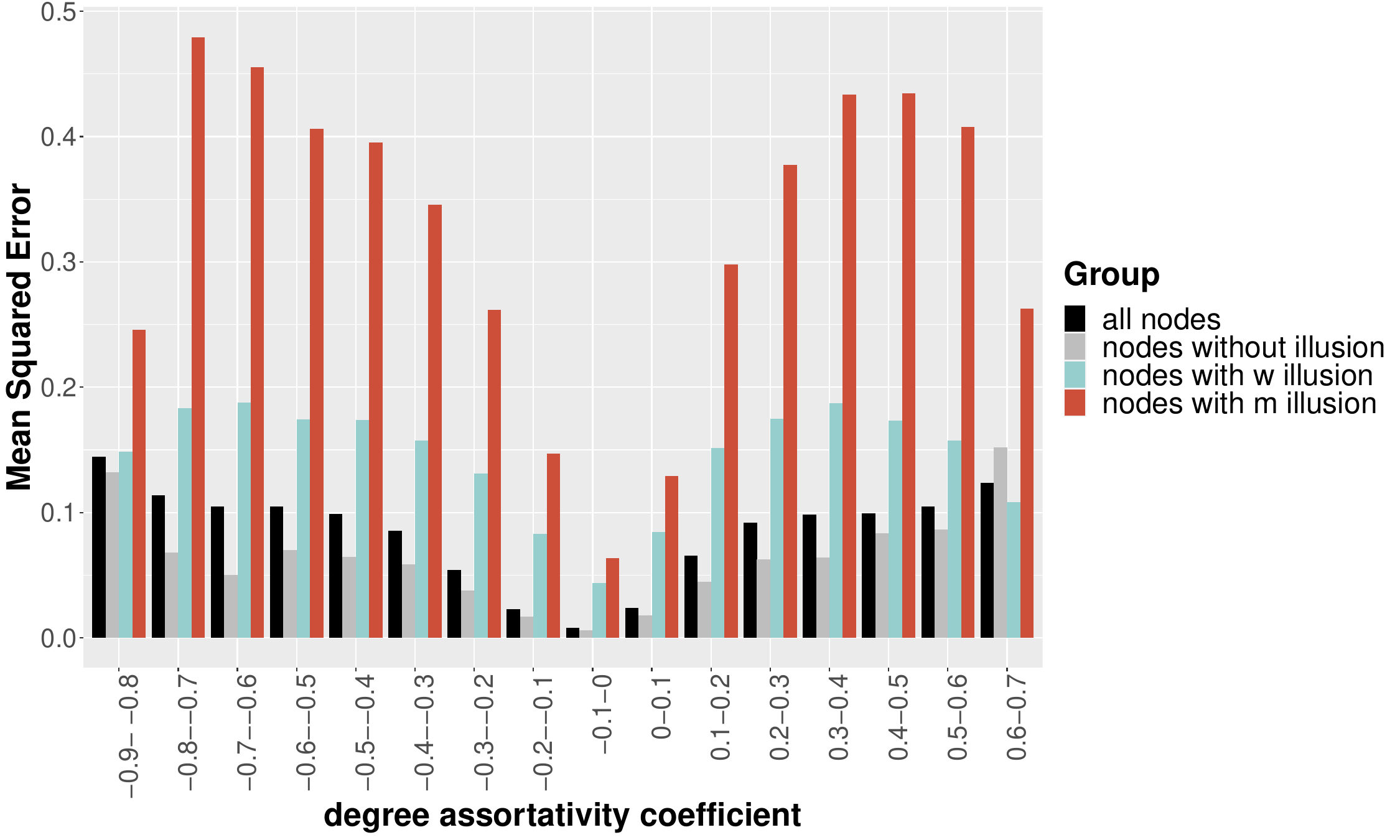}
        \caption{Mean squared error versus degree assortativity coefficient.}
    \end{subfigure}%
\caption{Median fraction of nodes under illusion / mean squared error versus \textbf{degree assortativity} in  Erd\H{o}s-R\'{e}nyi graphs. In (a) in grey bars in the background the distribution of the datapoints.}
\label{fig:deg_assort}
\end{figure}

\begin{figure}
    \centering
    \begin{subfigure}{\textwidth}
    \centering\captionsetup{width=\linewidth}
        \includegraphics[width=\textwidth]{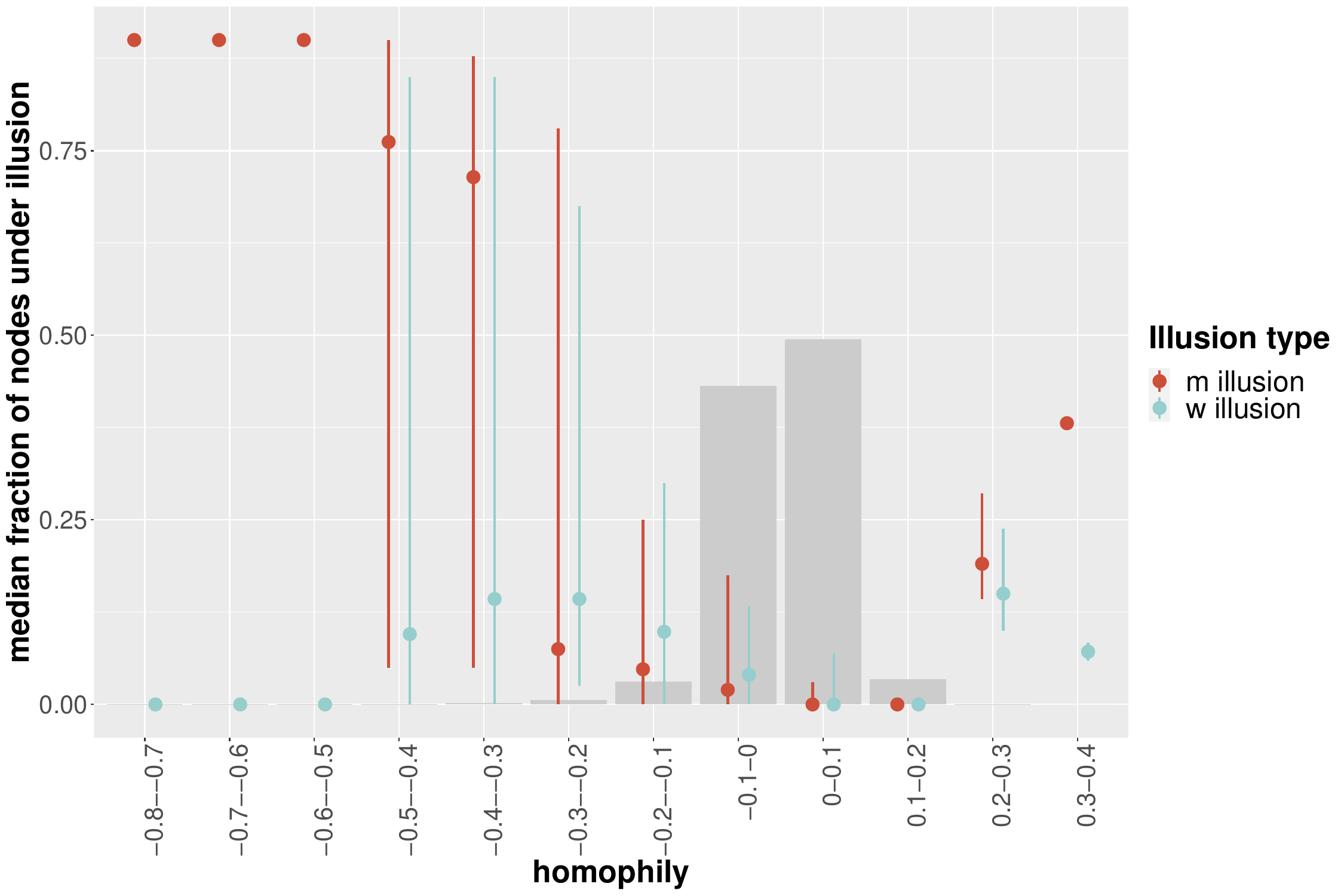}
        \caption{Median fraction of nodes under illusion versus homophily.}
    \end{subfigure}\\
    \begin{subfigure}{\textwidth}
    \centering\captionsetup{width=\linewidth}
    \includegraphics[width=\textwidth]{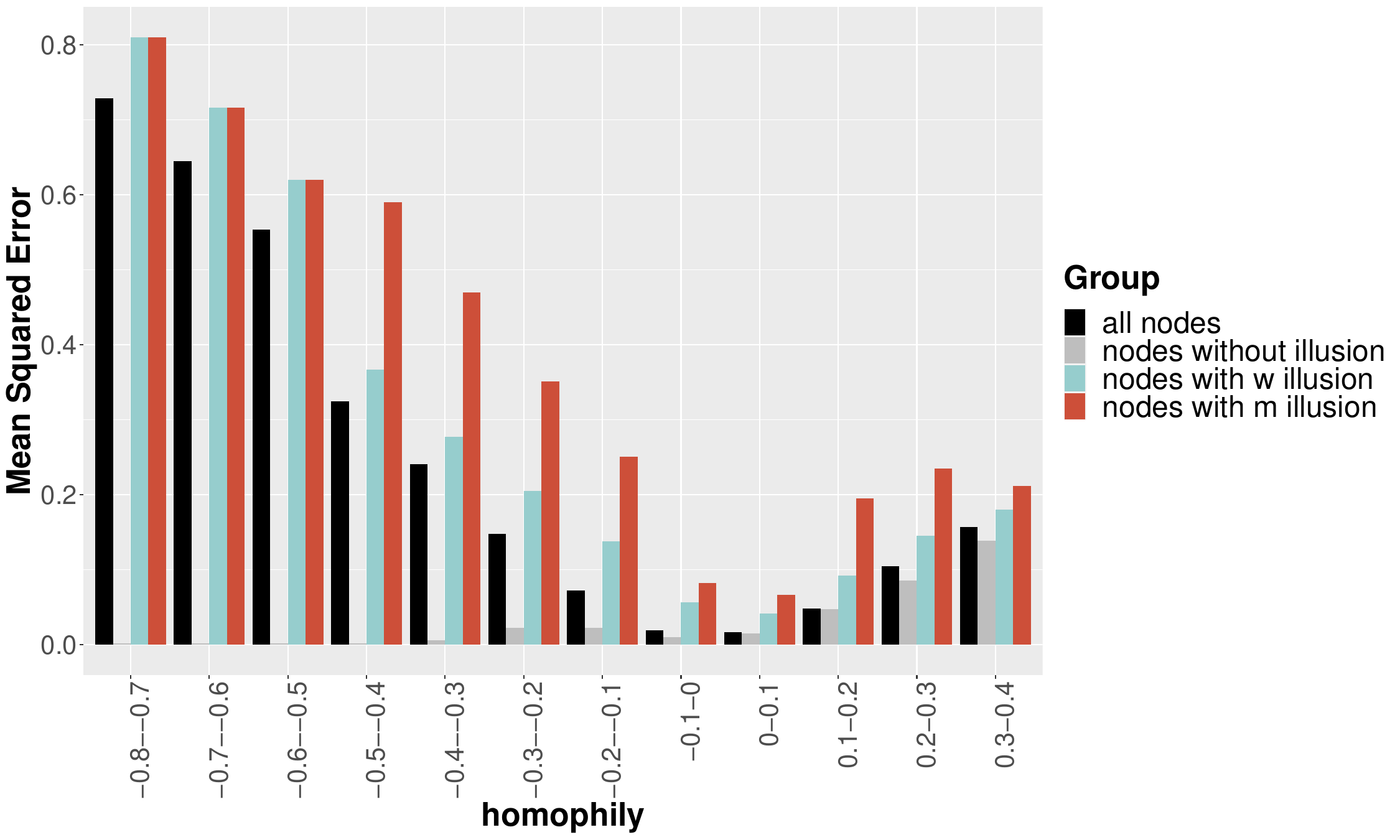}
        \caption{Mean squared error versus homophily.}
    \end{subfigure}
\caption{Median fraction of nodes under illusion / mean squared error versus  \textbf{homophily} in Holme-Kim graphs. In (a) in the background in grey bars the distribution of the datapoints.}
\label{fig:homophily}
\end{figure}

\end{appendices}

\clearpage
\bibliography{bibliography}
\end{document}